\documentclass[12pt]{article}

\usepackage{amsmath}
\usepackage{amssymb}
\usepackage{mathrsfs}
\usepackage{booktabs}
\usepackage{longtable}
\usepackage{afterpage}
\usepackage{sidecap}

\usepackage{amsfonts}
\usepackage{fullpage}
\usepackage[pdftex]{graphicx}
\usepackage[usenames,dvipsnames]{color}
\usepackage{epstopdf}
\usepackage{epsfig}
\usepackage{subfig}
\usepackage{latexsym}
\usepackage{url}

\usepackage[justification=raggedright,singlelinecheck=false]{caption}

\newtheorem{thm}{Theorem}

\newtheorem{duplicate}{Theorem}
\newtheorem{rmk}{Remark}
\newenvironment{proof}[1][Proof]
              {\par \normalfont
              \trivlist
             \item[
               \hspace{12pt}               \itshape #1{.}]\ignorespaces
        }{\hfill$\Box$ \endtrivlist}

\usepackage{bbm}
\newcommand{\indic}{\mathbbm{1}}

\DeclareMathOperator{\bbE}{\mathbb{E}}

\DeclareMathOperator{\sign}{\mathrm{sign}}
\newcommand{\e}{{\mathrm e}}
\newcommand{\ds}{\displaystyle}
\newcommand{\gtpdf}{{w_G}}
\newcommand{\gtrv}{{W_G}}
\newcommand{\gtrvt}{{\widetilde{W}_G}}
\newcommand{\gtrvz}[1]{{\Psi_{#1}}}
\newcommand{\gtrvzt}[1]{{\Phi_{#1}}}
\newcommand{\gtrvnt}[1]{{\widetilde{W}_G^{(#1)}}}
\newcommand{\gtrvznt}[2]{{\Phi_{#1}^{(#2)}}}
\newcommand{\gtpdfij}{{w_G^{(ii')}}}
\newcommand{\ev}{\mathcal{D}}
\newcommand{\pp}[1]{\left({#1}\right)} 
\renewcommand{\b}[1]{\left[{#1}\right]} 
\newcommand{\thp}[1]{\textit{(#1)}} 
\newcommand{\fpc}[1]{#1)} 
\newcommand{\fpt}[1]{#1} 

\begin{document}


\title{Reproduction numbers for epidemic models with households and other social structures II: comparisons and implications for vaccination}

\author{Frank Ball$^{1}$, Lorenzo Pellis$^{2}$ and Pieter Trapman$^{3}$}
\footnotetext[1]{School of Mathematical Sciences, University of Nottingham,
University Park, Nottingham, NG7 2RD, UK,}
\footnotetext[2]{Warwick Infectious Disease Epidemiology Research centre (WIDER) and Warwick
Mathematics Institute, University of Warwick, Coventry, CV4 7AL, UK,}
\footnotetext[3]{Department of Mathematics, Stockholm University, Stockholm 106 91, Sweden.}
\date{\today}
\maketitle

\begin{abstract}
In this paper we consider epidemic models of directly transmissible SIR (susceptible $\to$ infective $\to$ recovered) and SEIR (with an additional latent class) infections in fully-susceptible populations with a social structure, consisting either of households or of households and workplaces. We review most reproduction numbers defined in the literature for these models, including the basic reproduction number $R_0$ introduced in the companion paper of this, for which we provide a simpler, more elegant derivation. Extending previous work, we provide a complete overview of the inequalities among these reproduction numbers and resolve some open questions.
Special focus is put on the exponential-growth-associated reproduction number $R_r$, which is loosely defined as the estimate of $R_0$ based on the observed exponential growth of an emerging epidemic obtained when the social structure is ignored. We show that for the vast majority of the models considered in the literature $R_r \geq R_0$ when $R_0 \ge 1$ and $R_r \leq R_0$ when $R_0 \le 1$.
We show that, in contrast to models without social structure, vaccination of a fraction $1-1/R_0$ of the population, chosen uniformly at random, with a perfect vaccine is usually insufficient to prevent large epidemics. In addition, we provide significantly sharper bounds than the existing ones for bracketing the critical vaccination coverage between two analytically tractable quantities, which we illustrate by means of extensive numerical examples.
\end{abstract}



\section{Introduction}
\label{sec:Introduction}

The basic reproduction number $R_0$ is arguably the most important epidemiological parameter because of its clear biological interpretation and its properties: in the simplest epidemic models, where individuals are all identical, mix homogeneously, the population is large and the initial number of infectives is small, (i) a large epidemic is possible if and only if $R_0>1$ (threshold property), (ii) when $R_0 > 1$, vaccinating  a fraction $1-1/R_0$ of individuals chosen uniformly at random -- or, equivalently, isolating the same fraction of infected individuals before they have the chance to transmit further -- is sufficient to prevent a large outbreak (critical vaccination coverage) and (iii) the fraction of the population infected by a large epidemic depends only on $R_0$. The definition of $R_0$ is straightforward in single-type homogeneously mixing models and has been successfully extended to multitype models (see Diekmann et al.~\cite{DiekEtal2013}, Chapter 7).

In our earlier paper, we showed how to extend the definition of $R_0$ to many models with a social structure, namely the households models and certain types of network-households and households-workplaces models (Pellis et al.~\cite{PelBalTra2012}). The extension proposed there aims at preserving both the biological interpretation of $R_0$ as the average number of cases a typical individual generates early on in the epidemic and its threshold property. However, already in the case of multitype populations the simple relationship between  $R_0$ and the epidemic final size no longer holds. In this paper we show that, for models involving mixing in small groups, also the simple relationship between $R_0$ and the critical vaccination coverage breaks down. In particular, we find that vaccinating a fraction $1-1/R_0$ of the population is generally insufficient to prevent a major outbreak. This result stems from a series of inequalities which extend the work done by
Goldstein et al.~\cite{GoldsteinEtal2009}, and leads to sharper bounds for the critical vaccination coverage than previously available.

The definition of  $R_0$ given in \cite{PelBalTra2012} may be described briefly for an SIR (susceptible $\to$ infective $\to$ recovered) epidemic in a closed population as follows.  Consider the epidemic graph (see \cite{PelBalTra2012}, Section 1, and Section~\ref{Hmodel} of this paper), in which vertices correspond to individuals in the population and for any ordered pair of distinct individuals, $(i,i')$ say, there is a directed edge from $i$ to $i'$ if and only if $i$, if infected, makes at least one infectious contact with $i'$ (see Figure \ref{epigraphfig}).  Suppose that initially there is one infective and the remainder of the population is susceptible.  The initial infective is said to belong to generation $0$ (say, individual 0 in Figure \ref{epigraphfig}).  Any other individual, $i$ say, becomes infected if and only if in the epidemic graph there is a chain of directed edges from the initial infective to individual $i$, and in that case the generation of $i$ is defined to be the number of edges in the shortest such chain.  Thus, generation 1 consists of those individuals with whom the initial infective has at least one infectious contact (individuals 1 and 2 in Figure \ref{epigraphfig}), generation 2 consists of those individuals that are contacted by at least one generation-$1$ infective but not by the initial infective (individuals 4 and 5 in Figure \ref{epigraphfig})
and so on.  For $k=0,1,\cdots$, let $X_k^{(N)}$ denote the the number of generation-$k$ infectives, where $N$ denotes the population size. Thus, in  Figure \ref{epigraphfig}, $X_0^{(6)} =1$, $X_1^{(6)} =2$, $X_2^{(6)}=2$, $X_3^{(6)}=1$ and $X_k^{(6)}=0$ for $k \ge 4$. Then $R_0$ is defined by
\begin{equation}
\label{R0defn}
R_0=\lim_{k \to \infty}\lim_{N \to \infty} \left(\bbE\left[X_k^{(N)}\right]\right)^{1/k},
\end{equation}
i.e.~by the limit, as the population size tends to infinity, of the asymptotic geometric growth rate of the mean generation size \cite{PelBalTra2012}.

\begin{figure}[!ht] \begin{center}
\includegraphics[width=.3\textwidth]{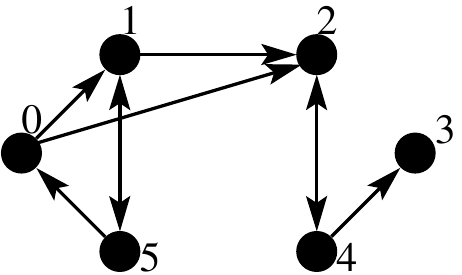}
 \caption{Example of epidemic graph in a population of size $N=6$.}
\label{epigraphfig}\end{center}
\end{figure}

For single- and multi-type unstructured populations the value of $R_0$ obtained using~\eqref{R0defn} coincides with that obtained using the usual definition as ``the expected number of secondary cases produced by a typical infected individual during its entire infectious period in a population consisting of susceptibles only'' (see Heesterbeek and Dietz~\cite{HeestDietz1996}).  (Note that, for fixed $k$, as $N\to\infty$ the epidemic process converges to a Galton-Watson branching process, i.e.\ we consider a linear approximation of the early phase of the epidemic.)
However, unlike the usual definition of $R_0$, definition~\eqref{R0defn} extends naturally to models with small mixing groups, such as the households and households-workplaces models.  In Pellis et al.~\cite{PelBalTra2012}, $R_0$ for these two models was obtained by exploiting difference equations describing variables related to the mean generation sizes.  In the present paper, we show that $R_0$ for these models may be obtained more easily from the discrete-time Lotka-Euler equation (cf.\ Equation \eqref{g0H}) that describes the asymptotic (Malthusian) geometric growth rate of the mean population size of an associated branching process, which approximates the early phase of the epidemic.


Note that the construction of the epidemic graph, and therefore most of the work of \cite{PelBalTra2012} and of this paper is based on the assumption that the behaviour of any infected individual can be decided before the epidemic starts. This is a common assumption in epidemic modelling, but it is quite a restrictive one. As noted by Pellis et al.~\cite{PelFerFra2008}, this condition is violated when the infectious behaviour of an individual depends on the time when he/she is infected (for example, if the number of other infectives at the time of infection matters or if a control policy is implemented at a certain time) and, in multi-type populations, on the type of the infector.
Theoretically, \eqref{R0defn} and all results in this paper require only that the epidemic admits a description in terms of generations of infection, which seems biologically plausible for most epidemic models. However, analytical progress is limited without invoking the assumption above.

In Section \ref{sec:house} we study reproduction numbers for the households model in great detail: in Sections \ref{Hmodel} and \ref{subsec:R0}, we introduce the households model and provide a simpler, more elegant derivation of the basic reproduction number $R_0$ than that presented in Pellis et al.~\cite{PelBalTra2012}; we then review the vast majority of the reproduction numbers defined in the literature for the households model in the remainder of Section \ref{sec:house} and we formulate our main results in Theorems \ref{Hcomp} and \ref{propos} in Section \ref{sec:HComparisons}, where virtually all comparisons are carefully examined and new, sharper bounds on the critical vaccination coverage are obtained.
For ease of reference, Table \ref{tab:RsH} collects all the households reproduction numbers with a reference to where they are discussed, and Table \ref{tab:Inequalities} 
summarises known and novel results, again with appropriate references. 
In Sections \ref{sec:housework} and \ref{sec:HWComparisons} we define and compare reproduction numbers for models with households and workspaces. Here we again provide a new and simpler derivation of $R_0$ than in~\cite{PelBalTra2012}.
Reproduction numbers are collected in Table \ref{tab:RsHW} and the inequalities among them are reported in Theorem \ref{Theorem 5.2} and in the extension of Theorem \ref{propos} to the households-workplaces model.
Extensive numerical illustrations are presented in Section \ref{sec:numerical}, while in Section  \ref{sec:proofs} we provide the proofs of the comparisons presented in Sections \ref{sec:HComparisons} and \ref{sec:HWComparisons}.  Section \ref{sec:conclusions} is devoted to comments and conclusions. We summarise the main notation used in the paper in Table \ref{tab:symbf}.


\LTcapwidth=\textwidth
\begin{longtable}{@{}cp{.7\textwidth}c@{}} 
\caption{Reproduction numbers for the \textbf{households} model (analogues of $R_0$, $R_\ast$, $R_I$, $R_V$, $R_{H\!I}$, $R_2$, $\hat{R}_{H\!I}$ and $\hat{R}_2$ are considered also for the network-households model).} \\
\toprule
Symbol & \centering Meaning & Section \\ \midrule
$R_0$ & Basic reproduction number (by default, based on rank generation numbers) & \ref{subsec:R0} \\
$R_\ast$ & Household reproduction number & \ref{subsec:R*} \\  
$R_I$ & Individual reproduction number & \ref{subsec:RI} \\ 
$R_{H\!I}$ & 
Individual reproduction number & \ref{subsec:RHI} \\ 
$\hat{R}_{H\!I}$, $\bar{R}_{H\!I}$ & Variants of $R_{H\!I}$ from \cite{GoldsteinEtal2009} ($\bar{R}_{H\!I}$ is \emph{not} a threshold parameter)
& \ref{subsec:RHI} \\ 
$R_2$ & 
Individual reproduction number &  \ref{subsec:R2} \\ 
$\hat{R}_2$ & Variant of $R_2$ from \cite{BalSirTra2010}
&  \ref{subsec:R2} \\ 
$R_V$ & Perfect vaccine-associated reproduction number &  \ref{subsec:RV} \\ 
$R_{V\!L}$ & Leaky vaccine-associated reproduction number &  \ref{subsec:RV} \\ 
$R_r$ & Exponential-growth-associated reproduction number &  \ref{subsec:Rr} \\ 
$\widetilde{R}_r$ & Variant of $R_r$
&  \ref{subsec:Rr} \\ 
$R_0^\mathrm{r}$ & Basic reproduction number based on rank generation numbers &  \ref{Hmodel}, \ref{subsubsec:genviewcomp} \\ 
$R_0^\mathrm{g}$ & Basic reproduction number based on true generation numbers &  \ref{Hmodel}, \ref{subsubsec:genviewcomp} \\ 
$R_A, R_B$ & Generic reproduction numbers & \ref{subsubsec:genviewcomp}  \\ 
\bottomrule
\label{tab:RsH}

\end{longtable}

\newpage

\LTcapwidth=\textwidth
\begin{longtable}{@{}p{.75\textwidth}l@{}} 
\caption{Existing, newly proved and conjectured inequalities.} 
\\
\toprule
\centering Result for growing epidemics &  \parbox{.18\textwidth}{\centering Reference} \\ \midrule
$R_{\ast} > R_I \ge R_V \ge R_0 \;(\ge R_2,\:\text{conjecture}) > R_{H\!I} > 1$ & Thm \ref{Hcomp} \& App~\ref{app:R0R2comp} \\
$R_I \ge R_2 > R_{H\!I} > 1$ & Thm \ref{Hcomp} \\
$\hat{R}_{H\!I} \ge \bar{R}_{H\!I}$ 
(
$\bar{R}_{H\!I}$ \emph{not} a threshold parameter) 
& Eq A4.1 of \cite{GoldsteinEtal2009}\\
$R_{H\!I} \ge \bar{R}_{H\!I}$ 
& App~\ref{app:R0RHIcomp} \\
$\hat{R}_{H\!I} \ge R_{H\!I}$, but $\hat{R}_{H\!I}$ and $R_0$ cannot be ordered 
& App~\ref{app:R0RHIcomp} \\ 
$R_0^\mathrm{r} \ge R_0^\mathrm{g}$ & Sec \ \ref{subsubsec:genviewcomp} \\
$R_{\ast} > R_{V\!L} \ge R_V \ge \bar{R}_{H\!I}$, but $R_{V\!L}$ and $R_I$ cannot be ordered
& Thm 1 of \cite{GoldsteinEtal2009} \& App~\ref{app:RIRVLcomp} \\ 
$R_\ast \ge R_r$ & Thm 1 of \cite{GoldsteinEtal2009} \\
$R_r \ge R_0$ in most commonly used models, but not in general & Thm \ref{propos}\\
$R_r \ge \widetilde{R}_r \ge R_0 $ in important special cases, but not in general & Thm \ref{propos}\\
$R_r$ and $R_{V\!L}$ cannot be ordered & App~\ref{app:RrRVLcomp} \\
$R_r$ and $R_V$ cannot be ordered & Sec~\ref{nonrandHM} \\
$R_r$ (or $\widetilde{R}_r$) and $R_I$ cannot be ordered & Sec~\ref{MarkovSIRHmod} \\
\midrule
\centering Result for declining epidemics &  \parbox{.18\textwidth}{\centering Reference} \\ \midrule
$R_{\ast} < R_I \le R_0 \;\le\text{(conjecture)}\, R_2 < R_{H\!I} < 1$ & Thm \ref{Hcomp} \& App~\ref{app:R0R2comp} \\
$\hat{R}_{H\!I} \ge \bar{R}_{H\!I}$ 
& Eq A4.1 of \cite{GoldsteinEtal2009}\\
$R_{H\!I} \ge \bar{R}_{H\!I}$ 
& App~\ref{app:R0RHIcomp} \\
$\bar{R}_{H\!I} \le \hat{R}_{H\!I} \le R_{H\!I} < 1$, but $\bar{R}_{H\!I}$ and $\hat{R}_{H\!I}$ cannot be ordered with $R_0$ 
& App~\ref{app:R0RHIcomp} \\ 
$R_\ast \le R_r$ & Sec~\ref{subsubsec:Rrcomp} \\
$R_r \le R_0$ in most commonly used models, but not in general & Thm \ref{propos}\\
$R_r \le \widetilde{R}_r \le R_0 $ in important special cases, but not in general & Thm \ref{propos}\\
\bottomrule
\label{tab:Inequalities}
\end{longtable}


\LTcapwidth=\textwidth
\begin{longtable}{@{}cp{.7\textwidth}c@{}} 
\caption{Reproduction numbers for the \textbf{households-workplaces} model.} \\
\toprule
Symbol & \centering Meaning & Section \\ \midrule
$R_0$ & Basic reproduction number (by default, based on rank generation numbers) &  \ref{subsec:HWR0}  \\
$R_\ast$ & Clump reproduction number &  \ref{subsec:HWR*} \\  
$R_H$ & Household-household reproduction number &   \ref{subsec:RHRW} \\ 
$R_W$ & Workplace-workplace reproduction number &   \ref{subsec:RHRW} \\ 
$R_I$ & Individual reproduction number &   \ref{subsec:HWRI} \\ 
$R_V$ & Perfect vaccine-associated reproduction number &   \ref{subsec:HWRV} \\ 
$R_{V\!L}$  & Leaky vaccine-associated reproduction number &   \ref{subsec:HWRV} \\ 
$R_r$ & Exponential-growth-associated reproduction number &   \ref{subsec:HWRr} \\ 
$\widetilde{R}_r$ & Approximation of $R_r$ 
&   \ref{subsec:HWRr} \\ 
\bottomrule
\label{tab:RsHW}
\end{longtable}

\begin{longtable}{@{}cp{.65\textwidth}c@{}}

\caption{Main symbols used in the paper, with reference to their first occurrence.} \\

\toprule

Symbol & \centering Meaning & Section \\ \midrule

$a^{(n)}$ & $= \mu_H^{(n)} / \pp{1 + \mu_H^{(n)}}$ (construction of $R_{H\!I}$) & \ref{subsec:RHI} \\

$a$ & $= \max\pp{a^{(n)}:n=1,2,\cdots,n_H}$ (construction of $R_{H\!I}$) & \ref{subsec:RHI} \\

$b$ & Mean number of secondary cases attributed to other secondary cases (construction of $R_2$) & \ref{subsec:R2} \\

$\mathcal{D}\,^{\mathcal{C}}$ & Complement of event $\mathcal{D}$& \ref{subsec:Hcompproof}\\

$\bbE_{X}$ & Expectation with respect to random variable $X$ & \ref{sec:Introduction} \\

$\mathcal{E}$ & Vaccine efficacy & \ref{subsec:RV} \\

$\mathcal{E}_C$ & Critical vaccine efficacy & \ref{subsec:RV} \\

$g_A(\lambda)$ & Characteristic equation (discrete Lotka-Euler equation) derived from $M_A$ defining reproduction number $R_A$ & Throughout \\

$H, W$ & Subscripts/superscripts referring to household or workplace & Throughout \\

$i,i'$ & Individuals' indices & \ref{sec:Introduction} \\

$\{ \mathcal{I}(t), t \ge 0 \}$ & Random infectivity profile & \ref{subsec:Rr} \\

$k$ & Generation index & \ref{sec:Introduction} \\

$\mathcal{L}_f$ & Laplace transform of non-negative function $f$, i.e.~$\mathcal{L}_f(\theta) = \int_{-\infty}^{\infty}{\e^{-\theta f(x)}{\rm d}x}$ & \ref{subsec:Rr} \\

$\mathcal{M}_X$ & Moment-generating function of random variable $X$, i.e.~$\mathcal{M}_X(\theta) = \bbE\b{\e^{-\theta X}}$ & \ref{subsec:Rr} \\

$M_A$ & Mean matrix associated with reproduction number $R_A$ & Throughout \\

$n$ & Household size index & \ref{Hmodel} \\

$n_H$ & Maximum household size & \ref{Hmodel} \\

$N$ & Total population size & \ref{sec:Introduction} \\

$N_{ii'}$ & Number of infectious contacts from $i$ to $i'$ between the infection and the recovery of $i$ (Perhaps delete this one)& \ref{subsec:HRrproof} \\

$p_C$ & Critical vaccination coverage & \ref{subsec:RV} \\

$r$ & Real-time growth rate, i.e.~Malthusian parameter for the epidemic growth & \ref{subsec:Rr} \\

$R_A$ & Reproduction number associated with construction process $A$ & Throughout \\

$T_i$ & Time of infection of $i$ & \ref{subsec:HRrproof} \\

$T_E$ & Duration of latency period in SEIR model & \ref{MarkovSIRHmod} \\

$T_I$ & Duration of infectious period in SIR and SEIR models & \ref{subsec:Rr} \\

$\gtrv$ & Random variable describing the time of an infectious contact between two individuals (since infection of the infector) & \ref{subsec:Rr} \\

$\gtpdf$ & Probability density function of $\gtrv$, also called generation-time distribution & \ref{subsec:Rr} \\

$\gtrvt$ & Random variable describing the time of the first infectious contact between two individuals (since infection of the infector), assuming at least one occurs & \ref{subsec:Rr} \\

$Y_k$ & Number of infected cases in generation $k$ of a (randomised) Reed-Frost model & \ref{subsec:Hcompproof} \\

$\alpha$ & Shape parameter of the gamma distribution & \ref{nonrandHM} \\

$\beta_H(t)$ & Mean rate at which global infections emanate from a household & \ref{subsec:Rr} \\

$\delta$ & Rate of progressing from latent to infectious state in SEIR model & \ref{MarkovSIRHmod} \\

$\gamma$ & Recovery rate for SIR and SEIR models when $T_I$ is exponentially distributed; also, scale parameter of the gamma distribution & \ref{subsec:Rr} \\

$\lambda_G$ & Multiplicative coefficient affecting rate at which each infective makes infectious contacts in the population at large & \ref{subsec:Rr} \\

$\lambda_H$ & Multiplicative coefficient affecting the rate at which each infective makes infectious contacts to any specified susceptible within household & \ref{subsec:Rr} \\

$\mu_G$ & Mean number of global contacts made by a typical infective & \ref{Hmodel} \\

$\mu_H$ & Mean size of a within-household epidemic & \ref{Hmodel} \\

$\mu_k$ & Mean number of cases in generation $k$ of a within-household epidemic ($\mu_0 = 1$ always) & \ref{Hmodel} \\

$\mu_H^{(n)}, \mu_k^{(n)}$ & Mean size of a within-household epidemic, or of generation $k$ in such epidemic, in a household of size $n$ & \ref{Hmodel} \\

$\pi_n$ & Probability that the household of an individual selected uniformly at random has size $n$ (size-biased distribution) & \ref{Hmodel} \\

$\indic_\mathcal{D}$ & Indicator function, with value 1 if $\mathcal{D}$ occurs and 0 otherwise & \ref{subsubsec:Rrcomp} \\

$\overset{st}{\leq}$ & Stochastically smaller & \ref{subsubsec:genviewcomp} \\

$\overset{n}{\leq}$ & Inequality, which is strict only if at least one household or workplace has size larger than $n$ and is an equality if all households and workplaces have size $\leq n$ & \ref{subsec:Hcompproof} \\

$\overset{D}{=}$ & Equal in distribution& App~\ref{app:inflonglat}\\




\bottomrule

\label{tab:symbf}

\end{longtable}

\newpage
\section{Households model and reproduction numbers}
\label{sec:house}

\subsection{Model and generations of infections}

\label{Hmodel}
In this section we outline the definition of the households model, giving sufficient detail so that $R_0$ can be calculated.
The salient features for this purpose are that the population is partitioned into households and that infectives make two types of infectious contacts, {\it local} contacts with individuals in the same households and {\it global} contacts with individuals chosen uniformly at random from the entire population.
The expected number of global contacts made by a typical infective
during his/her infectious period 
is assumed to be $\mu_G$ and is the same for all infectives.
The precise detail of local transmission is not required in order to define $R_0$, as long as we can compute the generations of infection in the {\it local} epidemic
(i.e.\ in the within-household epidemic obtained if all {\it global} contacts are ignored). 
We show now how this may be done.

Consider a local epidemic in a household of size $n$, with $1$ initial infective, labelled 0, and $n-1$ initial susceptibles, labelled $1,2,\cdots,n-1$ (See Figure \ref{epigraphfig}). 
For $i=0,1,\cdots,n-1$, construct a list of whom individual $i$ would attempt to infect in the household if $i$ were to become infected.  Then construct a directed graph, $\mathcal{G}^{(n)}$ say, with vertices labelled $0,1,\cdots,n-1$, in which for any ordered pair of distinct vertices $(i,i')$, there is
a directed edge from $i$ to $i'$ if and only if individual $i'$ is in individual $i$'s list
of attempted infections. 
The initial infective, i.e.~individual $0$, is said to have (household) generation $0$.
Those individuals who are in individual $0$'s list (i.e.\ individuals 1 and 2 in Figure \ref{epigraphfig}) are said to have generation $1$.  
Those individuals who are not in generations $0$ or $1$ but who are in a generation-$1$ infective's list (i.e.\ individuals 4 and 5 in Figure \ref{epigraphfig}) have generation $2$,
and so on. The set of people ultimately infected by the epidemic comprises those individuals in $\mathcal{G}^{(n)}$ that  have a chain of directed edges leading to them from individual $0$, and the generation number of such an infected individual, $i$ say, is the length of the shortest chain joining $0$ to $i$, where the length of a chain is the number of edges in it. Following Ludwig~\cite{Ludwig1975}, we call these generation numbers
{\em rank} generation numbers.

The rank generations of infectives may not correspond to real-time generations of infectives.  The latter may be obtained by augmenting the graph $\mathcal{G}^{(n)}$, so that for each directed edge, $i \to i'$ say, in $\mathcal{G}^{(n)}$ there is a number $t_{ii'}$ giving the time elapsing between $i$'s infection and 
time at which $i$ first attempts to infect $i'$. 
Then the generation number of an individual, $i$ say, that is infected in the single-household epidemic is the number of directed edges in the shortest chain joining $0$ to $i$, where now the length of a chain is the sum of the $t_{ii'}$ of its directed edges.  We call these generation numbers {\em true} generation numbers. 
As an example, suppose for the epidemic graph of Figure \ref{epigraphfig} that $t_{01} + t_{12} < t_{02}$, then the true generation of individual 2 is 2, instead of 1, which is his/her rank generation. 
For ease of exposition, unless stated explicitly otherwise, we assume rank generation numbers throughout this paper. This is in line with the choice of the definition of $R_0$ made in \cite{PelBalTra2012}. For further clarification, when both generation constructions are considered, as in Section \ref{subsubsec:genviewcomp},  we refer to the rank-generation basic reproduction number by using $R_0^{\mbox{r}}$, as opposed to the basic reproduction number $R_0^{\mbox{g}}$ which is obtained using the true generations. As explained in \cite{PelBalTra2012}, the reasons for the above choice are both analytical tractability and the fact that $R_0^{\mbox{r}}$ depends (in addition to the household structure) only on the distribution of the total infectivity of an individual, and not on the particular shape of his/her infectivity profile
(i.e.\ the distribution of the random development of the infectivity of an individual after he/she gets infected).

Consider a household of size $n$.  For $k=0,1,\cdots,n-1$, let $\mu_k^{(n)}$ be the mean
size of generation $k$ in the above single-household epidemic.  Thus $\mu_0^{(n)}=1$ and
$\mu_H^{(n)}=\mu_1^{(n)}+\mu_2^{(n)}+\cdots+\mu_{n-1}^{(n)}$ is the mean size of the epidemic, not including the initial case. (Note that $\mu_H^{(1)}=0$.)  If the population contains households of different sizes then we need to take appropriate averages of these quantities.  Let $n_H$ denote the size of the largest household in the population and, for $n=1,2,\cdots,n_H$, let $p_n$ denote the proportion of households in the population that have size $n$.  Then the probability that an individual chosen uniformly at random from the population resides in a household of size $n$
is given by
\begin{equation}
\label{sizebias}
\pi_n=\frac{n p_n}{\sum_{j=1}^{n_H} j p_j} \qquad (n=1,2,\cdots,n_H).
\end{equation}
Global contacts are made with individuals chosen uniformly at random from the population, so the mean generation sizes of a typical single-household epidemic are given by
\begin{equation}
\label{sizebiasgen}
\mu_k=\sum_{n=k+1}^{n_H} \pi_n \mu_k^{(n)} \qquad (k=0,1,\cdots,n_H-1).
\end{equation}
The mean size of a typical single-household epidemic, not including the initial infective, is then given by
\begin{equation}
\label{muLunequal}
\mu_H=\sum_{n=1}^{n_H} \pi_n \mu_H^{(n)}=\sum_{k=1}^{n_H-1} \mu_k.
\end{equation}
In what follows we assume that $\mu_G>0$, otherwise the infection does not spread between households, and that $\mu_H>0$ and $n_H\ge 2$, otherwise the model is homogeneously mixing.

\subsection{The basic reproduction number $R_0$}
\label{subsec:R0}

Consider the branching process that approximates the early spread of the epidemic, in which each individual in the branching process represents an infected household and the time of its birth is given by the global generation of the corresponding household primary case in the epidemic process.  (The global generation of an infective is its generation in the epidemic in the population at large.  A household primary case is the first infected individual in the household and all other cases are called secondary.) See Figure \ref{fig:epigraphbp} for a graphical representation. A typical, non-initial individual in this branching process (i.e.~a household) reproduces only at ages $1,2,\cdots$ and its mean number of offspring at age $k+1$ is $\nu_k$, where $\nu_k=\mu_G \mu_k$ $(k=0,1,\cdots,n_H-1)$ and $\nu_k=0$ otherwise.  The asymptotic (Malthusian) geometric growth rate of this branching process is given by the unique positive solution of the discrete-time Lotka-Euler equation $\sum_{k=0}^{\infty} \nu_k/\lambda^{k+1}=1$; see, for example, Haccou et al.~\cite{Haccou2005}, Section 3.3.1, adapted to the discrete-time setting.  The above branching process may be augmented to include the local spread within each household, i.e.\ considering all individuals in Figure~\ref{fig:epigraphbp}. Assume, as in Figure~\ref{fig:epigraphbp}, that all households live up to age $n_H$, even if local epidemics finish earlier.  (Note that this assumption does not alter the asymptotic geometric growth rate of the
branching process.)  Then, for $k \ge n_H$, the expected number of households in global generation $k$ of the branching process is
\[ \mu_G \left( x_{k-1} + x_{k-2} + \cdots + x_{k-n_H} \right), \]
where $x_k$ denotes the the expected number of individuals in global generation $k$ of the augmented process.  Therefore, the asymptotic geometric growth rate of the total number of infectives in the augmented process is the same as that of the branching process\footnote{A formal proof of this can easily be obtained using arguments similar to those in the proof of Lemma 3 of [2] (though note that the left-hand side of the second display after (3.15) should read $A_{n_H}^n$).}, so the basic reproduction number $R_0$ for the above households model is given by the unique positive root of the function
\begin{equation}
\label{g0H}
g_0 (\lambda )=1-\sum _{k=0}^{n_{H} -1}\frac{\nu_k}{\lambda ^{k+1} } = 1-\mu_G \sum _{k=0}^{n_{H} -1} \frac{\mu_k}{\lambda ^{k+1} } ,
\end{equation}
yielding a simpler proof of Corollary 1 in~\cite{PelBalTra2012}.
For future reference, we note that $$g_0 (\lambda )=\sum_{n=1}^{n_H}\pi_n g_0^{(n)}(\lambda),$$ where
\begin{equation}
\label{gnH}
g_0^{(n)} (\lambda )=1-\mu_G\sum_{k=0}^{n-1}\frac{\mu_k^{(n)}}{\lambda^{k+1}}.
\end{equation}

\begin{figure}[!ht]
\centering
\includegraphics[width=.8\textwidth]{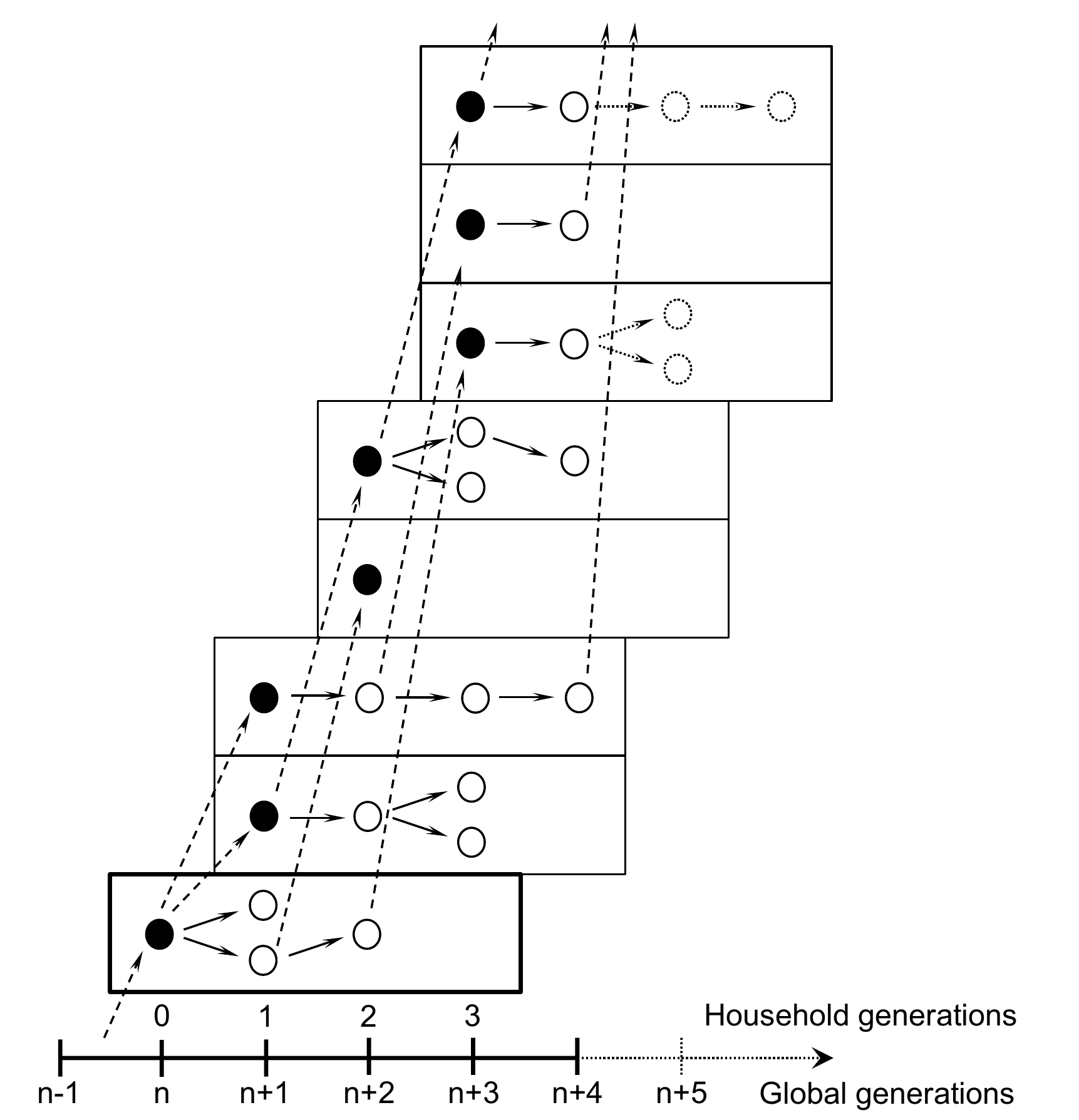}
\caption{Graphical representation of the branching process construction used in the derivation of $R_0$, for a population of households of size 4. At global generation $n$, a specific household (thick borders) is infected, i.e.\ a global infection generates its primary case (black dot), who then starts a within-household epidemic driven by household infections (normal arrows). Household and global generations then proceed at the same pace, with both primary and secondary cases generating only new household primary cases through global contacts (dashed arrows). 
Note that all households have lifespans of 4 generations (the latest time, measured in generations, at which global contacts can be made given their size), even if the within-household epidemics are shorter.}\label{fig:epigraphbp}
\end{figure}

In the above we assume that all infected individuals make the same expected number of global contacts $\mu_G$.  This is the case for most households models that have appeared in the literature.  One exception is the network-households model of Ball et al.~\cite{BalSirTra2009,BalSirTra2010}, in which the mean number of global contacts made by primary and secondary household infectives are $\tilde{\mu}_G$ and $\mu_G$, respectively,
where $\tilde{\mu}_G$ and $\mu_G$ may be unequal. Pellis et al.~\cite{PelBalTra2012} show that
$R_0$ for the network-households model is given by the unique positive root of $g_0$ but with $\nu_0=\tilde{\mu}_G \mu_0$ (all other $\nu_k$ remain unchanged).

\subsection{The household reproduction number $R_*$}
\label{subsec:R*}

The most commonly used reproduction number for the households model is given by the mean number of households infected by a typical infected household in an otherwise susceptible population.  It is usually denoted by $R_*$ and in our notation is given by
\begin{equation}
R_* = \mu_G (1+\mu_H ) = \sum_{k=0}^{n_H-1} \nu_k .\label{rstar}
\end{equation}
The popularity of $R_*$ stems largely from its ease of calculation and from the fact that, if $R_*>1$, selecting a fraction $1-1\!/\!R_*$ of households uniformly at random and vaccinating all their members is enough to prevent an epidemic.

\subsection{The individual reproduction number $R_I$}
\label{subsec:RI}
Several authors have proposed individual-based reproduction numbers for the households model.  One approach (see, for example, Becker and Dietz~\cite{BecDie1995} and Ball et al.~\cite{BalMolSca1997}) is to attribute all secondary cases in a household to the primary case, leading to the reproduction number $R_I$ given by the dominant eigenvalue of the next-generation matrix
\[
M_I = \left[ \begin{array}{cc}\mu_G&\mu_H\\\mu_G&0\end{array}\right] .
\]
It is easily verified that $R_I$ is given by the unique solution in $(0,\infty)$ of
$g_I (\lambda)=0$, where
\begin{equation}
g_I (\lambda)=1- \frac{\mu_G}{\lambda} - \frac{\mu_H \mu_G}{\lambda^2}.\label{gind}
\end{equation}

\subsection{The individual reproduction number $R_{H\!I}$}
\label{subsec:RHI}

Goldstein et al.~\cite{GoldsteinEtal2009} consider an individual reproduction number, which they denote by $R_{H\!I}$, and which represents ``the expected number of secondary cases caused by an average individual from an average infected household, including those outside and inside the household'' (see also Trapman \cite{Trapman2007}).  Suppose first that all households have the same size.
Then, in an ``average'' household epidemic, there are $\mu_H$ secondary cases caused by $\mu_H +1$ infectives, leading to
\begin{equation}
R_{H\!I} = \mu_G + \frac{\mu_H}{1+\mu_H}.\label{rhiequal}
\end{equation}

Goldstein et al.~\cite{GoldsteinEtal2009} also consider an extension of \eqref{rhiequal} to variable household sizes\footnote{In \cite{GoldsteinEtal2009}, this extension is also denoted by $R_{H\!I}$.}, defined by
\begin{equation}
\bar{R}_{H\!I}=\mu_G + \sum_{n=1}^{n_H} \pi_n \left( \frac{\mu_H^{(n)}}{1+\mu_H^{(n)}} \right).\label{rhiunequal}
\end{equation}
However, $\bar{R}_{H\!I}$ given by \eqref{rhiunequal} is not necessarily a threshold parameter.  For this reason, Goldstein et al.~\cite{GoldsteinEtal2009} proposed another extension of \eqref{rhiequal}, defined by\footnote{In \cite{GoldsteinEtal2009}, this is denoted by $R_{H\!I}'$.}
\begin{equation}
\label{rhihatnew}
\hat{R}_{H\!I} = \mu_G + \frac{\mu_H}{1+\mu_H},
\end{equation}
with $\mu_H$ as in \eqref{muLunequal}, which is a threshold parameter. The advantages and disadvantages of $\bar{R}_{H\!I}$ and $\hat{R}_{H\!I}$ are discussed in \cite{GoldsteinEtal2009}. The problem with $\bar{R}_{H\!I}$ is that it is not generally a threshold parameter. The problem with $\hat{R}_{H\!I}$ is that (unlike $\bar{R}_{H\!I}$)  there exist household structures for which $\hat{R}_{H\!I}$ does not satisfy the general orderings of reproduction numbers proved in~\cite{GoldsteinEtal2009}.
We renamed the original definitions because we now introduce a new definition of
$R_{H\!I}$ for populations of unequally sized households, which overcomes both these shortcomings and coincides with both $\bar{R}_{H\!I}$ and $\hat{R}_{H\!I}$ when all households have the same size.

Returning to the setting where all households have the same size, note that~\eqref{rhiequal} assumes that each household member produces on average $a=\mu_H/(1+\mu_H)$ secondary cases within the household, so the mean generation sizes are given by $\eta_k=a^k$ $(k=0,1,\cdots)$ and, cf.~(\ref{g0H}),
$R_{H\!I}$ ($=\mu_G+a$) is given by the unique root in $(a,\infty)$ of the function
\begin{align}
\label{gnHI}
g_{H\!I}(\lambda ) &= 1-\mu_G \sum_{k=0}^{\infty}\frac{\eta_k}{\lambda ^{k+1} } =1-\frac{\mu_G}{\lambda -a}.
\end{align}
Using this approach, if the households are not all the same size, then the mean generation sizes are given by $\eta_k=\sum_{n=1}^{n_H} \pi_n \left(a^{(n)}\right)^k$ ($i=0,1,\cdots$), where $a^{(n)}=\mu_H^{(n)}/(1+\mu_H^{(n)})$ for $(n=1,2,\cdots,n_H)$, which leads to the reproduction number $R_{H\!I}$ given by the unique root in $(a, \infty)$, where now $a = \max(a^{(n)}:n=1,2,\cdots,n_H)$, of the function
\begin{align}
\label{gHI}
g_{H\!I}(\lambda) &= 1-\mu_G \sum_{k=0}^{\infty}\frac{\eta_k}{\lambda ^{k+1} } = 1-\mu_G \sum_{n=1}^{n_H} \frac{\pi_n}{\lambda-a^{(n)}} \qquad (\lambda > a).
\end{align}

\subsection{The individual reproduction number $R_2$}
\label{subsec:R2}

A disadvantage of $R_I$ is that every secondary case in a household is attributed to the primary case whereas in practice some should normally be attributed to other secondary cases.  Suppose that all households have the same size, which is at least two.  Ball et al.~\cite{BalSirTra2010} consider a modification of $R_I$ in which $M_I$ is replaced by
\[
M_2 = \left[ \begin{array}{cc}\mu_G&\mu_1\\\mu_G&b\end{array}\right],
\]
where $b=1-\mu_1 / \mu_H$.  Thus every secondary case produces on average $b$ further secondary cases, with the value of $b$ being chosen so that the within-household spread yields the correct expected final size, i.e.~so that $\mu_H = \mu_1 (1+b+b^2+ \cdots )= \mu_1 / (1-b)$. 
Note that $R_2$ satisfies $\hat{g}_2 (R_2)=0$, where
\begin{equation}
\hat{g}_2 (\lambda) = 1 - \frac{\mu_G +b}{\lambda} + \frac{\mu_G (b-\mu_1)}{\lambda^2}. \label{g2}
\end{equation}
At the end of the proof of Theorem~\ref{Hcomp} (see Section~\ref{subsec:HWcompproof}) we show that $b < \mu_1$.  It then follows that $R_2$ is given by the unique root of $\hat{g}_2$ in $(0,\infty)$.

Observe that the above assumes that the mean generation sizes are given by $\upsilon_0=1$ and  $\upsilon_k=\mu_1 b^{k-1}$ $(k=1,2,\cdots)$.
It follows that $R_2$ is given by the unique root in $(b,\infty)$ of the function
\begin{align}
\label{gn2}
g_2(\lambda )&=1-\mu_G \sum_{k=0}^{\infty}\frac{\upsilon_k}{\lambda ^{k+1} }=1-\frac{\mu_G}{\lambda}\left(1+\frac{\mu_1}{\lambda-b}\right).
\end{align}
(It is easily verified that $\hat{g}_2(\lambda) = \left(1-\frac{b}{\lambda}\right) g_2(\lambda)$.)

If the households are not all the same size, we can define the mean generation sizes as $\upsilon_0 = 1$ and $\upsilon_k = \sum_{n=2}^{n_H}{\pi_n \mu_1^{(n)} \left(b^{(n)}\right)^{k-1}}$ ($k=1,2,\cdots$),  where $b^{(n)}=1-\mu_1^{(n)}/\mu_H^{(n)}$ for $(n=2,3,\cdots,n_H)$.
The reproduction number $R_2$ is then given,
to be the unique root in $(b, \infty)$ of the function
\begin{equation}
\label{g2unequal}
g_{2}(\lambda )= 1 - \mu_G \sum_{k=0}^{\infty}{\frac{\upsilon_k}{\lambda^{k+1}}}
= 1-\frac{\mu_G}{\lambda}\left(1+\sum_{n=2}^{n_H} \pi_n \frac{ \mu_1^{(n)}}{\lambda-b^{(n)}}\right),
\end{equation}
where 
$b$ is now given by $b = \max(b^{(n)}:n=2,3,\cdots,n_H)$.  Note that~\eqref{g2unequal} reduces to~\eqref{gn2} when all households
have the same size.

This extension of $R_2$ to unequal household sizes differs from that in \cite{BalSirTra2010}, where $R_2$ is defined to be the
dominant eigenvalue of $M_2$ above, with $\mu_1$ and $\mu_H$ as in \eqref{sizebiasgen} and \eqref{muLunequal}. 
We denote the latter by $\hat{R}_2$, as it is similar in spirit to $\hat{R}_{H\!I}$. We do not consider it further.

\subsection{The perfect and leaky vaccine-associated reproduction numbers $R_V$ and $R_{V\!L}$}
\label{subsec:RV}

Goldstein et al.~\cite{GoldsteinEtal2009} consider two vaccine-associated reproduction numbers, $R_V$ and $R_{V\!L}$, corresponding to perfect and leaky vaccines, respectively.  Suppose that the epidemic is above threshold, i.e.~$R_* > 1$, and individuals are selected uniformly at random and vaccinated with a perfect (i.e.~100\% effective) vaccine.  Let $p_C$ be the proportion of the population that has to be vaccinated to reduce $R_*$ to 1.  Then
\begin{equation}
R_V = 1/(1-p_C).\label{rv}
\end{equation}
Thus $R_V$ is defined in such a way that the critical vaccination coverage is given by $1-1/R_V$, paralleling the usual formula for a homogeneously mixing epidemic, where, if $R_0>1$, the critical vaccination coverage is $1-1/R_0$. Goldstein et al.~\cite{GoldsteinEtal2009} also introduce in Section 7.2 of their paper a reproduction number $R_{V\!\!A}$, which approximates $R_V$.  In our notation, $R_{V\!\!A}$ is obtained by multiplying both $\mu_H$ and  $\mu_G$ by $(1-p)$ in \eqref{rstar}, finding the critical vaccination coverage $p_C$ that reduces $R_*$ to 1, and then using \eqref{rv} to obtain an approximation $R_{V\!\!A}$ to $R_V$. It is easily checked that $R_{V\!\!A}=R_I$ (see the proof in Section~\ref{subsec:Hcompproof} of $R_I \ge R_V$ in Theorem~\ref{Hcomp}\thp{b}\footnote{Note though that there is a small misprint in the formula for $R_{V\!\!A}$
at the foot of page 19 of \cite{GoldsteinEtal2009} ($\sqrt{4(f-1)/R_G}$
should be replaced by $\sqrt{1+4(f-1)/R_G}$).}).

A leaky vaccine with efficacy $\mathcal{E}$, is one which multiplies a vaccinee's susceptibility to a disease by a factor $1-\mathcal{E}$ but has no effect on a vaccinee's infectivity if he/she becomes infected. More specifically, each time any infective attempts to infect a vaccinated susceptible individual that individual is infected independently with probability $1-\mathcal{E}$. Suppose that $R_* > 1$ and the entire population is vaccinated with a leaky vaccine.  Then
\begin{equation}
R_{V\!L} = 1/(1-\mathcal{E}_C),\label{rvl}
\end{equation}
where $\mathcal{E}_C$ is the efficacy required to reduce $R_*$ to 1.

The above definitions of $R_V$ and $R_{V\!L}$ assume that $R_* > 1$.  Goldstein et al.~\cite{GoldsteinEtal2009} did not define $R_V$ and $R_{V\!L}$ when $R_*\leq 1$. In that case we define $R_V=R_{V\!L}=1$, since a major outbreak cannot occur even if nobody is vaccinated.

\subsection{The exponential-growth-associated reproduction number $R_r$ }
\label{subsec:Rr}

A final reproductive number considered in~\cite{GoldsteinEtal2009} is the exponential-growth-associated reproduction number $R_r$, whose definition requires a more detailed description of the transmission model.  Goldstein et al.~\cite{GoldsteinEtal2009} consider a households models in which infectives have independent and identically distributed infectivity profiles. A typical infectivity profile, $\mathcal{I}(t)\ (t\ge 0)$, is the realisation of a stochastic process; conditional upon its infectivity profile, an infectious individual, $t$ time units after being infected, makes global contacts at overall rate $\mu_G \mathcal{I}(t)$ and contacts any given susceptible in his/her household at rate $\lambda_H^{(n)}\mathcal{I}(t)$, where $n$ is the size of his/her household\footnote{The notation has been changed to fit more closely that of our paper.}. All infectious contacts, whether of the same or different type (i.e.~local or global) are independent of each other. For $t \ge 0$, let $\gtpdf(t)=\mathbb{E}[\mathcal{I}(t)]$ and note that, since $\mu_G$ is the mean number of global contacts made by a typical infective, $\int_0^{\infty} \gtpdf(t){\rm d}t=1$.  Thus $\gtpdf$ may be interpreted as the probability density function of a random variable, $\gtrv$ say, describing an infectious contact interval (see e.g.~\cite{GoldsteinEtal2009} and~\cite{PelHouseSpe2015}).

Suppose first that $\lambda_H^{(n)}=0$ for all $n$, so the epidemic is homogeneously mixing, with basic reproduction number $R_0=\mu_G$ and real-time growth rate $r$ given by the implicit solution of the Lotka-Euler equation
\begin{equation}
\label{LotkaEuler}
\int_0^{\infty} \mu_G \gtpdf(t){\rm e}^{-rt} {\rm d}t=1.
\end{equation}
Thus, $R_0=\left(\mathcal{M}_{\gtrv}(r)\right)^{-1},$
where $\mathcal{M}_{\gtrv}(\theta)=\int_0^{\infty}{\rm e}^{-\theta t}\gtpdf(t){\rm d}t$ is the moment-generating function of $\gtrv$.  (Throughout the paper,
for a random variable $X$ we denote its moment-generating function by $\mathcal{M}_{X}(\theta)=\bbE\b{\e ^{- \theta X}}$.)
This provides a method of estimating $R_0$ from data on an emerging epidemic, when information on $\gtrv$ and the exponential growth rate $r$ are available, assuming a homogeneous mixing model (see Nowak et al.~\cite{Nowaketal1997}, Lloyd~\cite{Lloyd2001}, Wallinga and Lipsitch~\cite{WalLip2007} and Roberts and Heesterbeek~\cite{RobHee2007}).

The exponential-growth-associated reproduction number $R_r$ in \cite{GoldsteinEtal2009} is given by
\begin{equation}
\label{Rrfromr}
R_r=\frac{1}{\mathcal{M}_{\gtrv}(r)},
\end{equation}
where $r$ is the real-time growth rate of the households model.
Thus, in the above inferential setting, $R_r$ is the estimate one obtains of $R_0$ if
the household structure of the population is ignored.

To calculate $R_r$, it is necessary to calculate first the real-time growth rate $r$ of the households model, which
generally is far from straightforward.  For $t>0$, let $\beta_H(t)$ denote the mean rate at which global contacts emanate from a typical single-household epidemic $t$ time units after the household was infected. Similarly to~\eqref{LotkaEuler}, the real-time growth rate $r$ is now given by the unique real solution of the Lotka-Euler equation
\begin{equation}
\label{rtgrowth1}
\mathcal{L}_{\beta_H}(r)=1,
\end{equation}
where $\mathcal{L}_{\beta_H}(r)=\int_0^{\infty} \beta_H(t){\rm e}^{-r t} {\rm d}t$.
Note that $\mathcal{L}_{\beta_H}(r)$ is the Laplace transform of the household infectivity profile; hereafter, we denote by $\mathcal{L}_{f}(\theta)$ the Laplace transform of a function $f$ calculated in $\theta\in (-\infty,\infty)$\footnote{For ease of notation we give the domain as $(-\infty,\infty)$ but, as all the functions we consider are non-negative, we note that the domain
usually takes the form $(\theta_f, \infty)$, where $\theta_f$ depends on the function $f$, as the integral is infinite for $\theta \le \theta_f$.}.
The difficulty in calculating $r$ from \eqref{rtgrowth1} is that $\mathcal{L}_{\beta_H}(r)$ is generally not mathematically tractable unless the disease dynamics are Markovian. Consequently, Fraser~\cite{Fra2007} introduced an approximation, further explored in Pellis et al.~\cite{PelFerFra2011}, which essentially assumes that cases are attributed to generations according to the rank-based process and real infection intervals (not only infectious contact intervals) are independent realisations of the random variable $\gtrv$ (see Pellis et al.~\cite{PelHouseSpe2015} for an extensive discussion).  With this approximation, the time elapsing from the initial infection of the household to the infection of a typical household generation-$k$ infective is given by the sum of $k$ independent copies of $\gtrv$, so $\mathcal{L}_{\beta_H}(r) \approx \mathcal{L}_{\beta_H}^{(0)}(r)$, where
\begin{equation}
\label{rtgrowth2}
 \mathcal{L}_{\beta_H}^{(0)}(r)=\mu_G \mathcal{M}_{\gtrv}(r)\left\{1+ \sum_{k=1}^{n_H-1} \mu_k \left(\mathcal{M}_{\gtrv}(r)\right)^k\right\}.
\end{equation}
Substituting this approximation into~\eqref{rtgrowth1}, using~\eqref{Rrfromr} and recalling that $\mu_0=1$, yields  
\[
\mu_G \sum_{k=0}^{n_H-1}\frac{\mu_k}{R_r^{k+1}}=1,
\]
so, recalling~\eqref{g0H}, $g_0(R_r)=0$.  Thus, using Fraser's approximation leads to $R_r$ being given by $R_0$.

A second approximation to $R_r$ is perhaps most easily introduced by considering the infectivity profile
given by
\begin{equation}
\label{IGSE}
\mathcal{I}(t)=\left\{
\begin{array}{ll}
1\qquad&\text{if } 0 \le t \le T_I,\\
0 &\text{if } t>T_I,
\end{array} \right.
\end{equation}
where $T_I \sim \mathrm{Exp}(1)$.  (For $\gamma>0$, $\mathrm{Exp}(\gamma)$ denotes an exponential distribution having rate $\gamma$, and hence mean $\gamma^{-1}$.)  Thus, for $t\ge 0$, we have $\gtpdf(t)=\mathbb{E}[\mathcal{I}(t)]=\mathbb{P}(T_I \ge t)={\rm e}^{-t}$, so $\gtrv \sim \mathrm{Exp}(1)$. Suppose that $\lambda_H^{(n)}=\lambda_H$ for all $n$.  Let the random variable $\gtrvt$ describe the time of the first local infectious contact from a given infective to a given susceptible in the same household, conditional upon there being at least one such contact.  The infective contacts the susceptible at rate $\lambda_H$ and recovers independently at rate $1$, so the time until the first event (contact of the susceptible or recovery of the infective) has an $\mathrm{Exp}(1+\lambda_H)$ distribution.  Moreover, whether or not this event is a recovery is independent of its time.  Thus $\gtrvt \sim \mathrm{Exp}(1+\lambda_H)$.  Note that $\gtrvt$ has a different distribution from $\gtrv$, so another (and usually improved) approximation is
$\mathcal{L}_{\beta_H}(r) \approx \widetilde{\mathcal{L}}_{\beta_H}(r)$, where
\begin{equation}
\label{rtgrowth3}
 \widetilde{\mathcal{L}}_{\beta_H}(r)=\mu_G \mathcal{M}_{\gtrv}(r)\left\{1+ \sum_{k=1}^{n_H-1} \mu_k \left(\mathcal{M}_{\gtrvt}(r)\right)^k\right\}.
\end{equation}
Note that the first $\mathcal{M}_{\gtrv}(r)$ in~\eqref{rtgrowth2} is not replaced by $\mathcal{M}_{\gtrvt}(r)$ since it corresponds to a global contact (recall that, asymptotically, an infective makes at most one global contact to any given susceptible).  The random variable ${\gtrvt}$ can be defined in a similar fashion for any arbitrary but specified infectivity profile (see~\cite{PelHouseSpe2015} for numerous explicit examples).  Using this approximation, the real-time growth rate $r$ is approximated by $\tilde{r}$, where $\tilde{r}$ is
the unique real solution of
\begin{equation}
\label{rtgrowth4}
\widetilde{\mathcal{L}}_{\beta_H}(r)=1,
\end{equation}
which leads to $R_r$ being approximated by the reproduction number
\begin{equation}
\label{Rtilderdef}
\widetilde{R}_r=\frac{1}{\mathcal{M}_{\gtrv}(\tilde{r})}.
\end{equation}

The above example shows that if $\lambda_H$ varies with household size $n$ then so does the distribution
of ${\gtrvt}$.  In that case,~\eqref{rtgrowth3} becomes
\begin{equation}
\label{rtgrowth5}
\widetilde{\mathcal{L}}_{\beta_H}(r)=\mu_G \mathcal{M}_{\gtrv}(r)\left\{1+\sum_{n=2}^{n_H} \pi_n \sum_{k=1}^{n-1} \mu_k^{(n)} \left(\mathcal{M}_{{\gtrvt}^{(n)}}(r)\right)^k\right\},
\end{equation}
where ${\gtrvt}^{(n)}$ is a random variable distributed as ${\gtrvt}$ when the household size is $n$, and $\widetilde{R}$ is then obtained as before.

Observe that the approximation $\mathcal{L}_{\beta_H}(r) \approx \widetilde{\mathcal{L}}_{\beta_H}(r)$ is exact if $n_H \le 2$, so in that case $R_r=\widetilde{R}_r$.

\section{Comparisons of households model reproduction numbers}
\label{sec:HComparisons}

We distinguish between an epidemic in which $R_*> 1$ and one in which $R_*<1$; we call the former \emph{growing} (following Goldstein et al.~\cite{GoldsteinEtal2009}) and the latter \emph{declining}.  As stated 
before, we assume implicitly that $n_H \ge 2$, $\mu_H > 0$ and $\mu_G>0$.  We also assume that if $n_H \ge 3$, then $\mu_1 \neq \mu_H$.  Thus we exclude the highly locally infectious case studied by Becker and Dietz~\cite{BecDie1995}, in which the initial infective in a household necessarily infects all other susceptible household members. We comment on this case after Theorems~\ref{Hcomp} and~\ref{propos}.

\subsection{Comparisons not involving $R_r$}
\label{subsec:HcompnotRr}

\subsubsection{Main theorem}
\label{subsubsec:theorem}

The following theorem is proved in Section~\ref{subsec:Hcompproof}.

\begin{thm}\label{Hcomp}
\
\begin{enumerate}
\item[(a)]
$R_*=1 \iff R_I =1 \iff R_0 = 1 \iff R_2=0 \iff R_{H\!I}=1 \implies R_V = 1$.

\item[(b)]
In a growing epidemic,
\[
R_* > R_I \ge R_V \ge R_0 > R_{H\!I} >1\quad \mbox{ and }\quad R_I \ge R_2 > R_{H\!I} > 1,
\]
and in a declining epidemic
\[
R_* < R_I \le R_0 < R_{H\!I}  <1  \quad\mbox{ and }\quad R_I \le R_2 < R_{H\!I} < 1.
\]
The inequalities $R_I \ge R_V$, $R_I \ge R_2$, $R_I \le R_0$ and $R_I \le R_2$ are strict if and only if $n_H>2$.  The inequality $R_V \ge R_0$ is strict if and only if $n_H>3$.
\end{enumerate}
\end{thm}

\begin{rmk}
We conjecture that, in addition to Theorem~\ref{Hcomp}\thp{b}, $R_0 \ge R_2$ in a growing epidemic and $R_0 \le R_2$ in a declining epidemic, with strict inequalities if and only if $n_H>2$, so that Theorem \ref{Hcomp}\thp{b} should take the form $R_* > R_I \ge R_V \ge R_0 \ge R_2 > R_{H\!I} >1$ and $R_* < R_I \le R_0 \le R_2 < R_{H\!I}  <1$ in the two cases, respectively. Although we have yet to find a complete proof, the conjecture is supported by extensive numerical results. We discuss it further in Appendix~\ref{app:R0R2comp}, where it is proved for $n_H \le 3$.
\end{rmk}

\begin{rmk}
If the epidemic is highly locally infectious then $\mu_1=\mu_H$ and it is readily seen that part \thp{a} of Theorem \ref{Hcomp} still holds, $R_*>R_I=R_V=R_0=R_2>R_{H\!I}>1$ in a growing epidemic and $R_*<R_I=R_0=R_2<R_{H\!I}<1$ in a declining epidemic.
\end{rmk}

A key finding of Goldstein et al.~\cite{GoldsteinEtal2009} is that, for a growing epidemic, $R_* \ge R_V \ge \bar{R}_{H\!I}$, thus enabling upper and lower bounds to be obtained for the critical vaccination coverage when individuals are vaccinated uniformly at random with a perfect vaccine (note, though, that $\bar{R}_{H\!I}$ is not a threshold parameter and can be smaller than 1 even in a growing epidemic).  Note that Theorem~\ref{Hcomp} implies that
$R_I$ is a sharper upper bound than $R_*$ for $R_V$ and $R_0$ is a sharper lower bound than $R_{H\!I}$ (which coincides with $\bar{R}_{H\!I}$ when all households have the same size and, as shown below, is greater than or equal to $\bar{R}_{H\!I}$ in a growing epidemic). Goldstein et al.~\cite{GoldsteinEtal2009} show that $R_\ast\ge R_{V\!L} \ge R_V$ for a growing epidemic. We show in Appendix~\ref{app:RIRVLcomp} that $R_{V\!L}$ and $R_{I}$ cannot in general be ordered.

In Appendix~\ref{app:R0RHIcomp} we investigate the possible ordering of variants of $R_{H\!I}$.
Concerning the  reproduction numbers $\bar{R}_{H\!I}$ and $\hat{R}_{H\!I}$, Goldstein et al.~\cite{GoldsteinEtal2009} prove that\footnote{In our notation.} $\bar{R}_{H\!I} \le \hat{R}_{H\!I}$ always holds (see their Proposition A4.1) and we show that $\bar{R}_{H\!I}\le R_{H\!I}$. So we conclude that, by virtue of Theorem \ref{Hcomp}\thp{b}, in a growing epidemic,
\begin{equation*}
R_0  > R_{H\!I} \geq \bar{R}_{H\!I}.
\end{equation*}
(but note that even in a growing epidemic $\bar{R}_{H\!I}$ might or might not be greater than 1). However, in a declining epidemic, $R_0$ and $\bar{R}_{H\!I}$ cannot be ordered in general. Finally, we also construct an example to show that  no general order exists between $R_0$ and $\hat{R}_{H\!I}$ (either in a growing or a declining epidemic).


One may argue that our generalisation of $R_{H\!I}$ to populations with unequal household sizes is more natural than $\bar{R}_{H\!I}$.  Unlike $\bar{R}_{H\!I}$, it is a threshold parameter and, unlike
$\hat{R}_{H\!I}$, it can always be ordered with $R_0$.  Moreover, for a growing epidemic, $R_{H\!I}$ is a sharper lower bound than $\bar{R}_{H\!I}$ for $R_0$, and hence also for
$R_V$ (see Theorem~\ref{Hcomp}). In a similar vein, our generalisation of $R_2$ to populations with unequal household sizes seems more natural than that in \cite{BalSirTra2010}.

\subsubsection{Network-households model}
\label{subsubsec:NHcomp}
We now consider briefly relations among reproduction numbers for the network-households model, for which the calculation of $R_0$ is outlined at the end of Section~\ref{subsec:R0}.  Analogues of $R_*,  R_I, R_V, R_{H\!I}$ and $R_2$ are easily obtained.  Omitting the details, $R_*=\tilde{\mu}_G+\mu_G\mu_H$, $R_I$ is the unique root in $(0,\infty)$ of $g_I^{N\!H}(\lambda)$,
where
\begin{equation*}
g_I^{N\!H}(\lambda)=1- \frac{\tilde{\mu}_G}{\lambda} - \frac{\mu_H \mu_G}{\lambda^2},
\end{equation*}
$R_V$ is defined in the usual way via the (perfect vaccine) critical vaccination coverage,
$R_{H\!I}$ is the  unique root in $(a,\infty)$ of $g_{H\!I}^{N\!H}(\lambda)$, where (cf.~\eqref{gHI})
\begin{equation*}
g_{H\!I}^{N\!H}(\lambda)=1- \frac{\tilde{\mu}_G}{\lambda} - \frac{\mu_G}{\lambda} \sum_{n=2}^{n_H} \pi_n\frac{a^{(n)}}{\lambda-a^{(n)}} \qquad (\lambda > a),
\end{equation*}
and $R_2$ is the  unique root in $(b,\infty)$ of $g_2^{N\!H}(\lambda)$, where (cf.~\eqref{g2unequal})
\begin{equation*}
g_2^{N\!H}(\lambda)=1- \frac{\tilde{\mu}_G}{\lambda} - \frac{\mu_G}{\lambda} \sum_{n=2}^{n_H} \pi_n\frac{\mu_1^{(n)}}{\lambda-b^{(n)}} \qquad (\lambda > b).
\end{equation*}
With the above definitions, Theorem~\ref{Hcomp} holds also for the network-households model; the proof is essentially the same as for the households model and hence omitted.

Analogues of $\hat{R}_{H\!I}$ and $\hat{R}_2$ can also be defined.  On average, a fraction $1/(1+\mu_H)$ of infectives are household primary cases, who each make a mean of $\tilde{\mu}_G$ global contacts, and a  fraction $\mu_H/(1+\mu_H)$ of infectives are household secondary cases, who each make a mean of $\mu_G$ global contacts.  Arguing as in the derivation of~\eqref{rhihatnew} then leads to
\begin{equation*}
\hat{R}_{H\!I}=\frac{\tilde{\mu}_G+\mu_H(1+\mu_G)}{1+\mu_H},
\end{equation*}
but note that $\hat{R}_{H\!I}$ does not necessarily equal $R_{H\!I}$ when all households have the same size.  Arguing as in the derivation of~\eqref{g2} yields that $\hat{R}_2$ is the largest positive root of $\hat{g}_2^{N\!H}(\lambda)$, where
\begin{equation*}
\hat{g}_2^{N\!H}(\lambda)=1- \frac{b+\tilde{\mu}_G}{\lambda} +\frac{b\tilde{\mu}_G-\mu_1\mu_G}{\lambda^2}.
\end{equation*}
The reproductions numbers $\hat{R}_2$ and $R_2$ do coincide when all households have the same size.
Comparisons involving $R_0, \hat{R}_2$ and $\hat{R}_{H\!I}$ are more involved and are not considered here.


\subsubsection{Generational view of comparisons}
\label{subsubsec:genviewcomp}

For the households model, the reproduction numbers $R_0, R_*, R_I, R_{H\!I}$ and $R_2$ can all be obtained by viewing local epidemics on an appropriate generation basis, with any
such reproduction number, $R_A$ say, being given by the unique positive root of the function $g_A$ defined by
\begin{equation}
\label{gA}
g_A(\lambda)=1-\mu_G\sum_{k=0}^{\infty}\frac{\mu_k^A}{\lambda^{k+1}},
\end{equation}
where $\mu_0^A,\mu_1^A,\cdots$ are the mean generation sizes associated with $R_A$,
averaged with respect to the size-biased household size distribution $(\pi_n) = (\pi_n: n=1,2,\cdots,n_H)$.  The mean generations sizes associated with $R_0,R_{H\!I}$ and $R_2$ have been described previously and lead to~\eqref{g0H}, \eqref{gHI} and \eqref{g2unequal}, respectively.  For $R_*$, they are given by $\mu_0^*=1+\mu_H$
and $\mu^*_k=0$ $(k=1,2,\cdots)$, so $g_*(\lambda)=1-\mu_G\frac{1+\mu_H}{\lambda}$, whence $R_*$ is given by \eqref{rstar}. For $R_I$ they are given by $\mu_0^I=1, \mu_1^I=\mu_H$ and $\mu^I_k=0$ $(k=2,3,\cdots)$, leading to \eqref{gind}.

Observe that, for each $A$, $\sum _{k=0}^{\infty}\mu_k^A=1+\mu_H$, so we can define a random variable $X^A$ having probability mass function $\mathbb{P}(X^A=k)=\mu_k^A/(1+\mu_H)$ $(k=0,1,\cdots)$, whose interpretation is the household-generation (associated with $R_A$) of an an infective chosen uniformly at random from all infectives in a household with size chosen according to the size-biased distribution $(\pi_n)$.  Moreover, $R_A$ is then given by the unique solution in $(0,\infty)$ of the equation
\begin{equation}
\label{RAXlambda}
\mathbb{E}\left[\lambda^{-(X^A+1)}\right]=\frac{1}{\mu_G(1+\mu_H)}.
\end{equation}
Now, for $x\ge0$, $\lambda^{-x}$ is increasing in $x$ if $\lambda<1$ and decreasing if
$\lambda>1$.  Thus, if for two reproduction numbers, $R_A$ and $R_B$ say, $X^A \overset{st}{\le} X^B$ ($X^A$ stochastically smaller than $X^B$, i.e.~$\mathbb{P}(X^A\le x ) \ge \mathbb{P}(X^B \le x)$ for all $x\in \mathbb{R}$)
then it follows that $R_A \ge R_B$ in a growing epidemic and $R_A \le R_B$ in a declining epidemic.

The above observation provides an intuitive explanation for all of the comparisons in Theorem~\ref{Hcomp} (except those involving $R_V$) and also for the conjecture concerning $R_0$ and $R_2$.  Indeed, the comparisons in Theorem~\ref{Hcomp} can be proved by showing stochastic ordering of the associated $X^A$s,
though this approach is generally no easier and sometimes harder than the proofs in Section~\ref{subsec:Hcompproof}.

The above approach provides a simple proof of comparisons of $R_0^{\mbox{r}}$ and $R_0^{\mbox{g}}$, where $R_0^{\mbox{r}}$ and $R_0^{\mbox{g}}$ denote the values of $R_0$ obtained using rank and true generations, respectively (see Section \ref{Hmodel}). Suppose first that all households have size $n$.  Let $\mu_0^{\mbox{r}}, \mu_1^{\mbox{r}},\cdots, \mu_{n-1}^{\mbox{r}}$ and $\mu_0^{\mbox{g}}, \mu_1^{\mbox{g}},\cdots, \mu_{n-1}^{\mbox{g}}$ denote the mean rank and mean true generation sizes, respectively, and let $X^{\mbox{r}}$ and
$X^{\mbox{g}}$ denote the corresponding induced generation random variables.  Consider a realisation of the augmented version of the random graph $\mathcal{G}^{(n)}$ defined in Section~\ref{Hmodel}. If $n\le2$ then the rank and true generation numbers coincide for all infectives.  Suppose that $n \ge 3$ and for any infective, $i$ say, let $r_i$ and $g_i$ denote its rank and true generation numbers, respectively.  Then $r_i \le g_i$, since $r_i$ is the number of edges in the \emph{shortest} chain joining the initial infective $0$ to $i$.  However, if there is a chain joining $0$ to $i$ having strictly more edges than $r_i$ but strictly less total time than any such chain of length $r_i$ then $g_i > r_i$. It follows that, for $n \ge 3$, $\sum_{j=0}^k \mu_j^{\mbox{r}} \ge \sum_{j=0}^k \mu_j^{\mbox{g}}$ $(k=0,1,\cdots,n-1)$,
which implies that $X^{\mbox{r}} \overset{st}{\le} X^{\mbox{g}}$.  Taking expectations with respect to the size-biased household size distribution $(\pi_n)$ shows that the same result holds for populations
with unequal household sizes, provided $n_H \ge 3$. Hence in a growing epidemic
$R_0^{\mbox{r}}  \geq R_0^{\mbox{g}}$, whilst in a declining epidemic
$R_0^{\mbox{r}}  \leq R_0^{\mbox{g}}$.

\subsection {Comparisons involving $R_r$}
\label{subsubsec:Rrcomp}

Although $R_r$ and $R_0$ cannot in general be ordered (see Appendix~\ref{app:RandomTVI}), Theorem~\ref{propos} below (proved in Section~\ref{subsec:HRrproof}) shows that, for the most commonly-studied models in the literature, in a growing epidemic $R_r \ge R_0$ and in a declining epidemic $R_r \le R_0$. For this purpose, it is convenient to consider two broad classes of models. The first class contains those models for which $\mathcal{I}(t)=J {\gtpdf}(t)$  for all $t\ge 0$, for which the shape of the infectivity profile is not random, but the magnitude $J$ is.  (Recalling that $\int_0^{\infty} {\gtpdf}(t){\rm d}t=1$, we have that $J=\int_0^{\infty} \mathcal{I}(t){\rm d}t$ and $\mathbb{E}[J]=1$.)
Another class assumes that the duration of the infectious period is random but, conditioned on an individual being still infectious $t$ time units after being infected, the infectivity is non-random, i.e., $\mathcal{I}(t) =f(t) \indic(T_I>t)$ for $t\geq 0$, where $f(t)$ is a deterministic function and $T_I$ is a random variable denoting the infectious period, which satisfy $\int_0^{\infty} f(t)\mathbb{P}(T_I>t) {\rm d}t=1$. (Throughout the paper, for an event, $\ev$ say, $\indic(\ev)$ denotes its indicator function; i.e.~$\indic(\ev)=1$ if the event $\ev$ occurs and $\indic(\ev)=0$ if $\ev$ does not occur.  Thus, in the present setting, $\mathcal{I}(t) =f(t)$ if $t<T_I$ and $\mathcal{I}(t) = 0$ if $t\ge T_I$. ) Note that the standard stochastic SIR model (Andersson and Britton \cite{AndBri2000}, Chapter 2) is in this class ($f(t)$ is constant).
A non-random time-varying infectivity profile, i.e.~$\mathcal{I}(t) = {\gtpdf}(t)$ for all $t\ge 0$ is a special case of both classes.

\begin{thm}
\label{propos}
\begin{itemize}
\item[(a)] For all choices of infectivity profile $\mathcal{I}(t)$  $(t \geq 0)$,
\[
R_r=1 \iff \widetilde{R}_r=1 \iff R_0 = 1.
\]
\item[(b)] If $\mathcal{I}(t) = J w(t)$ ($t \geq 0$), where $J$ is a non-negative random variable,
then in a growing epidemic,
\[
R_r \ge R_0 >1,
\]
and in a declining epidemic,
\[
R_r \le R_0 <1.
\]
\item[(c)] If $\mathcal{I}(t) =f(t) \indic(T_I>t)$ ($t \geq 0$), where $f(t)$ is a deterministic function and $T_I$ a non-negative random variable, then in a growing epidemic,
\[
R_r \ge \widetilde{R}_r \ge R_0 >1,
\]
and in a declining epidemic,
\[
R_r \le \widetilde{R}_r \le R_0 <1.
\]
\end{itemize}
The above results still hold if a latent period independent of the remainder of the infectivity profile is added.
\end{thm}

\begin{rmk}
Note that for Reed-Frost type models (i.e.~models in which the latent period is constant and the infectious period is reduced to a single point in time, cf. Bailey \cite{Bailey1975}, Section 14.2, and Diekmann et al.~\cite{DiekEtal2013}, Section 3.2.1), the approximation $\mathcal{L}_{\beta_H}(r) \approx \mathcal{L}_{\beta_H}^{(0)}(r)$ (see \eqref{rtgrowth2}) is exact, as all infectious intervals equal the constant latent period, so $R_0 =R_r$. Thus, it is not generally possible to obtain strict inequalities in Theorem~\ref{propos}\thp{b}. However, if ${\gtpdf}(t)$ is a proper density function, i.e.~${\gtpdf}(t)< \infty$ for all $t \geq 0$, then the inequalities are strict (recall that we have assumed that not all households have size 1).  As noted in Section~\ref{subsec:Rr}, $R_r=\widetilde{R}_r$ if $n_H \le 2$.  See Remark~\ref{remarkstrict} after the proof of Theorem~\ref{propos} in Section~\ref{subsec:HRrproof} for further details.
\end{rmk}
\begin{rmk}
The proof of Theorem \ref{propos} also suggests how to construct the counterexample presented in Appendix~\ref{app:RandomTVI}, which gives a model (not belonging to either of the classes considered in Theorem \protect\ref{propos}) for which $R_r < R_0$ in a growing epidemic.
\end{rmk}

\begin{rmk}
Suppose that all secondary infections take place as soon as the primary individual in a household is infected. Then $\beta_H(t) = (1+\mu_H) \mu_G \gtpdf(t)=R_* \gtpdf(t)$ $(t \geq 0)$. Hence, $\mathcal{L}_{\beta_H}(r) = R_* \mathcal{M}_{\gtrv}(r)$ and it follows from \eqref{Rrfromr} and \eqref{rtgrowth1} that $R_r = R_*$.  This happens in the highly locally infectious limit $\lambda_H^{(n)} \to \infty$ $(n=2,3,\cdots)$, provided $\gtrv$ has mass arbitrarily close to zero, i.e.~provided $\inf\{t > 0:\gtpdf(t)>0\}=0$.
\end{rmk}

For a growing epidemic, Goldstein et al.~\cite{GoldsteinEtal2009} prove that $R_* \ge R_r$.  They also note that in most numerical simulations $R_{V\!L} \ge R_r \ge R_V$, though they show that the second inequality can be violated if the latent period is very large and they do not have a proof for the first inequality.  The first inequality held in all of their numerical simulations but the question whether or not the result holds in general was left open. In Appendix~\ref{app:RrRVLcomp} we show that $R_r$ and $R_{V\!L}$ cannot in general be ordered.  In their numerical simulations for a households SEIR model with exponentially distributed infectious and latent periods, Goldstein et al.~\cite{GoldsteinEtal2009} noted that $R_r$ can be less than $R_V$ when the mean latent period is very long and, in Appendix B of their paper, they give a mathematical explanation of that observation.  However, their proof assumes a constant latent period and does not hold for the model with exponentially distributed latent periods.  This is discussed further in Appendix~\ref{app:inflonglat}; see also the numerical example in Section~\ref{MarkovSIRHmod}.

Finally, although Goldstein et al.~\cite{GoldsteinEtal2009} consider only the growing epidemic case, it is easy to see (from (6.2.2) and Lemma 6.2.1 of their paper) that the same argument they use to prove $R_\ast\ge R_r$ leads, in a declining epidemic, to $R_\ast\le R_r$.

\section{Households-workplaces model and reproduction numbers}
\label{sec:housework}

\subsection{Model and generations of infections}
\label{sec:HWmod}
In this model each individual belongs to a household and to a workplace, and infectives make three types of contacts: global contacts, with individuals chosen uniformly at random from the entire population; household contacts, with individuals in the infective's own household;
and workplace contacts, with individuals in the infective's own workplace. In order to make branching process approximations for the early stages of the epidemic, and thus define threshold parameters, it is necessary to assume that, as the population size tends to infinity, the only short cycles of local contacts (see below) that can occur with non-zero probability are either within the same household or within the same workplace, which implies that a household and a workplace cannot share more than one person; see Ball and Neal~\cite{BalNea2002} and Pellis et al.~\cite{PelFerFra2009,PelBalTra2012} for further detail.

The mean number of global contacts made by a typical infective is $\mu_G$.  Household and workplace contacts are called local contacts.  As with the households model, we do not specify the full detail of local infection transmission, but we do assume that the spread within a household and the spread within a workplace can each be described in terms of generations of infection.
Let $n_H$ and $n_W$ denote respectively the sizes of the largest household and the largest workplace in the population.  Then, for $\ell=0,1,\cdots,n_H-1$, let $\mu_\ell^H$ be the mean size of the $\ell$th generation in a typical single-household epidemic with $1$ primary case and, for $\ell'=0,1,\cdots,n_W-1$, define $\mu_{\ell'}^W$ similarly for a typical single-workplace epidemic.
By a typical single-household (workplace) epidemic we mean one in which the primary case is obtained by choosing an individual uniformly at random from the entire population, so $\mu_\ell^H$
is household size-biased, as at (\ref{sizebiasgen}), and $\mu_{\ell'}^W$ is size-biased using the workplace size-biased distribution corresponding to \eqref{sizebias}. We also assume that the sizes of any given individual's household and workplace are asymptotically independent as the population size tends to infinity.

Let $\mu_H=\mu_1^H+\mu_2^H+\cdots+\mu_{n_H-1}^H$ be the mean size of a typical single-household epidemic, not including the primary case, and define $\mu_W$ similarly for a typical single-workplace epidemic.  We assume that $\mu_H>0$ and $\mu_W>0$, and that the population contains households and workplaces of size at least two.  If any of these conditions fails to hold then the model effectively reduces to the households model. For simplicity we assume that $\mu_G>0$.  We comment on the case $\mu_G=0$ at the end of Section~\ref{sec:HWComparisons}.

\subsection{The basic reproduction number $R_0$}
\label{subsec:HWR0}

The basic reproduction number $R_0$ for the households-workplaces model may be obtained by considering the following $3$-type branching process, which approximates the process of infectives in the epidemic model.  The three types of individual in the branching process are double-primary cases (type $1$), household-primary cases (type $2$) and workplace-primary cases (type $3$), which correspond to cases who are infected by a global contact, a workplace contact and a household contact, respectively.  In the branching process, the mother of a double-primary case is the person who infected it in the epidemic process, the mother of a household-primary case is the primary case in the corresponding single-workplace epidemic and the mother of a workplace-primary case is the primary case in the corresponding single-household epidemic.  Time in the branching process corresponds to generation number in the epidemic at large.  Thus, in the branching process, a typical double-primary case spawns on average $\mu_G$ double-primary cases at age $1$, $\mu_{\ell'}^W$ household-primary cases at age $\ell'$ ($\ell'=1,2,\cdots,n_W-1$) and $\mu_\ell^H$ workplace-primary cases at age $\ell$ ($\ell=1,2,\cdots,n_H-1$);
a typical household-primary case spawns on average $\mu_G$ double-primary cases at age $1$ and $\mu_\ell^H$ workplace-primary cases at age $\ell$ ($\ell=1,2,\cdots,n_H-1$); and a typical workplace-primary case spawns on average $\mu_G$ double-primary cases at age $1$ and $\mu_{\ell'}^W$ household-primary cases at age $\ell'$ ($\ell'=1,2,\cdots,n_W-1$).  The total number of individuals at time $k$ in this branching process corresponds to the total number of infectives in global generation $k$ in the epidemic process, so $R_0$ is given by the asymptotic geometric growth rate of this branching process.

It is convenient to introduce the following notation for future reference. For $d,d'=1,2,3$ and $k=0,1,\cdots$, let $\nu^{(dd')}_k$ be the mean number of type-$d'$ individuals spawned by a typical type-$d$ individual at age $k+1$ and, for $\lambda \in (0,\infty)$, let $\nu_{dd'}(\lambda)=\sum_{k=0}^{\infty} \nu^{(dd')}_k/\lambda^{k+1}$.  By the theory of multi-type general branching processes (see, for example, Haccou et al.~\cite{Haccou2005}, Section 3.3.2, and Jagers~\cite{Jagers1989}), the asymptotic geometric growth rate of the branching process, and hence also $R_0$, is given by the value of $\lambda$ such that the dominant eigenvalue of the matrix
\begin{equation}
\label{AHWmatrix}
A^{(HW)}(\lambda) = \left[ \nu_{dd'} (\lambda) \right ] = \begin{bmatrix}
\ds \frac{\mu_G}{\lambda} && \ds\sum_{\ell'=1}^{n_W-1}\dfrac{\mu_{\ell'}^W}{\lambda^{\ell'}} && \ds\sum_{\ell=1}^{n_H-1}\dfrac{\mu_\ell^H}{\lambda^\ell} \\
\ds\frac{\mu_G}{\lambda} && 0 && \ds\sum_{\ell=1}^{n_H-1}\dfrac{\mu_\ell^H}{\lambda^\ell} \\
\ds\frac{\mu_G}{\lambda} && \ds\sum_{\ell'=1}^{n_W-1}\dfrac{\mu_{\ell'}^W}{\lambda^{\ell'}} && 0
\end{bmatrix}
\end{equation}
is 1. Letting $A=\mu_G/\lambda$, $B=\sum_{{\ell'}=1}^{n_W-1}\mu_{\ell'}^W/\lambda^{\ell'}$ and $C=\sum_{\ell=1}^{n_H-1}\mu_\ell^H/\lambda^\ell$, the characteristic polynomial of $A^{(HW)}(\lambda)$ is
\begin{equation}
\label{charpolyAnu}
f(x)=x^3-Ax^2-(AB+AC+BC)x-ABC,
\end{equation}
which has a unique positive root.  Thus, since the matrix $A^{(HW)}(\lambda)$ is non-negative, its dominant eigenvalue is 1 if and only if $f(1)=0$.

Now
\begin{align*}
f(1)=0 &\iff ABC+AB+AC+BC+A-1=0 \\
&\iff A(B+1)(C+1)+BC-1=0.
\end{align*}
Further,
\begin{equation*}
BC=\left(\sum_{\ell=1}^{n_H-1}\frac{\mu_\ell^H}{\lambda^\ell}\right)\left(\sum_{{\ell'}=1}^{n_W-1}\frac{\mu_{\ell'}^W}{\lambda^{\ell'}}\right)=\sum_{k=1}^{n_H+n_W-3} 
\left(\sum_{\ell=\max(1,k-n_W+2)}^{\min(k,n_H-1)}\frac{\mu_\ell^H\mu_{k+1-\ell}^W}{\lambda^{k+1}}\right)
\end{equation*}
and, recalling that $\mu_0^H=\mu_0^W=1$,
\begin{align*}
A(B+1)(C+1) &=\frac{\mu_G}{\lambda}\left(\sum_{\ell=0}^{n_H-1}\frac{\mu_{\ell}^H}{\lambda^\ell}\right)\left(\sum_{{\ell'}=0}^{n_W-1}\frac{\mu_{\ell'}^W}{\lambda^{\ell'}}\right) \\
&= \mu_G \sum_{k=0}^{n_H+n_W-2} 
\left(\sum_{\ell=\max(0,k-n_W+1)}^{\min(k,n_H-1)}\frac{\mu_\ell^H\mu_{k-\ell}^W}{\lambda^{k+1}}\right).
\end{align*}
Thus, the dominant eigenvalue of $A^{(HW)}(\lambda)$ is 1 if and only if $g^{(HW)}_{0} (\lambda )=0$, where
\begin{equation}
\label{gnHnW}
g^{(HW)}_{0} (\lambda )=1-\sum _{k=0}^{n_{H}+n_W -2}\frac{c_{k} }{\lambda ^{k+1} },
\end{equation}
with $c_0=\mu_G$ and, for $k=1,2,\cdots,n_H+n_W-2$,
\begin{equation}
c_k=\mu_G\sum_{\ell=\max(0,k-n_W+1)}^{\min(k,n_H-1)}\mu_\ell^H\mu_{k-\ell}^W
+\sum_{\ell=\max(1,k-n_W+2)}^{\min(k,n_H-1)} \mu_\ell^H\mu_{k+1-\ell}^W,
 \label{eq:coefficientsHW}
\end{equation}
where the second sum in~\eqref{eq:coefficientsHW} is zero when $k=n_H+n_W-2$.

It follows that $R_0$ is given by the unique positive root of $g^{(HW)}_{0}$, giving a new (and simpler) proof of Pellis et al.~\cite{PelBalTra2012}, Corollary 2.

\subsection{The clump reproduction number $R_*$}
\label{subsec:HWR*}

The first reproduction number proposed for the households-workplaces model was the reproduction number for the proliferation of local infectious clumps, denoted by $R_*$; see Ball and Neal~\cite{BalNea2002}.  A local infectious clump is the set of individuals infected by chains of local infections (i.e.~through households and workplaces) from a typical single initial infective in an otherwise fully susceptible population.  In the early stages of an epidemic, initiated by few infectives in a large population, such clumps (which are initiated by global contacts) intersect with small probability, unless the local epidemic is itself supercritical.  The clump reproduction number $R_*$ is the expected number of clumps generated by a typical clump and is given by
\begin{equation}
\label{hwrstar}
R_*=\left\{
\begin{array}{ll}
\frac{\mu_G(1+\mu_H)(1+\mu_W)}{1-\mu_H \mu_W}\qquad&\text{if } \mu_H \mu_W < 1,\\
\infty &\text{otherwise.}
\end{array} \right.
\end{equation}
Note that setting $\mu_W=0$ in~\eqref{hwrstar} yields~\eqref{rstar}; when $\mu_W=0$, the model reduces to the households model and a typical local infectious clump becomes the set of people infected in a typical single-household epidemic.

\subsection{The household-household and workplace-workplace reproduction numbers $R_H$ and $R_W$}
\label{subsec:RHRW}

The clump reproduction number $R_*$ has a number of disadvantages, as pointed out by Pellis et al.~\cite{PelFerFra2009}.  In particular, it can be infinite and the time for a clump to form increases as $\mu_H \mu_W$ tends to 1 and becomes comparable with the time of the entire epidemic.  Thus
a household-to-household reproduction number $R_H$, defined as the expected number of households infected by a typical infected household in an otherwise totally susceptible population, was introduced in \cite{PelFerFra2009}.  A household may be infected either globally (i.e.~via a global contact) or locally (i.e.~via a contact within a workplace).  It follows
(see~\cite{PelFerFra2009} for details) that $R_H$ is the largest eigenvalue of the household next-generation matrix
\begin{equation}
\label{HWM_H}
M_H = \left[ \begin{array}{cc}\mu_G(1+\mu_H)&\mu_W(1+\mu_H)\\\mu_G(1+\mu_H)&\mu_H \mu_W\end{array}\right] ,
\end{equation}
whence $R_H$ is given by the unique solution in $(0, \infty)$ of $g_H(\lambda)=0$, where
\begin{equation}
\label{gh}
g_H(\lambda)=1-\frac{\mu_G+\mu_G \mu_H+\mu_H\mu_W}{\lambda}-\frac{\mu_G\mu_W(1+\mu_H)}{\lambda^2}.
\end{equation}
A workplace-to-workplace reproduction number $R_W$ can be defined in a similar fashion.  

\subsection{The individual reproduction number $R_I$}
\label{subsec:HWRI}

An individual-based reproduction number $R_I$ can also be defined (see Pellis et al.~\cite{PelFerFra2009}, supplementary material), as for the households model, by attributing all secondary cases in a household or workplace to the corresponding primary case, leading to the next-generation matrix
\[
M_I^{(HW)}=\left[ \begin{array}{ccc}\mu_G&\mu_H&\mu_W\\
\mu_G&0&\mu_W\\
\mu_G&\mu_H&0\end{array}\right] .
\]
Calculating the characteristic polynomial of $M_I^{(HW)}$ shows that $R_I$ is given by the unique solution in $(0,\infty)$ of $g_I^{(HW)}(\lambda)=0$, where
\begin{equation}
\label{ghwi}
g_I^{(HW)}(\lambda)=1-\frac{\mu_G}{\lambda}-\frac{\mu_G\mu_H+\mu_G\mu_W+\mu_H\mu_W}{\lambda^2}
-\frac{\mu_G\mu_H\mu_W}{\lambda^3}.
\end{equation}

\subsection{The perfect and leaky vaccine-associated reproduction numbers $R_V$ and $R_{V\!L}$}
\label{subsec:HWRV}

The perfect and leaky vaccine-associated reproduction numbers, $R_V$ and
$R_{V\!L}$,  can be defined in an analogous fashion as for the
households model at (\ref{rv}) and (\ref{rvl}), respectively, though we
consider only the former in the comparisons in Section~\ref{subsec:HcompnotRr}.

\subsection{The exponential-growth-associated reproduction number $R_r$ }
\label{subsec:HWRr}

An exponential-growth-associated reproduction number $R_r$ can be defined in a similar vein as for the households model as follows.  Consider the $3$-type branching process used in Section~\ref{subsec:HWR0} to derive $R_0$, but run in real time rather than in generations. Let $r$ be
the Malthusian parameter (real-time growth rate) of this branching process.  Then $R_r$ is defined as at~\eqref{Rrfromr} for the households model.

To determine $R_r$, a more-detailed description of the households-workplaces model is required, which we now give.
As in the households model, suppose that infectives have independent infectivity profiles, each distributed as $\mathcal{I}(t)$ $(t \ge 0)$. Again $\mathcal{I}(t)$ is normalised so that $\mathbb{E}\left[\int_0^{\infty}\mathcal{I}(t) {\rm d}t\right] =1$. If an infective has infectivity profile $\mathcal{I}(t)$ $(t \ge 0)$, then $t$ time units after infection he/she makes global infectious contacts at overall rate $\mu_G \mathcal{I}(t)$, infectious contacts to any given member of his/her household at rate $\lambda_H^{(n)}I(t)$ and to any given member of his/her workplace at rate $\lambda_W^{(n')}I(t)$, where $n$ and $n'$ are the sizes of the infective's household and workplace, respectively.  As previously, let $\gtpdf(t)=\mathbb{E}[\mathcal{I}(t)]$ $(t \ge 0)$ and recall that $\gtpdf$ is the probability density function of a random variable $\gtrv$ having moment-generating function $\mathcal{M}_{\gtrv}(\theta)$.

Consider a typical single-household epidemic with one initial infective, who becomes infected at time $t=0$.  For $t\ge0$, let $\xi_H(t)$ be the rate at which new infections occur in that single-household epidemic at time $t$.  Define $\xi_W(t)$ $(t \ge 0)$ similarly for a typical single-workplace epidemic.

Recall the 3-type real-time branching process introduced in Section~\ref{subsec:HWR0}.  For $t\ge0$, let $M(t)=[m_{dd'}(t)]$, where $m_{dd'}(t)$ is the mean rate at which a type-$d$ individual having age $t$ spawns type-$d'$ individuals ($d,d'=1,2,3$).
Then
\begin{equation*}
M(t)=\left[ \begin{array}{ccc}\mu_G \gtpdf(t)&\xi_W(t) &\xi_H(t)\\
\mu_G \gtpdf(t)&0&\xi_H(t)\\
\mu_G \gtpdf(t)&\xi_W(t) &0\end{array}\right] .
\end{equation*}

For $r \in(-\infty,\infty)$, let $\mathcal{L}_M(r)=\int_0^{\infty} M(t)\mathrm{e}^{-\theta t}\mathrm{d}t$, where the integration is elementwise.  Then the real-time growth rate $r$ is given by the unique real value of $r$ such that the dominant eigenvalue of $\mathcal{L}_M(r)$ is one.

Observe that the matrix $\mathcal{L}_M(r)$ has the same structure of non-zero elements as the matrix $A^{(HW)}(\lambda)$ defined at~\eqref{AHWmatrix}.  The same argument as used in Section~\ref{subsec:HWR0} shows that $r$ is the unique real solution of the equation
\begin{equation}
\label{HWrdefeq}
\mu_G \mathcal{M}_{\gtrv}(r) (\mathcal{L}_{\xi_H}(r)+1)(\mathcal{L}_{\xi_W}(r)+1)+\mathcal{L}_{\xi_H}(r)\mathcal{L}_{\xi_W}(r)=1.
\end{equation}

Pellis et al.~\cite{PelFerFra2011} determine the real-time growth rate of the households-workplaces model by using a two-type branching process having mean offspring matrix $M_H$ (used at~\eqref{HWM_H} to define $R_H$) but again run in real time, which of course gives the same result.  We use the above $3$-type branching process to facilitate comparison of $R_r$ with $R_0$.

Similar to the households model, the difficulty in using~\eqref{HWrdefeq} to calculate $r$ is that generally there is no tractable expression for $\mathcal{L}_{\xi_H}(r)$ or $\mathcal{L}_{\xi_W}(r)$.  However, we can use similar approximations to those used in Section~\ref{subsec:Rr} for the households model.  First, it is easily verified that using the approximations (cf.~\eqref{rtgrowth2})
$\mathcal{L}_{\xi_H}(r)\approx \mathcal{L}_{\xi_H}^{(0)}(r)$ and $\mathcal{L}_{\xi_W}(r)\approx \mathcal{L}_{\xi_W}^{(0)}(r)$, where
\begin{equation}
\label{LxiH0}
\mathcal{L}_{\xi_H}^{(0)}(r)= \sum_{\ell=1}^{n_H-1} \mu_\ell^H \left(\mathcal{M}_\gtrv(r)\right)^\ell
\end{equation}
and
\begin{equation}
\label{LxiW0}
\mathcal{L}_{\xi_W}^{(0)}(r)=\sum_{{\ell'}=1}^{n_W-1} \mu_{\ell'}^W \left(\mathcal{M}_\gtrv(r)\right)^{\ell'},
\end{equation}
leads to $R_r$ being given by $R_0$.


Second, for $n=2,3,\cdots,n_H$, let $\gtrvt^{(n,H)}$ be a random variable describing the time of the first within-household
contact from one specified individual to another specified individual in a household of size $n$ and, for $n'=2,3,\cdots,n_W$, define
$\gtrvt^{(n',W)}$ analogously for a workplace contact.  Also, for $n=1,2,\cdots,n_H$, let $\pi_n^H$ be the (size-biased) probability an individual chosen uniformly at random from the population resides in a household of size $n$ and, for $n'=1,2,\cdots,n_W$, let $\pi_n^W$ be the corresponding workplace size-biased probability.
Then (cf.~\eqref{rtgrowth3} and~\eqref{rtgrowth5}), we have
$\mathcal{L}_{\xi_H}(r)\approx \widetilde{\mathcal{L}}_{\xi_H}(r)$ and $\mathcal{L}_{\xi_W}(r)\approx \widetilde{\mathcal{L}}_{\xi_W}(r)$, where
\begin{equation}
\label{wtildeHapprox}
\widetilde{\mathcal{L}}_{\xi_H}(r)=\sum_{n=2}^{n_H} \pi_n^H \sum_{\ell=1}^{n-1} \mu_\ell^{(n,H)} \left(\mathcal{M}_{\gtrvt^{(n,H)}}(r)\right)^\ell
\end{equation}
and
\begin{equation}
\label{wtildeWapprox}
\widetilde{\mathcal{L}}_{\xi_W}(r)= \sum_{n'=2}^{n_W} \pi_{n'}^W \sum_{{\ell'}=1}^{n-1} \mu_{\ell'}^{(n',W)} \left(\mathcal{M}_{\gtrvt^{(n',W)}}(r)\right)^{\ell'},
\end{equation}
and $\mu_\ell^{(n,H)}$ and $\mu_{\ell'}^{({n'},W)}$ denote the mean size of generation $\ell$ in a single size-$n$ household epidemic
and generation $\ell'$ in a single size-$n'$ workplace epidemic, respectively. Substituting the approximations~\eqref{wtildeHapprox} and~\eqref{wtildeWapprox}
into~\eqref{HWrdefeq} and solving for $r$ yields an approximation, $\tilde{r}$ say, to the growth rate of the households-workplaces model.  The reproduction number $\widetilde{R}_r$ is then defined as at~\eqref{Rtilderdef} for the households model.

Note that if $\lambda_H^{(n)}$ and $\lambda_W^{(n')}$ are independent of household size $n$ and workplace size $n'$, respectively,
then, in an obvious notation,~\eqref{wtildeHapprox} and~\eqref{wtildeWapprox} simplify to
\begin{equation}
\label{wtildeHapprox1}
\widetilde{\mathcal{L}}_{\xi_H}(r)= \sum_{\ell=1}^{n_H-1} \mu_\ell^H \left(\mathcal{M}_{\gtrvt^{H}}(r)\right)^\ell
\mbox{ and }
\widetilde{\mathcal{L}}_{\xi_W}(r)= \sum_{{\ell'}=1}^{n_W-1} \mu_{\ell'}^W \left(\mathcal{M}_{\gtrvt^{W}}(r)\right)^{\ell'}.
\end{equation}
The approximations $\mathcal{L}_{\beta_H}(r)\approx \widetilde{\mathcal{L}}_{\beta_H}(r)$ and $\mathcal{L}_{\beta_W}(r)\approx \widetilde{\mathcal{L}}_{\beta_W}(r)$ are 
exact if $\max(n_H,n_W)\le 2$, so in that case  $R_r=\widetilde{R}_r$.

\section{Comparisons of household-workplaces model reproduction numbers}
\label{sec:HWComparisons}

As stated at the end of Section~\ref{sec:HWmod}, we assume that $\mu_G$, $\mu_H$ and $\mu_W$ are all strictly positive, and that
$\min(n_H,n_W) \ge 2$.  By interchanging households and workplaces, $R_W$ and $R_H$ relate in a similar fashion to the other reproduction numbers, so we do not consider $R_W$ in the comparisons.  As with the households model, an epidemic is called growing if $R_* > 1$ and declining if $R_* < 1$.

The following theorem is proved in Section~\ref{subsec:HWcompproof}.
\begin{thm}

\label{Theorem 5.2}
\begin{enumerate}
\item[(a)]
$R_*=1 \iff R_H =1 \iff R_I=1 \iff R_0 = 1 \Longrightarrow R_V = 1 $.

\item[(b)]
In a growing epidemic,
\[
R_* > R_H > R_I \ge R_V \ge R_0 > 1,
\]
and in a declining epidemic
\[
R_* < R_H < R_I \le R_0 < 1.
\]
The inequalities $R_I \ge R_V$ and $R_I \le R_0$ are strict if and only if $\max(n_H,n_W)>2$.  The inequality $R_V \ge R_0$ is
strict if and only if $\max(n_H,n_W)>3$.

\item[(c)] Theorem~\ref{propos} holds also for the households-workplaces model.
\end{enumerate}
\end{thm}

The main practical use of Theorem~\ref{Theorem 5.2} is that, as for the households model, $R_I \ge R_V \ge R_0$ for a growing epidemic.  Thus, with a perfect vaccine, if individuals are selected for vaccination uniformly at random then the critical vaccination coverage $p_C$, assuming a growing epidemic, satisfies $$1-1/R_0 \le p_C \le 1-\/R_I.$$


Finally, consider the case $\mu_G=0$.  The reproduction numbers $R_0$, $R_H$, $R_W$, $R_I$, $R_r$ and $\widetilde{R}_r$ can all be defined essentially as before but note, for example, that the branching process underlying $R_0$ is now $2$-type, rather than $3$-type, since double-primary cases no longer occur (apart from in global generation $0$, i.e.\ the initial infectives in the epidemic at large).  The reproduction number $R_*=0$, since clumps no longer reproduce, though (cf.~Section~\ref{subsec:RHRW}) a clump may be infinite in size. It is easily seen that Theorem~\ref{Theorem 5.2}, with $R_*$ removed, continues to hold when $\mu_G=0$, as does the generalisation of
Theorem~\ref{propos} to the households-workplaces model.

\section{Numerical illustrations}
\label{sec:numerical}

	In this section we present some numerical examples which illustrate the inequalities between reproduction numbers considered in the paper.  Most of these reproduction numbers  are fairly straightforward to compute for a wide range of modelling assumptions.  This is not the case for the exponential-growth-associated reproduction number $R_r$, which generally cannot be computed explicitly. A notable exception is if the underlying epidemic model is Markovian and therefore most of our numerical examples are for such models.  The main practical interest in these illustrations is how well the various reproduction  numbers approximate the perfect-vaccine-associated reproduction number $R_V$.

\subsection{Markov SIR and SEIR households models}
\label{MarkovSIRHmod}
We consider the model introduced by Ball et al.~\cite{BalMolSca1997}, Section 3.1, specialised to exponential infectious periods.  Thus we assume that all households have common size $n$, that the total population size is $N$ and that
the infectious period of an infective has an exponential distribution having mean one.  (The unit of time may be chosen to be the mean of the infectious period.)  During his/her infectious period, a given infective makes global contact with any given susceptible in the population at the points of a homogeneous Poisson process having rate $\mu_G/N$ and, additionally, local contacts with any given susceptible in his/her household at the points of a homogeneous Poisson process having rate $\lambda_H$.  All the Poisson processes describing infectious contacts (whether or not either or both of the individuals involved are the same) and all the infectious periods are assumed to be independent.
There is no latent period, so a susceptible becomes an infective as soon as he/she is contacted by an infective. Denote this epidemic model by $\mathscr{E}^H(n,\mu_G,\lambda_H)$.

We now describe briefly the calculation of the various reproduction numbers for this model.  The mean generation sizes $\mu_1^{(n)}, \mu_2^{(n)}, \cdots, \mu_{n-1}^{(n)}$ for a single size-$n$ household epidemic may be computed using the method described by Pellis et al.~\cite{PelBalTra2012}, Appendix A, thus enabling $R_0$ to be calculated.  
The reproduction numbers $R_*, R_I, R_2$ and $R_{H\!I}$ are then easily calculated, since $\mu_H^{(n)}=\mu_1^{(n)}+\mu_2^{(n)}+\cdots+\mu_{n-1}^{(n)}$.  Alternatively, $\mu_H^{(n)}$ may be computed more directly using Ball~\cite{Ball1986}, equations (2.25) and (2.26).

The perfect-vaccine-associated reproduction number $R_V$ is computed as follows.  Suppose that a fraction $p$ of the population is vaccinated, with individuals selected for vaccination uniformly at random from the population.  After vaccination, the probability that a global contact is successful (i.e.~is with an unvaccinated individual) is $1-p$, so the mean number of global contacts made by an infective is $(1-p)\mu_G$.  If a global contact is successful then the number of other unvaccinated individuals in the globally contacted individual's household follows a binomial distribution, whence the expected number of households infected by a typical infected household in an otherwise uninfected population, $R_*^V(p)$ say, is given by
\begin{equation*}
R_*^V(p)=(1-p)\mu_G\left(1+\sum_{j=1}^{n-1} {{n-1} \choose j} (1-p)^j p^{n-1-j} \mu_H^{(j+1)}\right).
\end{equation*}
The corresponding critical vaccination coverage $p_C$ is found by solving $R_*^V(p)=1$ numerically and $R_V$ then follows using~\eqref{rv}.  The leaky-vaccine-associated reproduction number $R_{V\!L}$ may be computed by noting that if the entire population is vaccinated with a leaky vaccine having efficacy $\mathcal{E}$ then after vaccination the model behaves as $
\mathscr{E}^H(n,(1-\mathcal{E})\mu_G,(1-\mathcal{E})\lambda_H)$, so a post-vaccination reproduction number, $R_*^{V\!L}(\mathcal{E})$ say, is easily calculated.  The critical efficacy $\mathcal{E}_C$ is found by solving $R_*^{V\!L}(\mathcal{E})=1$ numerically and $R_{V\!L}$ is then given by~\eqref{rvl}.

Turning to the exponential-growth-associated reproduction number $R_r$ and its approximation $\widetilde{R}_r$, the real-time growth rate $r$ for the Markov SIR households model $\mathscr{E}^H(n,\mu_G,\lambda_H)$ may be computed
using the matrix method described in Pellis et al.~\cite{PelFerFra2011}, Section 4.2.  The infectivity profile of a typical infective in $\mathscr{E}^H(n,\mu_G,\lambda_H)$ is given by~\eqref{IGSE},
with $T_I \sim \mathrm{Exp}(1)$. Hence, $\gtrv \sim \mathrm{Exp}(1)$, so $\mathcal{M}_\gtrv(\theta)=(1+\theta)^{-1}$ $(\theta >-1)$ and, recalling~\eqref{Rrfromr}, $R_r=1+r$.  As explained just after~\eqref{IGSE}, $\gtrvt \sim \mathrm{Exp}(1+\lambda_H)$, so $\mathcal{M}_{\gtrvt}(\theta)=\frac{1+\lambda_H}{1+\lambda_H+\theta}$ $(\theta > -(1+\lambda_H))$.  It follows that $\widetilde{R}_r=1+\tilde{r}$, where $\tilde{r}$ solves \eqref{rtgrowth4}.

Figures~\ref{HSIRnondyn1} to~\ref{HSIRRr} show the various reproduction numbers as functions of the within-household infection rate $\lambda_H$ for various combinations of household size $n$ and overall global infection rate $\mu_G$.  The parameters and format are the same as in Figure 1 of Goldstein et al.~\cite{GoldsteinEtal2009}, though the range of values for $\lambda_H$ is reduced.
Note that in this model all of the reproduction numbers, except $R_r$ and $\widetilde{R}_r$, are invariant to the introduction of a latent period into the model.  Figure~\ref{HSIRnondyn1} compares all of the reproduction numbers except $R_{V\!L}$, $R_r$ and $\widetilde{R}_r$.
Observe that they are all ordered in accordance with Theorem~\ref{Hcomp}\thp{b} and that the conjectured comparison between $R_0$ and $R_2$ is also satisfied.  Moreover, $R_I=R_0=R_2$ ($=R_V$ in a growing epidemic) when $n=2$, as expected.  In particular, $R_0 \le R_V \le R_I$ in a growing epidemic.  Note that generally, in a growing epidemic, $R_*$ is appreciably greater than $R_I$ and is a poor approximation to $R_V$.  (Recall, though, that in the present setting, when all households have the same size, $R_*$ gives the correct critical vaccination coverage if households are either fully vaccinated or fully unvaccinated.)  Also, in a growing epidemic, $R_{H\!I}$ is generally a noticeably worse lower bound to $R_V$ than $R_2$.  Indeed $R_2$ and $R_0$ are very close and, as is the case in most of the figures, $R_0$ is very close to $R_V$.  Note that less knowledge of the epidemic model is required to compute $R_2$ than to compute $R_0$.

\begin{figure}[!ht] \begin{center}
\includegraphics[width=\textwidth]{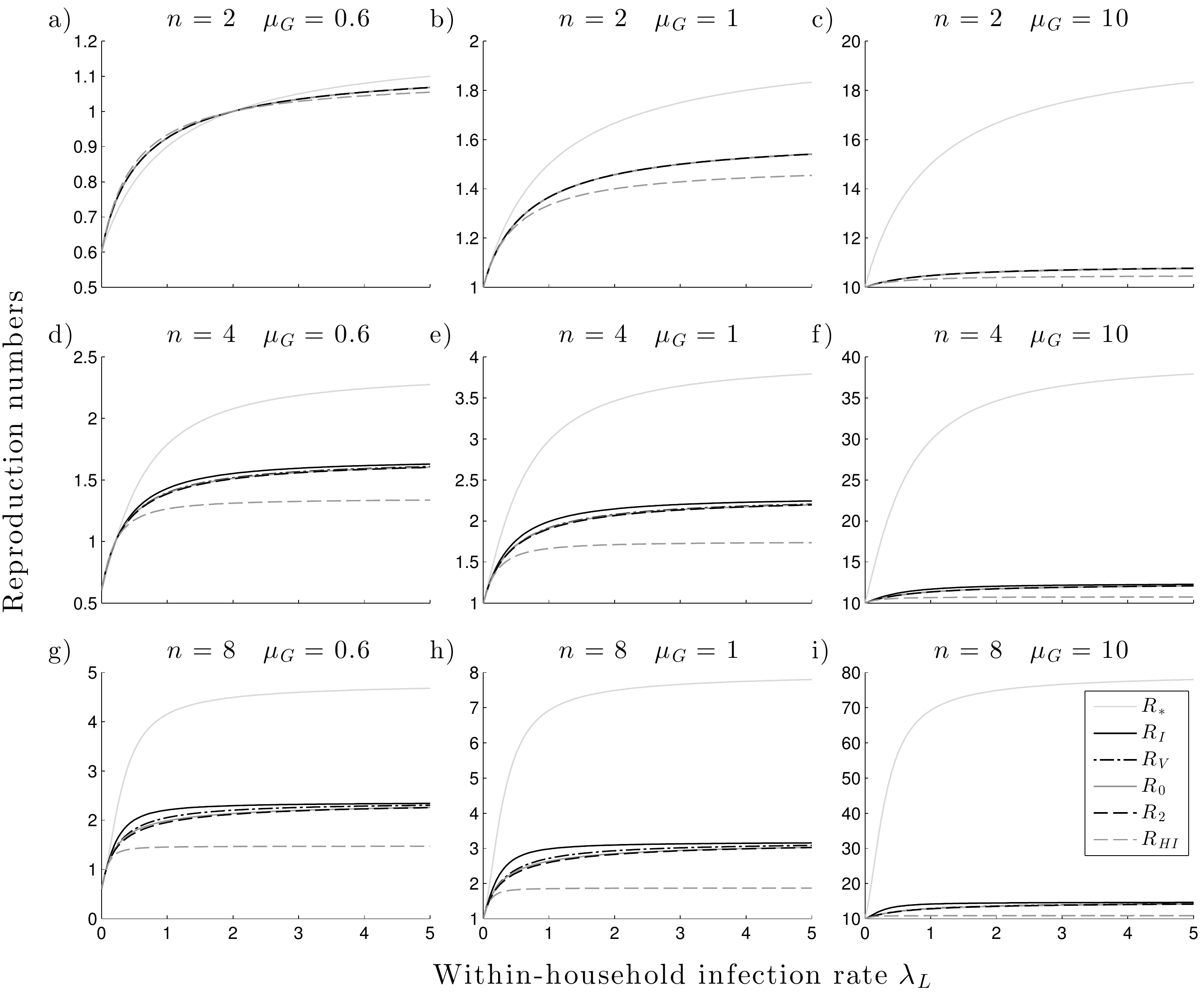}
 \caption{Reproduction numbers $R_*, R_I, R_V, R_0, R_2$ and $R_{H\!I}$ for the Markov SIR households model $\mathscr{E}^H(n,\mu_G,\lambda_H)$.}
\label{HSIRnondyn1}\end{center}
\end{figure}

Figure~\ref{HSIRnondyn2} compares the reproduction numbers $R_I, R_{V\!L}, R_V, R_0$ and  $R_2$.  Recall that, in our notation, Goldstein et al.~\cite{GoldsteinEtal2009} proved that, in a growing epidemic, $\bar{R}_{H\!I} \le R_V \le R_{V\!L} \le R_*$ so, using Theorem~\ref{Hcomp}\thp{b}, $R_0 \le R_V \le R_{V\!L}$, as is clearly seen in Figure~\ref{HSIRnondyn2}.  Note that although $R_I$ and $R_{V\!L}$ cannot be ordered in general (see the graphs when $n=8$ and $\mu_G=10$), $R_{V\!L}$ is generally appreciably larger than $R_I$, unless the within-household infection rate is small.

\begin{figure}[!ht] \begin{center}
\includegraphics[width=\textwidth]{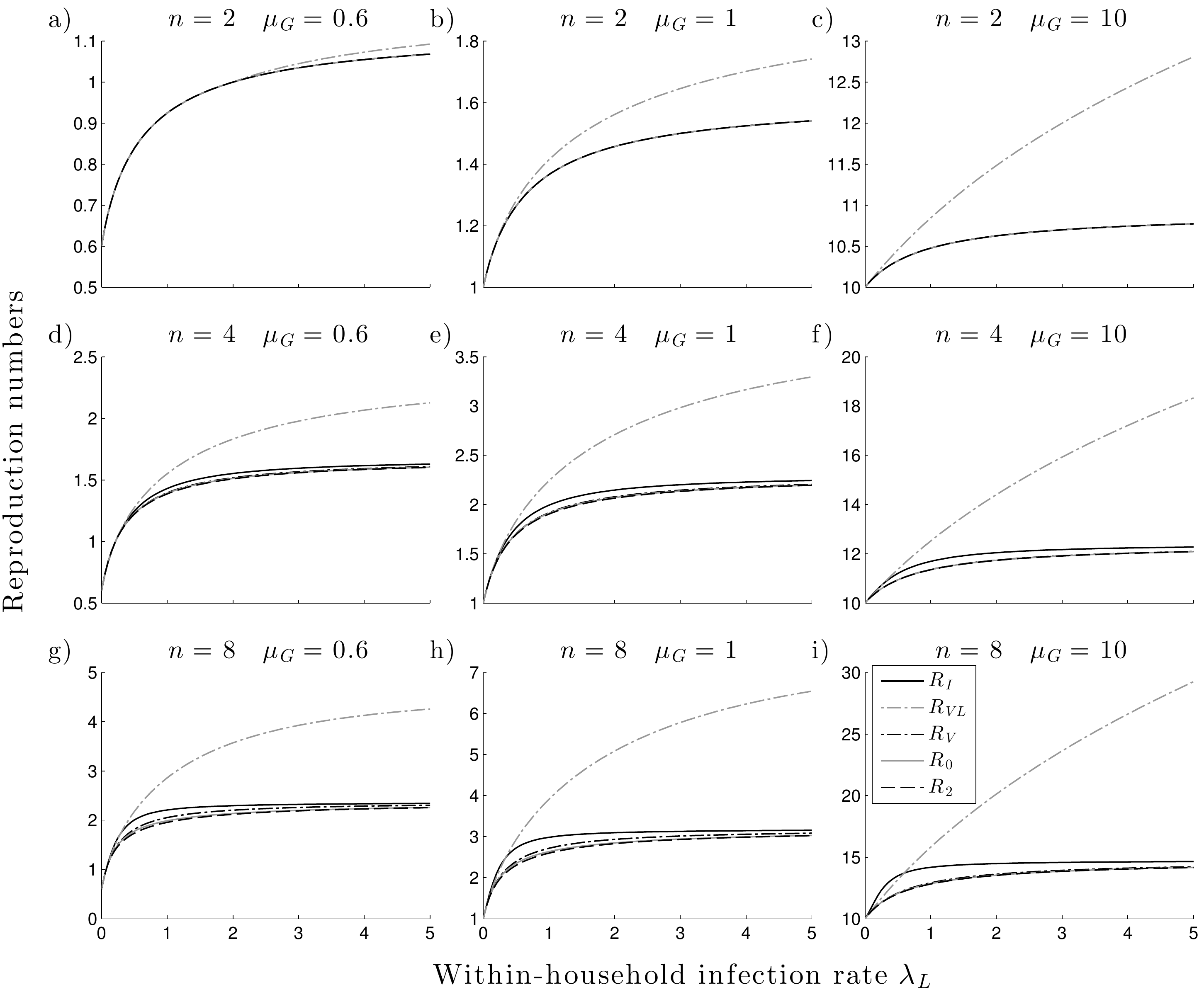}
 \caption{Reproduction numbers $R_I, R_{V\!L}, R_V, R_0$ and  $R_2$ for the Markov SIR households model $\mathscr{E}^H(n,\mu_G,\lambda_H)$. }
\label{HSIRnondyn2}\end{center}
\end{figure}

Figure~\ref{HSIRRr} compares the exponential-growth-associated reproduction number $R_r$ and its variant $\widetilde{R}_r$ with $R_I, R_V$ and $R_0$. Goldstein et al.~\cite{GoldsteinEtal2009} noted that in most plausible parameter regions $R_r \ge R_V$ and this is seen in Figure~\ref{HSIRRr}.  However, $R_r$ is usually an appreciably coarser upper bound than $R_I$ for $R_V$, though, as seen from the graphs when $n=8$ and $\mu_G=10$, it is not possible to order $R_r$ and $R_I$ in general.  As a particular case of Theorem~\ref{propos}\thp{c},
for the Markov SIR model, in all growing epidemics
$R_0 \le \widetilde{R}_r \le R_r$ and, for $n=2$, $R_r=\widetilde{R}_r$ since $\gtrvt$ gives the correct infection interval for local infection between the primary and secondary case in a household. Also note that, when $n=2$, $R_r = R_{V\!L}$. This is proved in Appendix~\ref{app:RrRVLcomp}, where it is shown that $R_r$ and $R_{V\!L}$ cannot in general be ordered.

\begin{figure}[!ht] \begin{center}
\includegraphics[width=\textwidth]{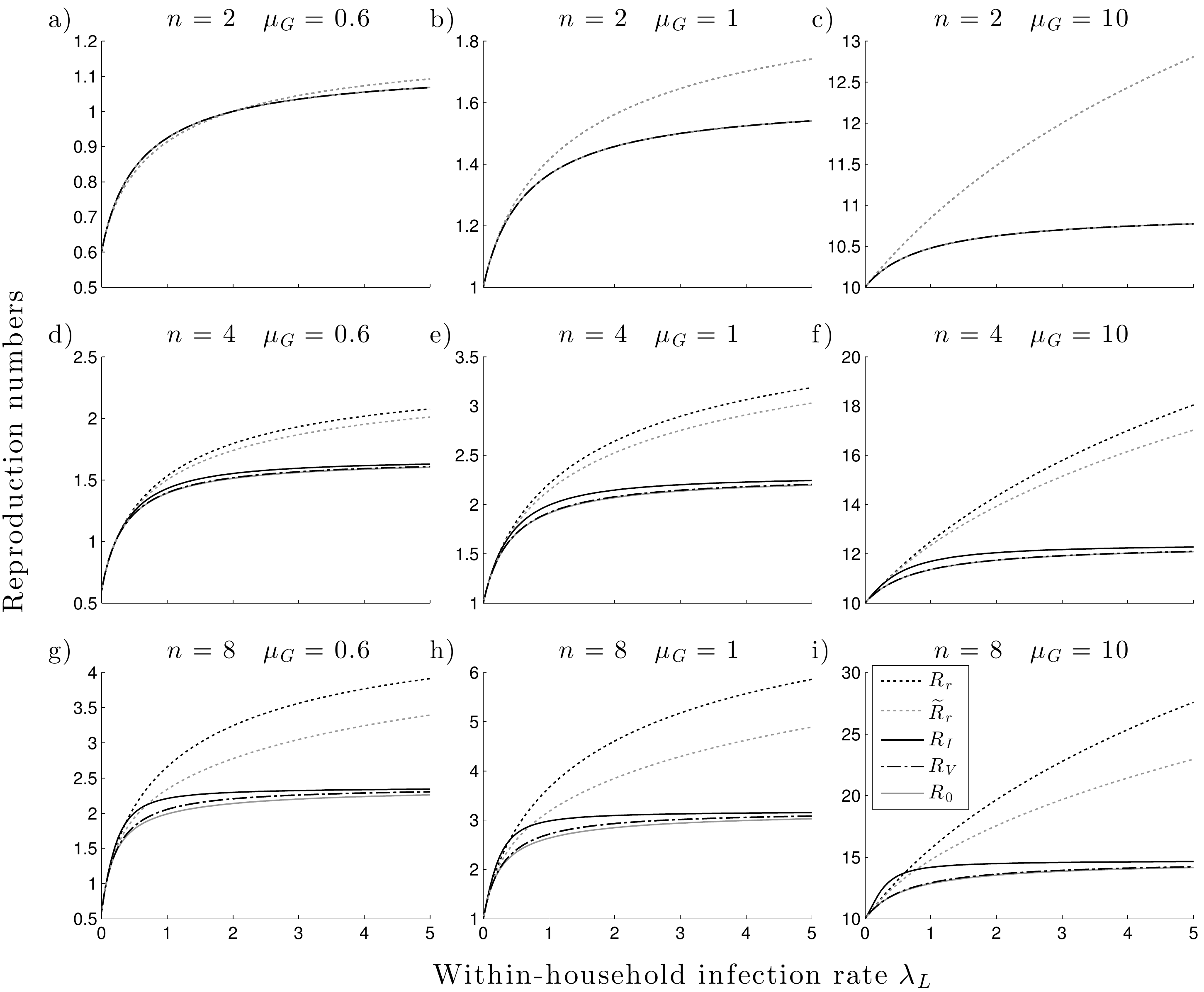}
 \caption{Reproduction numbers $R_r, \widetilde{R}_r, R_I, R_V$ and $R_0$ for the Markov SIR households model $\mathscr{E}^H(n,\mu_G,\lambda_H)$.  }
\label{HSIRRr}\end{center}
\end{figure}

We now add a latent period to the above model.  Specifically we assume that infectives have independent latent periods, each distributed as $\mathrm{Exp}(\delta)$, so the mean latent period is $\delta^{-1}$.  The latent periods are also independent of all the other random quantities used to define the model.  Thus the model is now a Markov SEIR households epidemic model and is identical to one used by Goldstein et al.~\cite{GoldsteinEtal2009} in their numerical illustrations.  As noted previously, the introduction of a latent period changes only the reproduction numbers $R_r$ and $\widetilde{R}_r$.

Denote the above model by $\mathscr{E}^H(n,\mu_G,\lambda_H,\delta)$. Goldstein et al.~\cite{GoldsteinEtal2009} determined the real-time growth rate $r$ for $\mathscr{E}^H(n,\mu_G,\lambda_H,\delta)$ by linearising a system of differential equations that describe the evolution of the relative numbers of households in different states (when the total population size $N$ is large), where the state of a household is given by the number of infected, latent and susceptible individuals it contains, and determining the corresponding largest eigenvalue.  We determine $r$ by extending the matrix method in Pellis et al.~\cite{PelFerFra2011}, Section 4.2, to incorporate a latent period.  The infectivity profile of a typical infective in $\mathscr{E}^H(n,\mu_G,\lambda_H,\delta)$ is given by
\begin{equation*}
\mathcal{I}(t)=\left\{
\begin{array}{ll}
1\qquad&\text{if } T_E \le t \le T_E+T_I,\\
0 &\text{otherwise },
\end{array} \right.
\end{equation*}
where $T_E \sim \mathrm{Exp}(\delta)$ and $T_I \sim \mathrm{Exp}(1)$ are independent random variables giving the latent and infectious periods of a typical infective.  It is then readily verified that $$\mathcal{M}_\gtrv(\theta)=\frac{\delta}{(\delta+\theta)(1+\theta)} \qquad (\theta > -\min(1,\delta)).$$   Note that $\gtrv=\gtrvz{0}+T_E$, where $\gtrvz{0}$ is the infectious contact interval for $\mathscr{E}^H(n,\mu_G,\lambda_H)$, and $\gtrvz{0}$ and $T_E$ are independent.  Further, in an obvious notation, $\gtrvt=\gtrvzt{0}+T_E$, whence
$$\mathcal{M}_{\gtrvt}(\theta)=\frac{\delta}{\delta+\theta}\frac{1+\lambda_H}{1+\lambda_H+\theta} \qquad (\theta>-\min(1+\lambda_H,\delta)).$$  Thus, given $r$, both $R_r$ and $\widetilde{R}_r$ are easily calculated.

Figure~\ref{HSEIRRr} shows the exponential-growth-associated reproduction numbers $R_r$ and $\widetilde{R}_r$, and also $R_I, R_{V\!L}, R_V$ and $R_0$, as functions of the mean latent period $\delta^{-1}$.  For the case $n=2$, $R_r=\widetilde{R}_r$, as with the SIR model, and $R_0=R_V=R_I$, agreeing with Theorem~\ref{Hcomp}\thp{b}.  Further, $\widetilde{R}_r \ne R_0$, since $\gtrv$ and $\gtrvt$ have different distributions.   Note that $R_r$ is decreasing in $\delta^{-1}$ and converges to $R_0$ as $\delta^{-1} \to \infty$. When $n=3$, a similar picture emerges except that $R_0=R_V<R_I$ and $R_r>\widetilde{R}_r$, though both $R_r$ and $\widetilde{R}_r$ converge to $R_0$ $(=R_V)$ as $\delta^{-1} \to \infty$.  Observe that neither $R_r$ nor $\widetilde{R}_r$ can be ordered with $R_I$.  The main differences between the cases $n=3$ and $n=4$ is that when $n=4$, $R_0<R_V$ and the exponential-growth-associated reproduction numbers $R_r$ and $\widetilde{R}_r$ tend to different limits as $\delta^{-1} \to \infty$, though the discrepancy is difficult to see as $R_0$ and $R_V$ are very close.  It is much clearer in the case when $n=8$.  Observe that $\widetilde{R}_r \to R_0$ as $\delta^{-1} \to \infty$, whilst $R_r$ converges to a limit lying strictly between $R_0$ and $R_V$.  The fact that
$R_r<R_V$ for very long latent periods when $n \ge 4$ is noted in Goldstein et al.~\cite{GoldsteinEtal2009}, though the proof in Appendix B of that paper, which in our terminology shows that $R_r \to R_0$ as the latent periods become infinitely long, does not hold for the Markov SEIR households model.  This is explored further in Appendix~\ref{app:inflonglat}, where it is proved that for the Markov SEIR households model, in the limit as $\delta^{-1} \to \infty$, if the maximum household size $n_H \le 3$ then $R_r=\widetilde{R}_r=R_0$ $(=R_V)$, whilst if $n_H \ge 4$ then $R_r>\widetilde{R}_r=R_0$.
Further, when $n_H=4$, we show that $R_V>R_r>R_0$, though we do not have a proof for $n_H \ge 5$.  Although such long latent periods do not occur in real-life infections, we let the mean latent period $\delta^{-1}$ in
Figures~\ref{HSEIRRr} and~\ref{HWSEIR} run up to $10^4$ times the mean infectious period to illustrate the limiting behaviour of the reproduction numbers as $\delta^{-1} \to \infty$.
\begin{figure}[!ht] \begin{center}
\includegraphics[width=\textwidth]{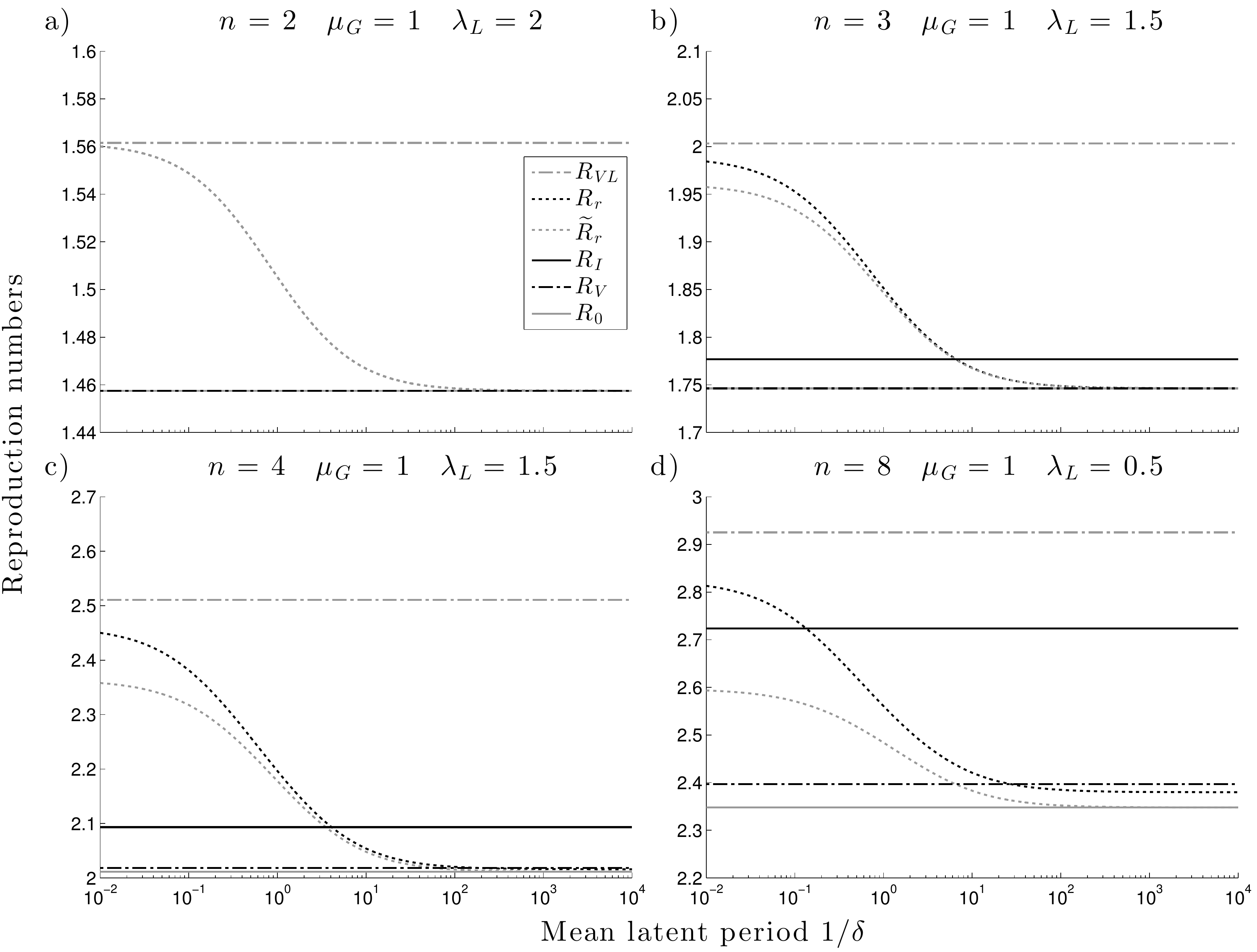}
 \caption{Reproduction numbers $R_r, \widetilde{R}_r, R_{V\!L},R_I, R_V$ and $R_0$ for the Markov SEIR households model $\mathscr{E}^H(n,\mu_G,\lambda_H,\delta)$.  }
\label{HSEIRRr}\end{center}
\end{figure}

\subsection{Households model with non-random infectivity profile}
\label{nonrandHM}
We now assume that the infectivity profile of an individual is non-random.  Specifically, following Fraser~\cite{Fra2007} and Goldstein et al.~\cite{GoldsteinEtal2009}, we assume that the infectious contact interval $\gtrv$ follows a gamma distribution, with parameters $\alpha>0$ and $\gamma>0$. Thus $\mathcal{I}(t)=\gtpdf(t)$ $(t \ge 0)$, where
\begin{equation}
\label{gammapdf}
\gtpdf(t)=\frac{\gamma^\alpha t^{\alpha-1} {\rm e}^{-\gamma t}}{\Gamma(\alpha)}
\end{equation}
and $\Gamma(\alpha) = \int_0^\infty{t^{\alpha-1}{\rm e}^{-t}\mathrm{d}t}$ is the gamma function.  Similar to Section~\ref{subsec:Rr}, we assume that, $t$ time units after he/she was infected, an infectious individual makes global contacts at overall rate $\mu_G \gtpdf(t)$ and, additionally, he/she contacts any given susceptible in his/her household at rate $\lambda_H \gtpdf(t)$.  Thus, since $\int_0^\infty \gtpdf(t) {\rm d}t=1$, a given infective infects locally other members of its household independently, each with probability $p=1- {\rm e}^{-\lambda_H}$.  It follows that the mean generation sizes  $\mu_1^{(n)}, \mu_2^{(n)}, \cdots, \mu_{n-1}^{(n)}$ for a single size-$n$ household epidemic coincide with those of a Reed-Frost model with escape probability $q=1-p$ and hence may be computed using the algorithm in Appendix A of Pellis et al.~\cite{PelBalTra2012}.  This enables all of the reproduction numbers, except for $R_r$ and $\widetilde{R}_r$, to be computed in a similar fashion as for the Markov SIR model.
(Again $\mu_H^{(n)}$ may be computed more directly using Ball~\cite{Ball1986}, Equations (2.25) and (2.26).)  Note that, except for $R_r$ and $\widetilde{R}_r$, all of the reproduction numbers are independent of the parameters $(\alpha,\gamma)$ of the gamma distribution that describes the infectivity profile; indeed they are independent of the infectivity profile, provided it is non-random.

To calculate $R_r$, the real-time growth rate $r$ is required, for which we are not aware of any exact method of calculation.  Goldstein et al.~\cite{GoldsteinEtal2009} used stochastic simulations, involving an approximate discrete-time model having a small time step, to estimate the mean infectivity profile $\beta_H(t)$ $(t \ge 0)$ of a single-household epidemic.  We use a simulation method, based on a Sellke~\cite{Sellke1983} construction and described in Appendix~\ref{app:sellke}, to estimate the Laplace transform $\mathcal{L}_{\beta_H}(\theta)$ of $\beta_H(t)$, whence $r$ is obtained by solving $\mathcal{L}_{\beta_H}(r)=1$ numerically. The reproduction number $R_r$ then follows, using~\eqref{Rrfromr} with $\mathcal{M}_\gtrv(r)=\left(\frac{\gamma}{\gamma+r}\right)^\alpha$.  In the present model there is no closed-form expression for $\mathcal{M}_{\gtrvt}(\theta)$, so we do not consider $\widetilde{R}_r$.

\begin{figure}[htp]
\begin{center}
\includegraphics[width=\textwidth]{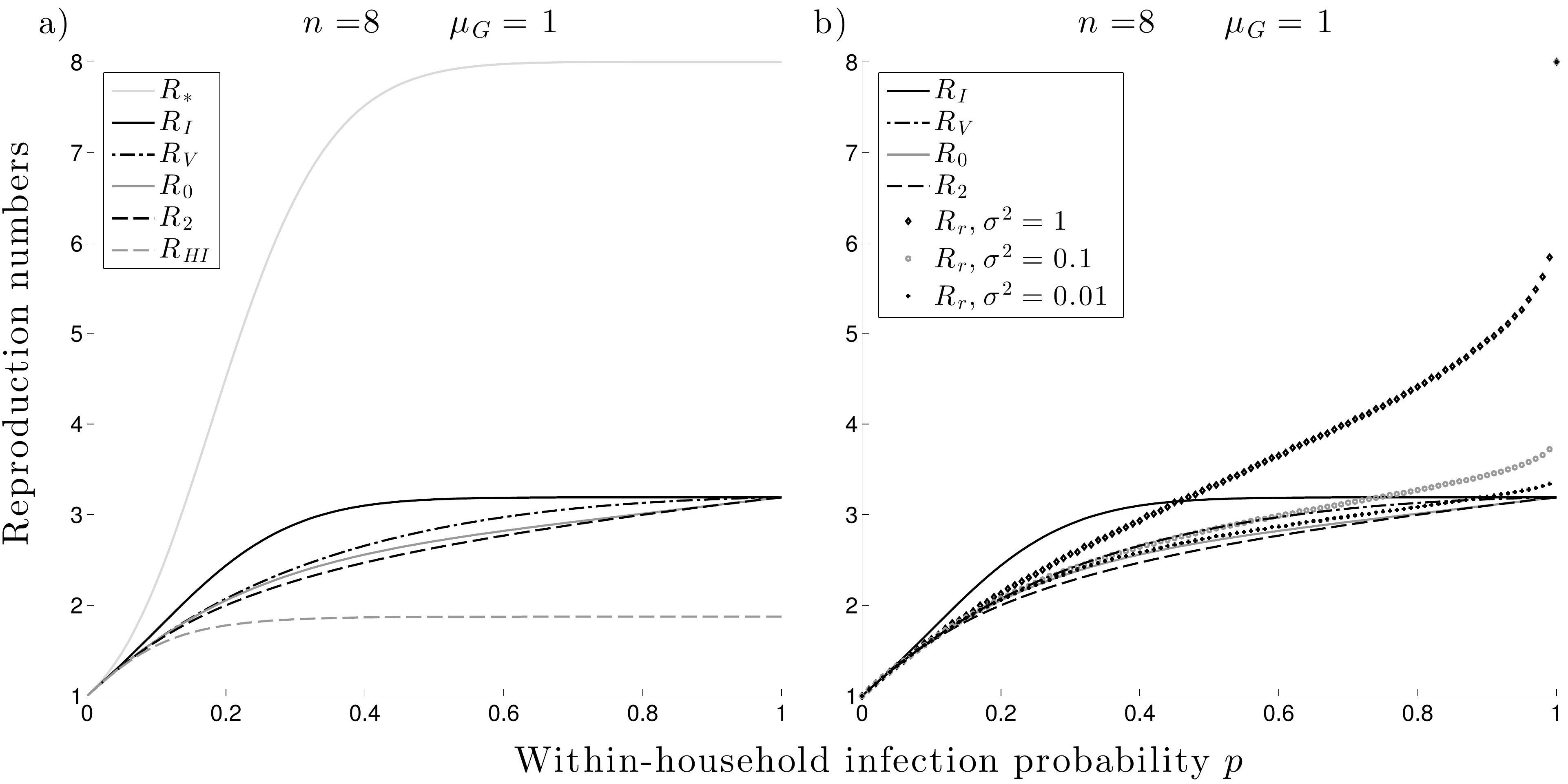}
 \caption{\fpc{a} Reproduction numbers $R_*, R_I, R_V, R_0, R_2$ and $R_{H\!I}$ for a households model with a non-random infectivity profile; and \fpc{b} reproduction numbers $R_I, R_V, R_0, R_2$ and $R_r$ (with $\sigma^2=1,0.1$ and $0.01$) for a households model with a non-random infectivity profile which follows a gamma distribution with mean $1$. }
\label{TVI}\end{center}
\end{figure}

For brevity we present results only for the case when all households are of size $8$ and $\mu_G=1$.
In Figure~\ref{TVI}\fpt{a}, the reproduction numbers $R_*, R_I, R_V, R_0, R_2$ and $R_{H\!I}$ are plotted against the within-household infection probability $p$.  These reproduction numbers satisfy $R_{H\!I}<R_0<R_V<R_I<R_*$ and $R_{H\!I}<R_2<R_I$, as predicted by Theorem~\ref{Hcomp}\thp{b}, and the conjecture $R_2<R_0$.  As $p \to 1$, so $\lambda_H \to \infty$, the mean generation sizes become $\mu_1=7$ and $\mu_k=0$ $(k=2,3,\cdots,7)$, and the corresponding limiting values of the reproduction numbers are easily obtained.  Note that unless $p$ is small, i.e.~unless there is very little enhanced spread of infection within households, $R_*$ is a coarse upper bound for $R_V$ and $R_{H\!I}$ is a coarse lower bound.  Further, $R_0$ is a good approximation to $R_V$ across the full range of values for $p$, though it is an underestimate.

Figure~\ref{TVI}\fpt{b} shows the reproduction numbers $R_2, R_0, R_V, R_I$ and $R_r$ as functions of $p$.  Note that $R_r$ depends on the parameters of the gamma distribution describing the non-random infectivity profile.  When $\gtrv$ has probability density function given by~\eqref{gammapdf}, $\bbE[\gtrv]=\frac{\alpha}{\gamma}$ and $\sigma^2=\mathrm{Var}(\gtrv)=\frac{\alpha}{\gamma^2}$.
In Figure~\ref{TVI}\fpt{b}, we assume that $\bbE[\gtrv]=1$, so $\alpha=\gamma$, and show $R_r$ when $\sigma^2=1$ $(\alpha=1)$, $\sigma^2=0.1$ $(\alpha=10)$ and $\sigma^2=0.01$ $(\alpha=100)$.  Each graph for $R_r$ is estimated from $10,000$ simulations of the corresponding single-household epidemic.  Observe that, for fixed $p$, the exponential-growth-associated reproduction number $R_r$ is a decreasing function of $\sigma^2$.  As $\sigma^2$ decreases to $0$ the epidemic model becomes more and more like a Reed-Frost type model, for which $R_r=R_0$.  The accuracy of $R_r$ as an approximation to $R_V$ depends on both the variance of the infectious contact interval $W$ and on how infectious the disease is within households.  Generally, the approximation is good when $p$ is small, since then there is little spread within households, and improves as $\sigma^2$ decreases.  Normally, $R_r$ overestimates $R_V$ but when the infectious contact interval is highly peaked it may be a slight underestimate, as is illustrated in the graph when $\sigma^2=0.01$.


\subsection{Markov SIR and SEIR households-workplaces models}
\label{MarkovSIRHWmod}
The Markov SIR households model described in Section~\ref{MarkovSIRHmod} is readily generalised to incorporate workplaces.  For simplicity, we assume that all households have common size $n_H$ and all workplaces  have common size $n_W$.  During his/her infectious period, which is distributed as $\mathrm{Exp}(1)$, a typical infective makes global contacts at overall rate $\mu_G$, infects any given susceptible in his/her household at rate $\lambda_H$ and any susceptible in his/her workplace at rate $\lambda_W$.  The mean generation sizes for within-household and within-workplace epidemics may be evaluated using the methods described for the Markov SIR households model, so, apart from $R_r$ and $\widetilde{R}_r$, the reproduction numbers are readily computed.

To compute the exponential-growth-associated reproduction number $R_r$, consider first a single-household epidemic and let $S(t)$ and $I(t)$ be respectively the numbers of susceptible and infectives at time $t$.  Then, at time $t$, new infections occur in this household at rate $\lambda_H S(t) I(t)$, so, in the notation of Section~\ref{subsec:HWRr}, $\xi_H(t)=\lambda_H\mathbb{E}[ S(t) I(t)]$ $(t \ge 0)$. Hence
\begin{equation*}
\mathcal{L}_{\xi_H}(\theta)=\int_0^\infty \lambda_H\mathbb{E}[ S(t) I(t)]\mathrm{e}^{-\theta t}\mathrm{d}t,
\end{equation*}
which can be evaluated numerically using the matrix method described in Pellis et al.~\cite{PelFerFra2011}, Section 4.3.  The Laplace transform $\mathcal{L}_{\xi_W}(\theta)$ may be computed similarly. The real-time growth rate $r$ may be computed by solving~\eqref{HWrdefeq} numerically  (recall that $\mathcal{M}_\gtrv(\theta)=(1+\theta)^{-1}$) and $R_r$ is then given by $R_r=1+r$.  Note that, in the notation of~\eqref{wtildeHapprox1} 
$\mathcal{M}_{\gtrvt^{H}}(r)=\frac{1+\lambda_H}{1+\lambda_H+r}$ and  $\mathcal{M}_{\gtrvt^{W}}(r)=\frac{1+\lambda_W}{1+\lambda_W+r}$, thus enabling $\tilde{r}$ to be computed, whence $\widetilde{R}_r=1+\tilde{r}$.

\begin{figure}[!ht] \begin{center}
\includegraphics[width=\textwidth]{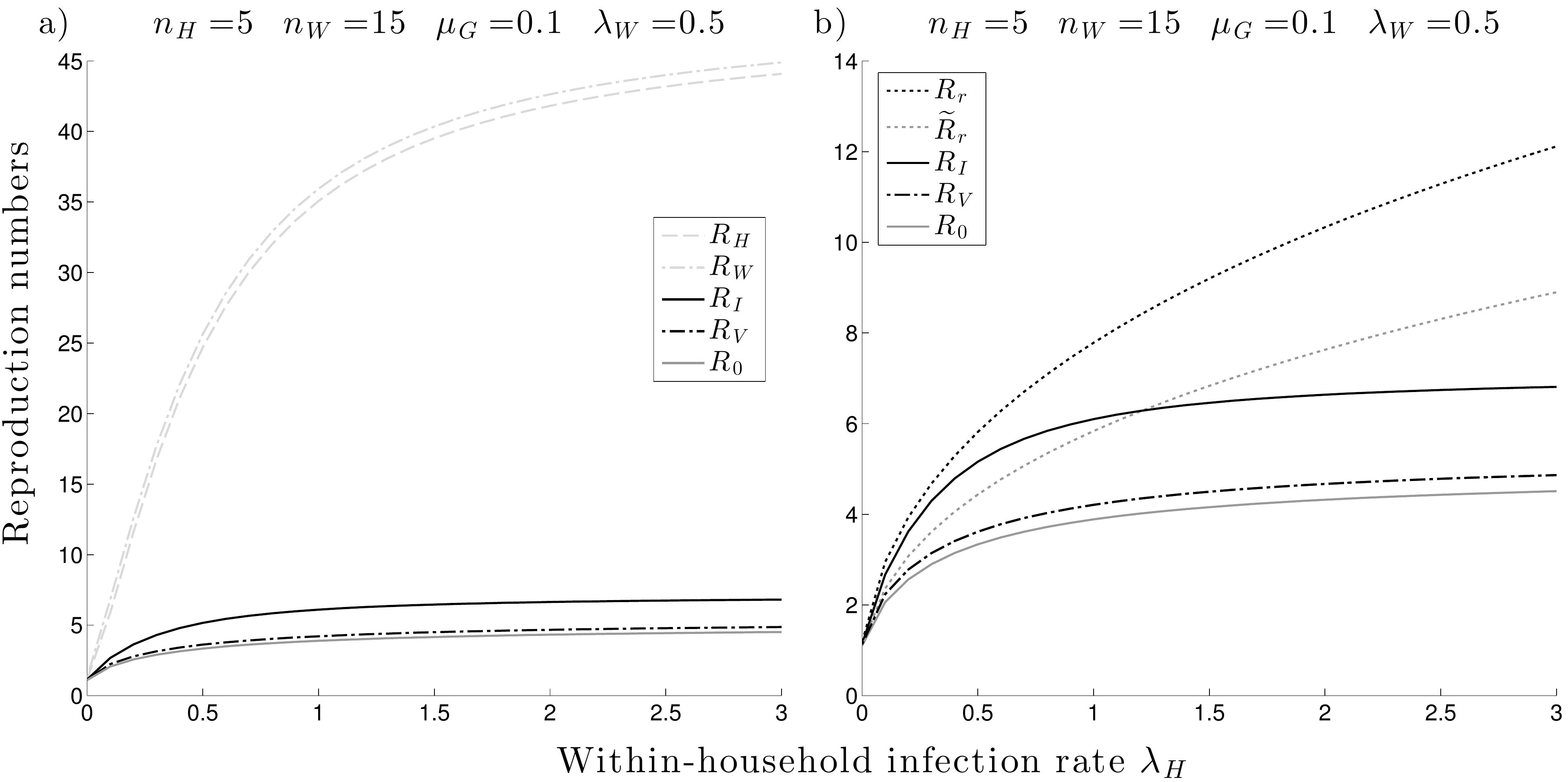}
 \caption{\fpc{a} Reproduction numbers $R_H, R_W, R_I, R_V$ and $R_0$ and \fpc{b} reproduction numbers $R_r, \widetilde{R}_r, R_I, R_V$ and $R_0$ for a Markov SIR households-workplaces model.}
\label{HWSIR}\end{center}
\end{figure}

Figure~\ref{HWSIR} is
for a model in which $n_H=5$ and $n_W=15$. Figure~\ref{HWSIR}\fpt{a} shows graphs of the reproduction numbers $R_H, R_W, R_I, R_V$ and $R_0$ against $\lambda_H$ when $\mu_G=0.1$ and $\lambda_W=0.5$.  For these parameter values, $R_H$ and $R_W$ are distinct, though their difference is small and both are useless as approximations to $R_V$ ($R_*$, not shown as it is so large, is even worse).  This is because of the large within-workplace epidemic sizes.  Note that the reproduction numbers satisfy the inequalities proved in Theorem~\ref{Theorem 5.2}\thp{b}. In Figure~\ref{HWSIR}\fpt{b}, the reproduction numbers $R_r, \widetilde{R}_r, R_I, R_V$ and $R_0$ are plotted against $\lambda_H$ when $\mu_G=0.5$ and
$\lambda_W=0.1$.  Observe that $R_0 < \widetilde{R}_r < R_r$ for $\lambda_H>0$, in
accordance with Theorem \ref{Theorem 5.2}\thp{c}, and that neither $R_r$ nor $\widetilde{R}_r$ can be ordered with $R_I$.  Unless $\lambda_H$ is small, $R_r$ is not a good approximation to $R_V$.  Note that
in Figure~\ref{HWSIR},
$R_0$ is a close approximation to $R_V$ for all values of $\lambda_H$.

Finally, we consider the Markov SEIR version of the above model, which incorporates a latent period having an $\mathrm{Exp}(\delta)$ distribution.  Again, apart from $R_r$ and $\widetilde{R}_r$, the reproduction numbers are unchanged by the inclusion of a latent period.  The method described above for computing the real-time growth rate $r$ is easily extended to the present model.  Note that $\mathcal{M}_\gtrv(\theta)$ is the same as in the above Markov SEIR households model, $\mathcal{M}_{\gtrvt^{H}}(\theta)=\frac{\delta}{\delta+\theta}\frac{1+\lambda_H}{1+\lambda_H+\theta}$ $(\theta>-\min(1+\lambda_H,\delta))$ and $\mathcal{M}_{\gtrvt^{W}}(\theta)=\frac{\delta}{\delta+\theta}\frac{1+\lambda_W}{1+\lambda_W+\theta}$ $(\theta>-\min(1+\lambda_W,\delta))$, thus enabling $R_r$ and $\widetilde{R}_r$ to be computed.

\begin{figure}[!ht] \begin{center}
\includegraphics[width=\textwidth]{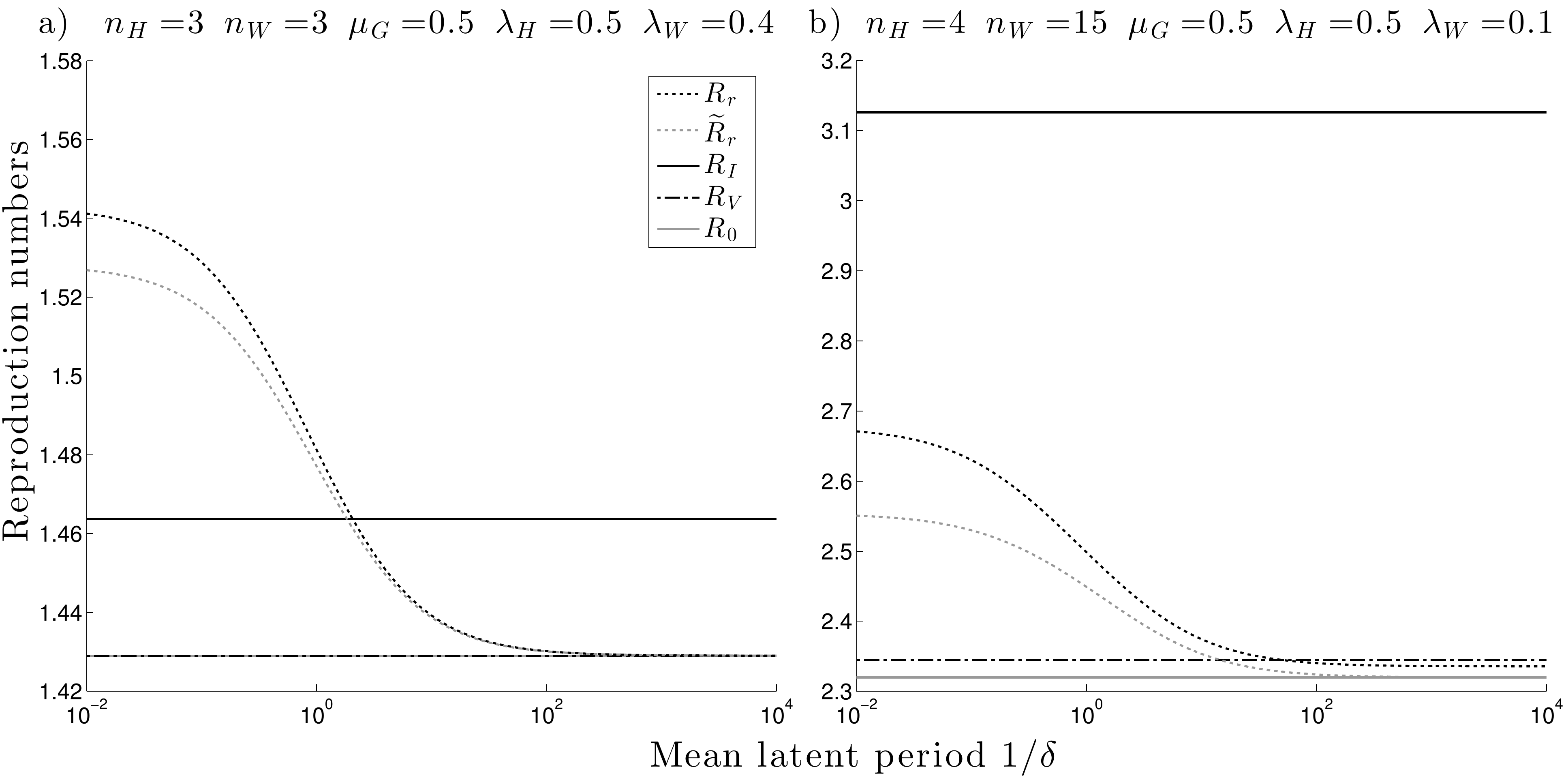}
 \caption{Reproduction numbers $R_r, \widetilde{R}_r, R_I, R_V$ and $R_0$ for a Markov SEIR households-workplaces model.}
\label{HWSEIR}\end{center}
\end{figure}

Figure~\ref{HWSEIR}\fpt{a} shows the dependence of the reproduction numbers $R_r, \widetilde{R}_r, R_I, R_V$ and $R_0$ on the mean latent period $\delta^{-1}$ when $n_H=n_W=3$, $\mu_G=0.5, \lambda_H=0.5$ and $\lambda_W=0.4$.  Note that $R_V=R_0<R_I$, as predicted by Theorem~\ref{Theorem 5.2}\thp{b}, and that both $R_r$ and $\widetilde{R}_r$ converge down to $R_0$ as $\delta^{-1} \to \infty$. Figure~\ref{HWSEIR}\fpt{b} shows the same reproduction numbers when $n_H=4$ and $n_W=15$.  The values of $\mu_G$ and $\lambda_H$ are as before and $\lambda_W$ is now $0.1$, in view of the larger workplace size.
Now, $R_0<R_V<R_I$, again as predicted by Theorem~\ref{Theorem 5.2}\thp{b}, and $\widetilde{R}_r \to R_0$ as $\delta^{-1} \to \infty$, whereas $R_r$ tends to a limit lying strictly between $R_0$ and $R_V$.  The limiting case when $\delta^{-1} \to \infty$ is analysed in Appendix~\ref{app:inflonglat}, where similar results as for the households model are proved.  Note that in Figure~\ref{HWSEIR}\fpt{b}, $R_I$ is appreciably greater than $R_r$, owing in part to the effect of large workplaces.

\section{Proofs}
\label{sec:proofs}
We define $\sign(x)$ to be $-1$, $0$ and $1$, for $x<0$, $x=0$ and $x>0$, respectively.
\subsection{Proof of Theorem~\ref{Hcomp}}
\label{subsec:Hcompproof}
 To shorten the exposition of the proof, we use the notation $\overset{n}{\geq}$ to denote that there is equality if the population contains no household with size strictly larger than $n$ and the inequality is strict if the population contains households with size strictly larger than $n$.  With this notation, the statement of Theorem~\ref{Hcomp} is as follows.

\begin{duplicate} 
\
\begin{enumerate}
\item[(a)]
$R_*=1 \iff R_I =1 \iff R_0 = 1 \iff R_2=0 \iff R_{H\!I}=1 \implies R_V = 1$.

\item[(b)]
In a growing epidemic,
\[
R_* > R_I \overset{2}{\geq} R_V \overset{3}{\geq} R_0 > R_{H\!I} >1\quad \mbox{ and }\quad R_I \overset{2}{\geq} R_2 > R_{H\!I} > 1,
\]
and in a declining epidemic
\[
R_* < R_I \overset{2}{\leq} R_0 < R_{H\!I}  <1  \quad\mbox{ and }\quad R_I \overset{2}{\leq} R_2 < R_{H\!I} < 1.
\]

\end{enumerate}
\end{duplicate}

\begin{proof}
We first prove \thp{a}.  Note from~\eqref{g0H} and~\eqref{gind} that  $$g_0(1)=g_I(1)=1-\mu_G(1+\mu_H).$$  Recalling~\eqref{gHI}, note that $1 - a^{(n)}=1/(1+\mu_H^{(n)})$ $(n=1,2,\cdots,n_H),$ so $$g_{H\!I}(1)=1-\mu_G\sum_{n=1}{n_H}
\pi_n (1+\mu_H^{(n)})=1-\mu_G(1+\mu_H).$$  Similarly, recalling~\eqref{g2unequal}, $1-b^{(n)}=\mu_1^{(n)}/\mu_H^{(n)}$ $(n=2,3,\cdots,n_H)$, so $$g_2(1)=1-\mu_G\sum_{n=2}{n_H}\pi_n (1+\mu_H^{(n)})=1-\mu_G(1+\mu_H).$$
Recall that $R_*=\mu_G(1+\mu_H)$. Now $g_0$ and $g_I$ are strictly increasing on $(0,\infty)$, $g_{H\!I}$ is strictly increasing on $(a,\infty)$ and $g_2$ is strictly increasing on $(b,\infty)$, so 
$$\sign(g_0(1))=\sign(g_I(1))= \sign(g_{H\!I}(1))=\sign(g_2(1))=\sign(1-R_*),$$ 
since $R_0, R_I, R_{H\!I}$ and $R_2$ are the unique roots of $g_0, g_I, g_{H\!I}$ and $g_2$, respectively.
Thus $$R_*=1 \iff R_I = 1 \iff R_0=1 \iff R_{H\!I}=1 \iff R_2=1,$$ as required. By definition, $R_V=1$ if $R_*=1$.

To prove \thp{b}, we first note that the above argument shows that the reproductions numbers $R_*, R_I, R_0, R_{H\!I}$ and $R_2$ are all strictly greater than 1 in a growing epidemic and strictly smaller than 1 in a declining
epidemic. We consider now each of the comparisons in turn.
\begin{enumerate}
\item[(i)]
\textit{$R_*$ and $R_I$.}

Suppose that $R_*>1$. From (\ref{gind}),
\[
g_I (R_*)=1-\frac{\mu_G}{R_*} - \frac{\mu_H \mu_G}{R_*^2} > 1-\frac{\mu_G}{R_*} - \frac{\mu_H \mu_G}{R_*}=0,
\]
since $R_*=\mu_G(1+\mu_H)$. Thus $R_I<R_*$, since $g_I$ is increasing in $(0,\infty)$ and $R_I$ is the unique root of $g_I$ in $(0,\infty)$. A similar argument shows that $R_* < R_I$ when $R_* < 1$.

\item[(ii)]
\textit{$R_I$ and $R_V$}

Suppose that $R_I>1$ and a fraction $p$ of the population is vaccinated with a perfect vaccine.  Then $\mu_G$ is reduced to $\mu_G (p) = (1-p) \mu_G$ and $\mu_H$ is reduced to $\mu_H (p)$, for which we now obtain a simple upper bound.  Consider the epidemic graph $\mathcal{G}^{(n)}$
defined in Section~\ref{Hmodel}.  For $i=1,2,\cdots,n-1$, let $\chi_i^{(n)}$  denote the event that individual $i$ becomes infected in the single household epidemic (i.e.~if in $\mathcal{G}^{(n)}$ there is a chain of directed edges from $0$ to $i$) and let $(\chi_i^{(n)})^{\mathcal{C}}$ denote its complement.  Then the mean size of the single household epidemic (not including the primary case) is given by $\mu_H^{(n)}=\sum_{i=1}^{n-1} \mathbb{P}(\chi_i^{(n)})$.  Now keep the same realisation of $\mathcal{G}^{(n)}$, vaccinate each initial susceptible independently with probability $p$, and hence obtain a realisation of the single-household epidemic with vaccination.  For $i=1,2,\cdots,n-1$, let $\chi_i^{(n)} (p)$ be the event that individual $i$ is infected by the epidemic in the vaccinated population and let
$(\chi_i^{(n)} (p))^{\mathcal{C}}$ be its complement.  Clearly, for $p>0$, if $\chi_i^{(n)} (p)$ occurs, then so does $\chi_i^{(n)}$ and $i$ is not vaccinated.  Hence, if $p > 0$, $\mathbb{P}(\chi_i^{(n)} (p)) \overset{2}{\leq} (1-p)\mathbb{P}(\chi_i^{(n)})$, since vaccination is independent of $\mathcal{G}^{(n)}$. (Note that $\chi_1^{(2)} (p)$ occurs if and only if $\chi_1^{(2)}$ occurs and 1 is not vaccinated, so 
$\mathbb{P}(\chi_1^{(2)} (p)) = (1-p)\mathbb{P}(\chi_1^{(2)})$.  However, for $n>2$, given that individual $1$ is
not vaccinated, it does not necessarily follow that he/she is infected in the vaccinated epidemic if he/she is infected in the unvaccinated epidemic, since all chains from from individual $0$ to individual $1$ may
still be broken by vaccination.) This inequality implies, in obvious notation, that $\mu_H^{(n)} (p) \overset{2}{\leq} (1-p) \mu_H^{(n)}$, and taking expectations with respect to the size-biased household size distribution $(\pi_n)$ then gives $\mu_H (p) \overset{2}{\leq} (1-p) \mu_H$.  Let $R_I (p)$ denote the post-vaccination version of $R_I$.  Then, as at (\ref{gind}), $R_I (p)$ is the unique solution of
$g_{I,p} (\lambda)=0$ in $(0,\infty)$, where
\[
g_{I,p} (\lambda) = 1-\frac{\mu_G (p)}{\lambda} - \frac{\mu_H (p) \mu_G (p)}{\lambda^2}.
\]
Now, for $R_I >1$ and $p>0$,
\begin{align*}
g_{I,p} ((1-p)R_I)&=1-\frac{\mu_G (p)}{(1-p)R_I} - \frac{\mu_H (p) \mu_G (p)}{(1-p)^2 R_I^2}\\
&\overset{2}{\geq} 1-\frac{\mu_G}{R_I} - \frac{\mu_H \mu_G}{R_I^2} =0,
\end{align*}
by the definition of $R_I$.  It follows that $R_I (p) \overset{2}{\leq} (1-p)R_I$ and, in particular, if $p=1-R_I^{-1}$ then
$R_I (p) \overset{2}{\leq} 1$.  Hence, $p_C \overset{2}{\leq} 1-R_I^{-1}$ and, using (\ref{rv}), $R_V \overset{2}{\leq} R_I$.

\item[(iii)]
\textit{$R_V$ and $R_0$}

Suppose that $R_0>1$ and a fraction $p$ of the population is vaccinated with a perfect vaccine.  Then, cf.~(\ref{g0H}), the post-vaccination basic reproduction number, $R_0 (p)$ say, is given by the unique solution in $(0,\infty)$ of $g_{0,p} (\lambda)=0$, where
\[
g_{0,p} (\lambda )=1-\sum_{k=0}^{n_H-1} \frac{\nu_k (p)}{\lambda^{k+1}},
\]
with $\nu_k (p) = \mu_k (p) \mu_G (p)$.  Here, $\mu_G (p) = (1-p) \mu_G$ (as above) and, for $k=0,1,\cdots$, $\mu_k (p)$ is the post-vaccination mean size of the $k$th generation in a typical single-household epidemic with one initial infective (who is not vaccinated, so $\mu_0 (p)=1$).  We now obtain a lower bound for $\mu_k (p)$ ($k=1,2,\cdots$).  Consider again the epidemic graph $\mathcal{G}^{(n)}$.  For $k,i=1,2,\cdots,n-1$, let $\chi_{k,i}^{(n)}$ be the event that individual $i$ is a generation-$k$ infective and let $(\chi_{k,i}^{(n)})^{\mathcal{C}}$ be its complement.  Then the mean size of the $k$th generation is given by
$\mu_k^{(n)}=\sum_{i=1}^{n-1} \mathbb{P}(\chi_{k,i}^{(n)})$.  Now construct a realisation of the post-vaccination single-household epidemic as above, and define $\mu_k^{(n)} (p)$ and $\chi_{k,i}^{(n)} (p)$ in the obvious fashion.  Then, fix generation $k$ and suppose that individual $i$ is a generation-$k$ infective in the unvaccinated epidemic. Then $\chi_{k,i}^{(n)}$ occurs and in $\mathcal{G}^{(n)}$ there exists at least one chain of directed arcs of length $k$ from the initial infective to individual $i$, and there is no shorter such chain connecting those individuals. Fix such a path of length $k$.  If all $k$ members of that path avoid vaccination, which happens with probability $(1-p)^k$ independently of $\mathcal{G}^{(n)}$, then $\chi_{k,i}^{(n)}(p)$ occurs. 
Therefore, $\mathbb{P}(\chi_{k,i}^{(n)} (p)|\chi_{k,i}^{(n)}) \geq (1-p)^k$, whence
\begin{align*}
\mathbb{P}(\chi_{k,i}^{(n)} (p))
&= \mathbb{P}(\chi_{k,i}^{(n)})\mathbb{P}(\chi_{k,i}^{(n)} (p)|\chi_{k,i}^{(n)})+\mathbb{P}((\chi_{k,i}^{(n)})^{\mathcal{C}})\mathbb{P}(\chi_{k,i}^{(n)} (p)|(\chi_{k,i}^{(n)})^{\mathcal{C}})\\
&\ge \mathbb{P}(\chi_{k,i}^{(n)})(1-p)^k,
\end{align*}
which implies that $\mu_k^{(n)} (p) \geq (1-p)^k \mu_k^{(n)}$ and hence that $\mu_k (p) \geq (1-p)^k \mu_k$.  Note that for households of size $n \leq 3$, $\mu_k^{(n)} (p) = (1-p)^k \mu_k$ ($k=0,1,\cdots,n-1)$, since there can be at most one chain of length $k$ linking an individual to the initial susceptible, but for $n \geq 4$ and $p > 0$ the inequality
$\mu_k^{(n)} (p) \geq (1-p)^k \mu_k^{(n)}$ is strict for at least one $k$, as two or more chains may link an individual to the initial infective.

Clearly $R_V=1$ if $R_0=1$.  Suppose that $R_0 >1$ and let $p_C'=1-R_0^{-1}$.  Then,
\begin{align*}
g_{0,p'_{C}} (1)&=1-\sum_{k=0}^{n_H-1} \nu_k (p_C')\\
&\overset{3}{\leq}1-\sum_{k=0}^{n_H-1} (1-p_C')^{k+1} \mu_G \mu_k=1-\sum_{k=0}^{n_H-1} \frac{\nu_k}{R_0^{k+1}}=0,
\end{align*}
as $R_0$ satisfies (\ref{g0H}).  Hence, $R_0 (p_C') \overset{3}{\geq} 1$, since $R_0 (p_C')$ is the unique positive solution of $g_{0,p_C'} (\lambda)=0$. 
Thus $p_C \overset{3}{\geq} 1-R_0^{-1}$ and, recalling (\ref{rv}), $R_V \overset{3}{\geq} R_0$.

\item[(iv)]
\textit{$R_0$ and $R_I$}

For a growing epidemic, we know that both $R_0$ and $R_I$ are strictly greater than 1, and we have proved above that $R_I \overset{2}{\ge} R_V \overset{3}{\ge} R_0$, so we need consider only a declining epidemic.  If $R_I < 1$,  then it follows from (\ref{gind}) and (\ref{g0H}) that $g_0 (R_I) \overset{2}{\leq} g_I (R_I)=0$, whence $R_0 \overset{2}{\ge} R_I$.

\item[(v)]
\textit{$R_0$ and $R_{H\!I}$.}

Let $h_0(\lambda)=g_{H\!I}(\lambda)-g_0(\lambda)$.  We show that, for $\lambda> a$, $\sign(h_0(\lambda))=\sign(\lambda-1)$.  It then
follows that in a growing epidemic $R_0 > R_{H\!I}$ and in a declining epidemic $R_0<R_{H\!I}$,
since $g_0$ and $g_{H\!I}$ are each strictly increasing on their respective domains.   Note that,
since $g_0(\lambda)=\sum_{n=1}^{n_H} \pi_n g_0^{(n)}(\lambda)$ and $g_{H\!I}(\lambda)=\sum_{n=1}^{n_H} \pi_n g_{H\!I}^{(n)}(\lambda)$, it is sufficient to show,
for each $n=2,3,\cdots,n_H$, that $\sign(h_0(\lambda))=\sign(\lambda-1)$ when all the households have size $n$. (It is easily verified that $g_0^{(1)}(\lambda)=g_{H\!I}^{(1)}(\lambda)=1-\mu_G/\lambda$, so households of size $1$ do not contribute to $h_0(\lambda)$.)  Thus we now assume that all households have size $n$, where $n \ge 2$. To ease the exposition, we suppress the explicit dependence on $n$.

It follows directly from (\ref{g0H}) and (\ref{gnHI}) that
\begin{align}
\label{hlambda}
h_0(\lambda) &= \frac{\mu_G}{\lambda-a}\left\{\left[\sum_{k=0}^{n-1}\frac{\mu_k(\lambda-a)}{\lambda^{k+1}}\right]-1\right\}\nonumber\\
&= \frac{\mu_G}{\lambda(\lambda-a)}\left\{\left[\sum_{k=1}^{n-1}\frac{\mu_k-a\mu_{k-1}}{\lambda^{k-1}}\right]-\frac{a\mu_{n-1}}{\lambda^n}\right\}\nonumber\\
&=-\frac{\mu_G}{\lambda(\lambda-a)} f(\lambda^{-1}),
\end{align}
where $f$ is the polynomial of degree $n-1$ given by $f(x)=\sum_{k=0}^{n-1}c_k x^k$, with 
$$c_k=a\mu_k-\mu_{k+1}  \mbox{ for $k=0,1,\cdots,n-2$} \qquad \mbox{and} \qquad c_{n-1}=a\mu_{n-1}.$$

Now $f(1)=a\sum_{k=0}^{n-1}\mu_k-\sum_{k=1}^{n-1}\mu_k=a(1+\mu_H)-\mu_H=0$, so
\begin{equation}
f(x)=(x-1)\tilde{f}(x),
\label{fnfntilde}
\end{equation}
where $\tilde{f}(x)$ is a polynomial of degree $n-2$, say
\begin{equation}
\label{fntildepower}
\tilde{f}(x)=\sum_{k=0}^{n-2}\tilde{c}_k x^k.
\end{equation}
Substituting (\ref{fntildepower}) into (\ref{fnfntilde}) yields, after equating coefficients of powers of $x$, that, for $k=0,1,\cdots,n-2$,
\begin{equation}
\label{ctilde}
\tilde{c}_k=\sum_{j=k+1}^{n-1} c_j=a\sum_{j=k+1}^{n-1}\mu_j-\sum_{j=k+2}^{n-1}\mu_j,
\end{equation}
where the final sum is zero if $k=n-2$.

Let $n_0=\max(k:\mu_k>0)$ and note that $n_0 \ge 1$, since otherwise $\mu_H=0$.
Then $\tilde{c}_{n_0-1}=a\mu_{n_0}>0$ and $\tilde{c}_k=0$ for $k \ge n_0$. Thus,
to complete the proof we show that $\tilde{c}_k \ge 0$ $(k=0,1,\cdots,n_0-2)$, since then  (\ref{hlambda}), (\ref{fnfntilde}) and (\ref{fntildepower}) imply that $$\sign(h_0(\lambda))=\sign(\lambda-1).$$

Recall that $a=\sum_{j=1}^{n-1}\mu_j/(1+\sum_{j=1}^{n-1}\mu_j)$, which on substituting into (\ref{ctilde}) shows that, for $i=0,1,\cdots,n-3$, $\tilde{c}_i>0$ if and only if
\begin{equation}
\sum_{j=k+2}^{n-1}\mu_j<\mu_{k+1}\sum_{j=1}^{n-1}\mu_j.
\label{muinequ}
\end{equation}

To prove (\ref{muinequ}), construct a realisation of a single-household epidemic
using the epidemic graph $\mathcal{G}^{(n)}$.
Let $Y_0,Y_1,\cdots,Y_{n-1}$ denote the sizes of the successive generations of infectives.
Then, for $k=0,1,\cdots,n_0-2$,
\begin{align}
\sum_{j=k+2}^{n-1}\mu_j&=\mathbb{E}\left[\sum_{j=k+2}^{n-1} Y_j\right] \nonumber\\
&=\mathbb{E}\left[\mathbb{E}\left[\sum_{j=k+2}^{n-1} Y_j|Y_0, Y_1,\cdots,Y_{k+1}\right]\right]\nonumber\\
&\le \mathbb{E}[Y_{k+1} \mathbb{E}[\Upsilon_{k+1}|Y_0, Y_1,\cdots,Y_{k+1}]],\label{muleeyz}
\end{align}
where $\Upsilon_{k+1}$ is the total number of infectives in generations $k+2,k+3,\cdots,n-1$ that are descended from (i.e.~in the epidemic graph have chain of directed edges from) a typical generation-$(k+1)$ infective.  Note that an infective in generation $j>k+1$ may be descended from
more than one generation-$(i+1)$ infective, hence the inequality in (\ref{muleeyz}).  Further,
$\Upsilon_{k+1}|Y_0, Y_1,\cdots,Y_{k+1}$ is distributed as the total number of infectives, $\Upsilon$ say, in generations $1,2,\cdots,n-2-k$ of a single-household epidemic with initially $1$ infective and $n-(Y_0+Y_1+\cdots+Y_{k+1})$ susceptibles.  Now $\Upsilon$ is stochastically strictly less than the total number of infectives in generations $1,2,\cdots,n-2-k$ of such an epidemic with initially $1$ infective and $n-1$ susceptibles, so
\begin{equation*}
\mathbb{E}[\Upsilon_{k+1}|Y_0, Y_1,\cdots,Y_{k+1}] < \sum_{j=1}^{n-2-k}\mu_j,
\end{equation*}
and (\ref{muleeyz}) yields
\begin{equation}
\label{muinequ1}
\sum_{j=k+2}^{n-1}\mu_j < \mathbb{E}\left[Y_{k+1}\sum_{j=1}^{n-2-k}\mu_j\right]=\mu_{k+1}\left(\sum_{j=1}^{n-2-k}\mu_j\right) \le \mu_{k+1} \sum_{j=1}^{n-1}{\mu_j},
\end{equation}
proving (\ref{muinequ}).

\item[(vi)]
\textit{$R_I$ and $R_2$.}

Let $h_I(\lambda)=g_2(\lambda)-g_I(\lambda)$.  From~\eqref{gind} and~\eqref{muLunequal},\newline $g_I(\lambda)=1-\frac{\mu_G}{\lambda}\left(1+\sum_{n=2}^{n_H} \pi_n \frac{\mu_H^{(n)}}{\lambda}\right)$, whence, recalling~\eqref{g2unequal},
\begin{eqnarray*}
\label{h1}
h_I(\lambda)&=&\frac{\mu_G}{\lambda}\sum_{n=2}^{n_H}\pi_n\left(\frac{\mu_H^{(n)}}{\lambda}-\frac{\mu_1^{(n)}}{\lambda-b^{(n)}}\right)\\ \nonumber
&=& \mu_G \frac{\lambda-1}{\lambda^2}\sum_{n=2}^{n_H}\pi_n \frac{\mu_H^{(n)}-\mu_1^{(n)}}{\lambda-b^{(n)}},
\end{eqnarray*}
since $b^{(n)}=1-\left(\mu_1^{(n)}/\mu_H^{(n)}\right)$. Now $\mu_H^{(2)}=\mu_1^{(2)}$ so, if $n_H=2$, then $h_I(\lambda) \equiv 0$
and $R_I=R_2$.  If $n_H>2$, then $\sign(h_I(\lambda))=\sign(\lambda-1)$ for $\lambda>b=\max\left(b^{(2)},b^{(3)},\cdots,b^{(n_H)}\right)$, since $\mu_H > \mu_1$, and, similar to the comparison of $R_0$ and $R_{H\!I}$, it follows that
in a growing epidemic $R_I > R_2$ and in a declining epidemic $R_I < R_2$.
\item[(vii)]
\textit{$R_2$ and $R_{H\!I}$.}

Let $h_2(\lambda)=g_{H\!I}(\lambda)-g_2(\lambda)$.  As in the proof of the comparison of $R_0$ and $R_{H\!I}$, it is sufficient to assume that all households have size $n$, where $n \ge 2$, and show that $\sign(h_2(\lambda))=\sign(\lambda-1)$ for $\lambda>\max(a,b)$.  Equations (\ref{gnHI}) and (\ref{gn2}) imply that, for $\lambda>\max(a,b)$,
\begin{align}
\label{hhatlambda}
h_2(\lambda) &=
\mu_G\left[\frac{1}{\lambda}+\frac{\mu_1}{\lambda(\lambda-b)}-\frac{1}{\lambda-a}\right]
\nonumber\\
&=
\frac{\mu_G}{\lambda(\lambda-a)(\lambda-b)}\left[\lambda(\mu_1-a)-a(\mu_1-b)\right].
\end{align}
It follows from (\ref{gnHI}) and (\ref{gn2}) that $g_{H\!I}(1)=g_2(1)=1-R_*$, so
$h_2(1)=0$, whence $\mu_1-a=a(\mu_1-b)$.  (The latter is easily checked directly using the definitions of $a$ and $b$.)  Setting $i=0$ in (\ref{muinequ}), recalling that $a=\sum_{k=1}^{n-1} \mu_k/(1+\sum_{k=1}^{n-1} \mu_k)$ and rearranging shows that $\mu_1>a$.  (Note that this implies $b<\mu_1$, as claimed after (\ref{g2}).)
Substituting $\mu_1-a=a(\mu_1-b)$ into (\ref{hhatlambda}) then shows that, for $\lambda>\max(a,b))$, $\sign(h_2(\lambda)=\sign(\lambda-1)$, which completes the proof. 

\end{enumerate}
\end{proof}

\subsection{Proof of Theorem~\ref{propos}}
\label{subsec:HRrproof}

In this subsection we prove Theorem~\ref{propos}, which we restate here.
\begin{duplicate}
\begin{itemize}
\item[(a)] For all choices of infectivity profile $\mathcal{I}(t)$  $(t \geq 0)$,
\[
R_r=1 \iff \widetilde{R}_r=1 \iff R_0 = 1.
\]
\item[(b)] If $\mathcal{I}(t) = J \gtpdf(t)$ ($t \geq 0$), where $J$ is a non-negative random variable,
then in a growing epidemic,
\[
R_r \ge R_0 >1,
\]
and in a declining epidemic,
\[
R_r \le R_0 <1.
\]
\item[(c)] If $\mathcal{I}(t) =f(t) \indic(T_I>t)$ ($t \geq 0$), where $f(t)$ is a deterministic function and $T_I$ a non-negative random variable, then in a growing epidemic,
\[
R_r \ge \widetilde{R}_r \ge R_0 >1,
\]
and in a declining epidemic,
\[
R_r \le \widetilde{R}_r \le R_0 <1.
\]

\end{itemize}
The above results still hold if a latent period independent of the remainder of the infectivity profile is added.
\end{duplicate}

\begin{proof}
To prove part \thp{a} of Theorem~\ref{propos}, note that
\begin{equation}
\label{Lbetaat0}
\mathcal{L}_{\beta_H}(0)=\widetilde{\mathcal{L}}_{\beta_H}(0)=\mathcal{L}_{\beta_H}^{(0)}(0)=\mu_G(1+\mu_H)=R_*.
\end{equation}
Let $r, \tilde{r}$ and $r^{(0)}$ be the unique real solutions of
\begin{equation}
\label{Lbetardef}
\mathcal{L}_{\beta_H}(\theta)=1, \quad \widetilde{\mathcal{L}}_{\beta_H}(\theta)=1 \quad \mbox{and}\quad \mathcal{L}_{\beta_H}^{(0)}(\theta)=1,
\end{equation}
respectively.  Then,
\begin{equation}
\label{RrsfromM}
R_r=\frac{1}{\mathcal{M}_\gtrv(r)}, \quad \widetilde{R}_r=\frac{1}{\mathcal{M}_\gtrv(\tilde{r})}\quad \mbox{and}\quad R_0=\frac{1}{\mathcal{M}_\gtrv(r^{(0)})}.
\end{equation}
The functions $\mathcal{L}_{\beta_H}(\theta), \widetilde{\mathcal{L}}_{\beta_H}(\theta)$ and $\mathcal{L}_{\beta_H}^{(0)}(\theta)$ are each strictly decreasing in $\theta$, so~\eqref{Lbetaat0} implies that
\[
\sign(r)=\sign(\tilde{r})=\sign(r^{(0)})=\sign(R_*-1),
\]
so, since $\mathcal{M}_\gtrv(\theta)$ is strictly decreasing in $\theta$ and $\mathcal{M}_\gtrv(0)=1$,
\[
\sign(R_r-1)=\sign(\widetilde{R}_r-1)=\sign(R_0-1)=\sign(R_*-1).
\]
This proves part \thp{a} and shows also that $R_r, \widetilde{R}_r, R_0$ and $R_*$ are all strictly greater than 1 in a growing epidemic and strictly less than 1 in a declining epidemic.

For ease of presentation, we assume first that all households have the same size $n$ and we drop the explicit dependence of $\lambda^{(n)}_H$ on $n$. We further assume $n\geq 2$ in order to avoid trivial cases. We outline at the end of the proof how it extends to variable household sizes.

Consider a local epidemic started by a single initial infective. Label the individuals $0,1, \cdots, n-1$, with individual 0 being the initial infective in the household.  Construct the augmented random graph derived from $\mathcal{G}^{(n)}$ as described in Section~\ref{Hmodel}. For future reference we refer to this augmented graph as $\tilde{\mathcal{G}}^{(n)}$.   Recall that $t_{ii'}$ denotes the time of this first contact (since $i$'s infection) to $i'$; we refer to $t_{ii'}$ as the real infection interval for $i$ to infect $i'$.
For $i=1,2,\cdots,n-1$, if individual $i$ is infected by the local epidemic, then his/her time of infection, denoted by $T_i$, is given by the minimum of the sums of the real infection intervals between every pair of linked individuals along
all directed paths from the initial infective to $i$; if $i$ is not infected by the local epidemic then we set $T_i = \infty$. We set $T_0=0$. This fully specifies the real-time construction of the epidemic.

The overall expected household infectivity profile $\beta_H(t)$ can be decomposed as $\sum_{i=0}^{n-1} \beta_{H,i}(t)$, where $\beta_{H,i}(t) = \mu_G \bbE\left[ \mathcal{I}_i(t-T_i) \right]$ is the contribution from individual $i$ and $\mathcal{I}_i(t)$ is his/her infectivity profile (with $\mathcal{I}_i(t) =0$ if $t<0$).
Now, $\int_0^{\infty} e^{-\theta t} \mathcal{I}_i(t-T_i) \mathrm{d}t = e^{-\theta T_i} \int_0^{\infty} e^{-\theta t} \mathcal{I}_i(t) \mathrm{d}t $ for $T_i < \infty$, while the integral is 0 for $T_i = \infty$.
Therefore, noting that $T_i$ and $\{\mathcal{I}_i(t):t\geq 0\}$ are independent, we have
\begin{eqnarray}
\label{decomposeBetaH}
\mathcal{L}_{\beta_H}(\theta) & = & \int_0^{\infty} \beta_H(t){\rm e}^{-\theta t} \mathrm{d}t = \sum_{i=0}^{n-1} \int_0^{\infty} \beta_{H,i}(t)  {\rm e}^{-\theta t}  \mathrm{d}t \nonumber \\
& = & \sum_{i=0}^{n-1} \bbE\left[ \mathrm{e}^{-\theta T_i} \right] \int_0^{\infty}  \mu_G \bbE[ \mathcal{I}_i(t)]  {\rm e}^{-\theta t}  \mathrm{d}t\nonumber\\
& =&\mu_G \mathcal{M}_\gtrv(\theta) \sum_{i=0}^{n-1} \bbE\left[ \mathrm{e}^{-\theta T_i} \right],
\end{eqnarray}
since $ \int_0^{\infty}  \bbE[ \mathcal{I}_i(t)] \mathrm{e}^{-\theta t} \mathrm{d}t= \mathcal{M}_\gtrv(\theta)$. Let $\chi_k(i)$ be the event that individual $i$ is in generation $k$ $(i,k =0,1,\cdots n-1)$.  Then,
\begin{equation}
\label{elaborationp23}
\sum_{i=1}^{n-1}  \bbE[{\rm e}^{-\theta T_i}] = \sum_{i=1}^{n-1}\sum_{k=1}^{n-1}\mathbb{P}(\chi_k(i))   \bbE[{\rm e}^{-\theta T_i}|\chi_k(i)]
 =  \sum_{k=1}^{n-1} \mu_k   \bbE[{\rm e}^{-\theta T_1}|\chi_k(1)],
\end{equation}
since $\bbE[{\rm e}^{-\theta T_i}|\chi_k(i)]$ is independent of $i$ and $\mu_k = \sum_{i=1}^{n-1}\mathbb{P}(\chi_k(i))$.
Hence, using~\eqref{decomposeBetaH}, and recalling that $T_0=0$,
\begin{equation}
\label{Lbetatheta}
\mathcal{L}_{\beta_H}(\theta)=\mu_G \mathcal{M}_\gtrv(\theta)\left\{1+ \sum_{k=1}^{n-1} \mu_k   \bbE[{\rm e}^{-\theta T_1}|\chi_k(1)]\right\}.
\end{equation}

Suppose that an individual, $i$ say, is in household generation $k$. Then there exists at least one path of length $k$, and no shorter path, from the initial case in the household to $i$.  Consider one such path and relabel the individuals so that the successive individuals in that path are $0,1, \cdots, k$, so our given individual now has label $k$.  Let $\hat{T}_k=\sum_{i'=0}^{k-1} t_{i',i'+1}$ and observe that, since $\hat{T}_k$ is defined using the ``length'' of a specific path and $T_k$ is the minimum length over a collection of possible paths, we have
\begin{equation}
\label{Tgbound}
T_k \le \hat{T}_k.
\end{equation}

Consider case \thp{b} of Theorem~\ref{propos}, in which $\mathcal{I}(t) = J \gtpdf(t)$ ($t \geq 0$), where $J$ is a non-negative random variable.  In this case the  augmented random graph $\tilde{\mathcal{G}}^{(n)}$ may be constructed by independently for each individual, $i$ say, first sampling $J_i$ according to $J$ and then, conditional on $J_i$, letting $N_{ii'}$ $(i' \ne i)$ be independent Poisson random variables
having mean $J_i$.  The random variable $N_{ii'}$ gives the number of infectious contacts individual $i$ makes towards individual
$i'$.  If $N_{ii'}>0 $ then the times of these contacts, relative to $i$'s time of infection are given by $\gtpdfij(1),\gtpdfij(2),\cdots,\gtpdfij(N_{ii'})$, which are mutually independent realisations of the random variable $\gtrv$.  Suppose that $i$ makes infectious contact with $i'$, so $N_{ii'}>0$.   Then $t_{ii'}=\min\left(\gtpdfij(1),\gtpdfij(2),\cdots,\gtpdfij(N_{ii'})\right)$ and, in particular,
$t_{ii'} \le \gtpdfij(1)$.  Since $\gtpdfij(m)$ $(i \ne i', m=1,2,\cdots, N_{ii'})$ are mutually independent it follows that
\begin{equation}
\label{casecThat}
 \hat{T}_k \overset{st}{\le} \sum_{l=0}^{k-1} \gtrvz{l},
\end{equation}
where $\gtrvz{0}, \gtrvz{1}, \cdots, \gtrvz{k-1}$ are independent and identically distributed copies of $\gtrv$.  Thus, if $\theta>0$, then~\eqref{Lbetatheta},~\eqref{Tgbound} and~\eqref{casecThat} imply that
\begin{equation}
\label{betaHcaseb}
\mathcal{L}_{\beta_H}(\theta)\ge \mu_G \mathcal{M}_\gtrv(\theta)\left\{1+ \sum_{k=1}^{n-1} \mu_k \mathcal{M}_\gtrv(\theta)^k\right\}=\mathcal{L}_{\beta_H}^{(0)}(\theta),
\end{equation}
where $\mathcal{L}_{\beta_H}^{(0)}$ is defined at~\eqref{rtgrowth2}, with the opposite inequality holding if $\theta<0$.

Suppose that the epidemic is growing, so $R_*>1$.  Then,~\eqref{betaHcaseb} states that, for $\theta>0$,
\[
\mathcal{L}_{\beta_H}(\theta) \ge \mathcal{L}_{\beta_H}^{(0)}(\theta),
\]
so, since $\mathcal{L}_{\beta_H}(\theta)$ and $\mathcal{L}_{\beta_H}^{(0)}(\theta)$ are strictly decreasing in $\theta$ and $\mathcal{L}_{\beta_H}(0)= \mathcal{L}_{\beta_H}^{(0)}(0)=R_*$ it follows that $0<r^{(0)}\le r$, whence, since $\mathcal{M}_\gtrv(\theta)$ is strictly decreasing in $\theta$,~\eqref{RrsfromM} implies that $R_r \ge R_0>1$.  A similar argument shows that  $R_r \le R_0<1$ in a declining epidemic. 

Turn now to case \thp{c} of Theorem~\ref{propos}, in which $\mathcal{I}(t) =f(t) \indic(T_I>t)$ ($t \geq 0$), where $f(t)$ is a deterministic function and $T_I$ a non-negative random variable.  For $m=0,1,\cdots,n-1$, let the random variable $\gtrvzt{m}$ be distributed
as the time of the first infectious contact between two given individuals (say from $i$ to $i'$), given that there is at least one such contact and that $i$ does not contact $m$ other given individuals.  Note that the condition that $i$ does not contact 0 other given individuals is necessarily satisfied, so $\gtrvzt{0}$ has the same distribution as the random variable $\gtrvt$ defined in Section~\ref{subsec:Rr}.

Observe that if we know that an individual is in household generation $k$, then there exists at least one path of length $k$, and no shorter path, from the initial case in the household to that individual. If we know one path of length $k$ and we condition on knowing all edges in the epidemic generating graph in the household, not starting at one of the individuals in that path, then we know that a contact is made from an individual to the next individual in the path and some other contacts are not made (namely contacts which would lead to paths of shorter length than $k$, e.g.~a contact from an individual to an individual more than one place further along the path). We call the latter contacts ``forbidden'' (see Figure \ref{epigraphforbfig}).

\begin{figure}[!ht] \begin{center}
\includegraphics[width=.3\textwidth]{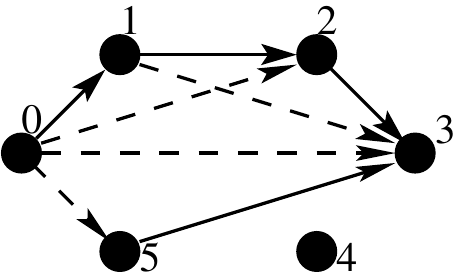}
 \caption{Epidemic graph, with relevant contacts represented by solid arrows. We consider the path from 0 to 3 (via 1 and 2). The forbidden contacts, which would make the path from 0 to 3 shorter are represented by dashed arrows. In this Figure, there are 3 forbidden contacts emanating from individual 0 and 1 from individual 1.}\label{epigraphforbfig}\end{center}
\end{figure}

Suppose that an individual, $i$ say, is in household generation $k$. Select a path of length $k$ from the initial case to individual $i$ and, as above, relabel the individuals so that the successive individuals in that path are $0,1, \cdots, k$.
 Denote the configuration, in the epidemic graph, of all edges not emanating from any node in the path \emph{plus} all those in the path itself by $\Xi$ (the solid arrows in Figure \ref{epigraphforbfig}) and let $m(i') = m(i',\Xi)$, $ (i'=0,1,\cdots,k-1)$ be the number of forbidden contacts from individual $i'$ under configuration $\Xi$. (In Figure \ref{epigraphforbfig}, $m(0) =3$, $m(1) =1$ and $m(2)=0$.)  Recall that $\hat{T}_k$ is the time it takes for individual $k$ to be infected along this path.  Then, as contacts emanating from different individuals are independent,
\begin{equation}
\label{Tghatmgf}
\mathbb{E}\left[{\rm e}^{-\theta \hat{T}_k}\right]  =  \mathbb{E}_{\Xi}\left[\prod_{i'=0}^{k-1} \mathbb{E}\left[{\rm e}^{-\theta \gtrvzt{m(i')}}|\Xi\right]\right], \qquad \theta \in (-\infty,\infty).
\end{equation}
(The notation $\mathbb{E}_{\Xi}$ denotes that the expectation is with respect to the distribution of $\Xi$.
Since the path of length $k$ is fixed, the randomness in $\Xi$ is contained in the distribution of edges 
in the epidemic graph that emanate from nodes not in the path.) 
We now show that
\begin{equation}
\label{Wktilde}
\gtrvzt{m} \overset{st}{\le} \gtrvt \quad \mbox{for all }  m=1,2,\cdots, n-1.
\end{equation}

For $m=0,1,\cdots,n-1$, let $\mathcal{D}_m$ be the event that there is at least one contact between two given individuals (say $i$ and $i'$) and $m$ other specified individuals are not contacted by $i$.  Let $\mathcal{D}$ be the event that there is at least one contact between two given individuals, so $\mathcal{D}=\mathcal{D}_0$.  Note that the probability of $\mathcal{D}_m$ depends on the infectious profile $\mathcal{I}_i$ ($=\{\mathcal{I}_i(t):t\geq 0\}$).  By the definition of conditional expectation
and noting that $\mathcal{P}(\mathcal{D}_m)=\mathbb{E}_{\mathcal{I}}[\indic(\mathcal{D}_m)]$, we have that, for $m=0,1,\cdots,n-1$,
\begin{equation}
\label{Wtildist}
\mathbb{P}(\gtrvzt{m} \leq t) = \mathbb{E}_{\mathcal{I}}\left[\indic(\mathcal{D}_m) \frac{\int_0^t {\rm e}^{-\int_0^{s}\lambda_H \mathcal{I}(x) {\rm d}x} \mathcal{I}(s) {\rm d}s}{\int_0^{\infty} {\rm e}^{-\int_0^{s}\lambda_H \mathcal{I}(x) {\rm d}x} \mathcal{I}(s) {\rm d}s}\right]/\mathbb{E}_{\mathcal{I}}[\indic(\mathcal{D}_m)] \qquad (t\geq 0).
\end{equation}
Making the substitution $u(t)=\int_0^t \mathcal{I}(x){\rm d}x$ 
we obtain $$\int_0^t {\rm e}^{-\int_0^{s}\lambda_H \mathcal{I}(x) {\rm d}x} \mathcal{I}(s) {\rm d}s = \int_0^t {\rm e}^{-\lambda_H u(s)} \frac{{\rm d}u(s)}{{\rm d}s} {\rm d}s = 1-{\rm e}^{-\lambda_H u(t)} = 1-{\rm e}^{- \lambda_H\int_0^t \mathcal{I}(s) {\rm d}s}$$
and similarly $\int_0^{\infty} {\rm e}^{-\int_0^{s}\lambda_H \mathcal{I}(x) {\rm d}x} \mathcal{I}(s) ds = 1-{\rm e}^{- \lambda_H\int_0^{\infty} \mathcal{I}(s) {\rm d}s}$.
Furthermore, $\mathcal{D}_m \subseteq \mathcal{D}$ and $$\mathbb{P}(\mathcal{D}_m|\int_0^{\infty} \mathcal{I}(s) {\rm d}s) = {\rm e}^{-m \lambda_H \int_0^{\infty}\mathcal{I}(s) {\rm d}s}.$$
Combining these observations yields
that, for $m=0,1,\cdots,n-1$, \eqref{Wtildist} can be rewritten as
\begin{equation}
\label{tilWkdist}
\mathbb{P}(\gtrvzt{m} \leq t) = \frac{\mathbb{E}_{\mathcal{I}|\mathcal{D}}\left[{\rm e}^{-m \lambda_H \int_0^{\infty}\mathcal{I}(s) {\rm d}s} \frac{1-{\rm e}^{- \lambda_H\int_0^t \mathcal{I}(s) {\rm d}s}}{1-{\rm e}^{- \lambda_H\int_0^{\infty} \mathcal{I}(s) {\rm d}s}}|\mathcal{D}\right]}{\mathbb{E}_{\mathcal{I}|\mathcal{D}}[{\rm e}^{-m \lambda_H \int_0^{\infty}\mathcal{I}(s) {\rm d}s}|\mathcal{D}]} \qquad (t\geq 0).
\end{equation}
(The notation $\mathbb{E}_{\mathcal{I}|\mathcal{D}}$ denotes that the expectation is with respect to the
distribution of the infectivity profile $\mathcal{I}=\{\mathcal{I}(t):t\geq 0\}$ of an infective given that
that infective makes at least one contact with a given individual.)
The distribution function of $\gtrvt$ may be obtained by setting $m=0$ in~\eqref{tilWkdist}.  Hence,
$\gtrvzt{m}\:\overset{st}{\mathop{\le }}\,\gtrvt$ if and only if, for all $t>0$,
\begin{multline}
\label{Harrisuse2}
\mathbb{E}_{\mathcal{I}|\mathcal{D}}\left[{\rm e}^{-m \lambda_H\int_0^{\infty}\mathcal{I}(s) {\rm d}s} \frac{1-{\rm e}^{- \lambda_H\int_0^t \mathcal{I}(s) {\rm d}s}}{1-{\rm e}^{- \lambda_H\int_0^{\infty} \mathcal{I}(s) {\rm d}s}}|\mathcal{D}\right] \\
\geq \mathbb{E}_{\mathcal{I}|\mathcal{D}}\left[{\rm e}^{-m \lambda_H\int_0^{\infty}\mathcal{I}(s) {\rm d}s} |\mathcal{D}\right]\mathbb{E}_{\mathcal{I}|\mathcal{D}}\left[\frac{1-{\rm e}^{- \lambda_H\int_0^t \mathcal{I}(s) {\rm d}s}}{1-{\rm e}^{- \lambda_H\int_0^{\infty} \mathcal{I}(s) {\rm d}s}}|\mathcal{D}\right].
\end{multline}

Recall Chebychev's `other' inequality (also referred to as Harris' inequality) (Hardy~\cite{Hardy1952}, p.~168), which states
that if $f_1(x)$ and $f_2(x)$ are both increasing or both decreasing functions and $X$ is a random variable, then
\begin{equation}
 \label{Harris}
 \mathbb{E}[f_1(X)f_2(X)] \geq \mathbb{E}[f_1(X)]\mathbb{E}[f_2(X)].
\end{equation}
From the proof of this inequality it follows immediately that the inequality is strict if both $f_1(X)$ and $f_2(X)$ have strictly positive variance, which is the case if (i) both functions are strictly increasing or both functions are strictly decreasing and (ii) $\mathrm{Var}(X)>0$.
We now apply Chebychev's `other' inequality to conditional expectations.
In case \thp{c} of Theorem~\ref{propos}, for $t\in (0,\infty]$ we have $\int_0^{t} \mathcal{I}(s) {\rm d}s = \int_0^{\min(T_I,t)} f(s) {\rm d}s$ and we observe that both $$f_1(x)= {\rm e}^{-m  \lambda_H\int_0^{x}\mathcal{I}(s) {\rm d}s}$$ and
$$f_2(x) = \frac{1-{\rm e}^{-\lambda_H\int_0^{\min(x,t)} \mathcal{I}(s) {\rm d}s}}{1-{\rm e}^{-\lambda_H\int_0^{x} \mathcal{I}(s) {\rm d}s}} = \indic(x<t) + \indic(x\geq t) \frac{1-{\rm e}^{-\lambda_H\int_0^t \mathcal{I}(s) {\rm d}s}}{1-{\rm e}^{-\lambda_H\int_0^{x} \mathcal{I}(s) {\rm d}s}}$$
are decreasing in $x$. Thus by \eqref{Harrisuse2} and \eqref{Harris} we have $\gtrvzt{m}\:\overset{st}{\mathop{\le }}\,\gtrvt$.

It follows from~\eqref{Wktilde} that $\mathcal{M}_{\gtrvzt{m}}(r) \ge \mathcal{M}_{\gtrvt}(r)$ if $r>0$ and $\mathcal{M}_{\gtrvzt{m}}(r) \le \mathcal{M}_{\gtrvt}(r)$ if $r<0$. It then follows using~\eqref{Lbetatheta},~\eqref{Tgbound} and~\eqref{Tghatmgf} that, if $r>0$, then
\begin{equation}
\label{Lequation2}
\mathcal{L}_{\beta_H}(r)\ge \mu_G \mathcal{M}_\gtrv(r)\left\{1+ \sum_{k=1}^{n-1} \mu_k \mathcal{M}_{\gtrvt}(r)^k\right\}=\widetilde{\mathcal{L}}_{\beta_H}(r),
\end{equation}
where $\widetilde{\mathcal{L}}_{\beta_H}$ is defined at~\eqref{rtgrowth3}, with the opposite inequality holding if $r<0$.
Arguing as for case \thp{b} above, shows that $R_r \ge \widetilde{R}_r >1$ in a growing epidemic and $R_r \le \widetilde{R}_r <1$
in a declining epidemic.

To compare $ \widetilde{R}_r$ and $R_0$, we need to show that
\begin{equation}
\label{MWcompar}
\mathcal{M}_{\gtrvt}(\theta)\ge \mathcal{M}_\gtrv(\theta) \mbox{ if $\theta>0$} \quad \mbox{and} \quad \mathcal{M}_{\gtrvt}(\theta)\le \mathcal{M}_\gtrv(\theta) \mbox{ if $\theta<0$.}
\end{equation}
Recall from Section \ref{subsec:Rr} that $\mathcal{M}_\gtrv(\theta) = \int_0^{\infty}{\rm e}^{-\theta t} \mathbb{E}[\mathcal{I}(t)]{\rm d}t$.
From the definition of $\gtrvt$ we see that
\begin{equation}
\label{PWGlet}
\mathbb{P}( \gtrvt \leq t) = \frac{\mathbb{E}_{\mathcal{I}}\left[ \int_0^t {\rm e}^{-\int_0^{s} \lambda_H \mathcal{I}(x) {\rm d}x }\lambda_H \mathcal{I}(s) {\rm d}s \right]}{\mathbb{E}_{\mathcal{I}}\left[ \int_0^{\infty} {\rm e}^{-\int_0^{s} \lambda_H \mathcal{I}(x) {\rm d}x }\lambda_H \mathcal{I}(s) {\rm d}s \right]}, \qquad t \geq 0.
\end{equation}
After first differentiating~\eqref{PWGlet} to obtain the density of $\gtrvt$, we obtain (using the dominated convergence theorem for a fully rigorous argument) that
$$\mathcal{M}_{\gtrvt}(\theta) = \frac{ \int_0^{\infty} {\rm e}^{-\theta t} \mathbb{E}_{\mathcal{I}}\left[ {\rm e}^{-\int_0^{t} \lambda_H \mathcal{I}(x) {\rm d}x } \mathcal{I}(t)  \right] {\rm d}t}{\mathbb{E}_{\mathcal{I}}\left[ \int_0^{\infty} {\rm e}^{-\int_0^{s} \lambda_H \mathcal{I}(x) {\rm d}x } \mathcal{I}(s) {\rm d}s \right]}, \qquad \theta \in (-\infty,\infty). $$ 
Now using that $\mathcal{I}(t) = f(t) \indic(T_I>t)$ for some random variable $T_I$ gives that for all real $\theta$,
$\mathcal{M}_\gtrv(\theta) = \int_0^{\infty}{\rm e}^{-\theta t} f(t)\mathbb{P}[T_I>t]{\rm d}t$ and
\begin{eqnarray*}
\mathcal{M}_{\gtrvt}(\theta) & = & \frac{ \int_0^{\infty} {\rm e}^{-\theta t} \mathbb{E}_{T_I}\left[ {\rm e}^{-\int_0^{t} \lambda_H f(x) {\rm d}x } f(t) \indic(T_I>t)  \right] {\rm d}t}{\mathbb{E}_{T_I}\left[ \int_0^{\infty} {\rm e}^{-\int_0^{s} \lambda_H f(x) {\rm d}x } f(s) \indic(T_I>s) {\rm d}s \right]}\\
\ & = & \frac{ \int_0^{\infty} {\rm e}^{-\theta t}  {\rm e}^{-\int_0^{t} \lambda_H f(x) {\rm d}x } f(t) \mathbb{P}[T_I>t]  {\rm d}t}{\int_0^{\infty} {\rm e}^{-\int_0^{s} \lambda_H f(x) {\rm d}x } f(s) \mathbb{P}[T_I>s] {\rm d}s }.
\end{eqnarray*}
Note that $f(t)\mathbb{P}[T_I>t] = \mathbb{E}[\mathcal{I}(t)] = \gtpdf(t)$, which is the density of $\gtrv$.
Using this definition of $\gtrv$, we obtain that
$$\mathcal{M}_\gtrv(\theta) = \mathbb{E}\left[{\rm e}^{-\theta \gtrv}\right] \qquad \mbox{and} \qquad \mathcal{M}_{\gtrvt}(\theta) = \frac{\mathbb{E}\left[{\rm e}^{-\theta \gtrv}{\rm e}^{-\int_0^{\gtrv} \lambda_H f(x) {\rm d}x}\right]}{\mathbb{E}\left[{\rm e}^{-\int_0^{\gtrv} \lambda_H f(x) {\rm d}x}\right]}. $$
For $\theta >0$, the inequality $\mathcal{M}_{\gtrvt}(\theta)\geq \mathcal{M}_\gtrv(\theta)$ now follows from applying Chebychev's `other' inequality (inequality \eqref{Harris}) to the functions $f_1(x) =  {\rm e}^{-\theta x}$ and $f_2(x) = {\rm e}^{-\int_0^{x} \lambda_H f(s) {\rm d}s}$. The result for $\theta < 0$ is proved in the same way, using $f_1(x) =  -{\rm e}^{-\theta x}$ instead. It follows using \eqref{rtgrowth2} and \eqref{rtgrowth3} that $\widetilde{\mathcal{L}}_{\beta_H}(\theta) \geq \mathcal{L}^{(0)}_{\beta_H}(\theta)$ if $\theta >0$ and
$\widetilde{\mathcal{L}}_{\beta_H}(\theta) \leq \mathcal{L}^{(0)}_{\beta_H}(\theta)$ if $\theta <0$, which implies that $\widetilde{R}_r \geq R_0$ in a growing epidemic and $\widetilde{R}_r \leq R_0$ in a declining epidemic.

Note that in case \thp{b} of Theorem~\ref{propos} (i.e.~when $\mathcal{I}(t) =J \gtpdf(t)$, with $J$ random and $\gtpdf(t)$ $(t \ge 0)$ non-random), $\gtrvzt{m}\:\overset{st}{\mathop{\le }}\,\gtrvt$ does not necessarily hold, because conditioning on the absence of some edges leads to relatively low realisations of $J$, which implies fewer infectious contacts even if there is at least one contact between given individuals, which in turn implies later first contacts. So, here we cannot conclude that, for example, $R_r \ge \widetilde{R}_r$ in a growing epidemic.

Note that in the variable household size setting \eqref{Lbetatheta} becomes
\begin{equation}
\label{Lbetathetavar}
\mathcal{L}_{\beta_H}(\theta)=\mu_G \mathcal{M}_\gtrv(\theta)\left\{1+ \sum_{n=2}^{n_H} \pi_n \sum_{k=1}^{n-1} \mu_k^{(n)}   \bbE[{\rm e}^{-\theta T_1^{(n)}}|\chi_k(1)]\right\},
\end{equation}
where $T_1^{(n)}$ is the time that individual $1$ is infected in a local epidemic in a household having size $n$. The expectations are conditioned on individual $1$ being in generation $k$.
The arguments leading to \eqref{casecThat} show that
\begin{equation}
\bbE[{\rm e}^{-\theta T_1^{(n)}}|\chi_k(1)]  \geq   [\mathcal{M}_\gtrv(\theta)]^k \qquad \mbox{for $\theta >0$,}
\end{equation}
and
\begin{equation}
\bbE[{\rm e}^{-\theta T_1^{(n)}}|\chi_k(1)]  \leq   [\mathcal{M}_\gtrv(\theta)]^k \qquad \mbox{for $\theta <0$.}
\end{equation}
It then follows using \eqref{sizebiasgen} that, if $\theta >0$, inequality \eqref{betaHcaseb} holds, with the opposite inequality holding if $\theta<0$. Case \thp{b} of Theorem~\ref{propos} now follows as before.

Turning to case \thp{c}, recall that if $\lambda_H^{(n)}$ varies with $n$ then so does the distribution of $\gtrvnt{n}$. In that case recall that $\gtrvnt{n}$ is a random variable distributed as $\gtrvt$ for a size-$n$ household and, for $m=0,1,\cdots n\!-\!1$, let $\gtrvznt{m}{n}$ be a random variable distributed as $\gtrvzt{m}$, again for a size-$n$ household. (Note that, cf.~\eqref{Lbetathetavar}, the distribution of $\gtrv$ is determined purely by the distribution of the infectivity profile $\{\mathcal{I}(t): t\geq 0\}$ and hence it is independent of household size $n$.)
Arguing as before shows that $\gtrvznt{m}{n} \overset{st}{\le}  \gtrv$ $(n=2,3,\cdots,n_H;m=0,1,\cdots, n-1)$ and using obvious extensions of \eqref{Tgbound} and \eqref{Tghatmgf} to the variable household size setting it follows from \eqref{Lbetathetavar} that, if $r>0$,  then \eqref{Lequation2} holds with $\widetilde{\mathcal{L}}_{\beta_H}(r)$ defined at \eqref{rtgrowth5}, and if $r<0$, then $\mathcal{L}_{\beta_H}(r)\leq \widetilde{\mathcal{L}}_{\beta_H}(r)$. As in the case when all households have the same size, this implies that $R_r \geq \widetilde{R}_r >1$ in a growing epidemic and $R_r \leq \widetilde{R}_r <1$ in a declining epidemic.
Finally the previous argument shows that for $n=2,3,\cdots n_H$, $\mathcal{M}_{\gtrvnt{n}}(\theta)\geq \mathcal{M}_{\gtrv}(\theta)$ if $\theta>0$ and $\mathcal{M}_{\gtrvnt{n}}(\theta)\leq \mathcal{M}_{\gtrv}(\theta)$ if $\theta<0$.
The comparison between $\widetilde{R}_r$ and $R_0$ then follows on noting that $\mathcal{L}_{\beta_H}^{(0)}(r)$ is obtained by replacing $\mathcal{M}_{\gtrvnt{n}}(r)$ by $\mathcal{M}_{\gtrv}(r)$ in equation \eqref{rtgrowth5}.

\end{proof}

\begin{rmk}
\label{remarkstrict}
Note that if $\gtpdf(t)$ is a proper density function, i.e.~$\gtpdf(t) < \infty$ for all $t \geq 0$, then the inequality in \eqref{betaHcaseb} is strict, so in that case the inequalities in Theorem~\ref{propos}\thp{b} are strict.
Note that for the infectivity profile considered in Theorem~\ref{propos}\thp{c}, the random variable $\gtrv$ necessarily has a proper density function, so the application of Chebychev's other inequality leads to strict inequalities in~\eqref{MWcompar}.  Thus, in this case, we have $\widetilde{R}_r>R_0$ in a growing epidemic and $\widetilde{R}_r<R_0$ in a declining epidemic. 

We have already observed that $\mathcal{L}_{\beta_H}(\theta) = \widetilde{\mathcal{L}}_{\beta_H}(\theta)$ and therefore $\widetilde{R}_r=R_r$ for $n_H \leq 2$, since there is only one possible path from the initial infective to the other individual in a household of size 2. In households of size 3, it depends on the distribution of $T_I$, whether the (real) time needed for the epidemic to traverse an infection path of length 2 might be shorter than the time needed to traverse an infection path of length 1. If this is impossible, then $\widetilde{R}_r=R_r$ for $n_H =3$, otherwise $\widetilde{R}_r \neq R_r$ unless they are both one.
In households of size 4 or larger, it is possible that there are two disjoint paths of length 2 from the initial infective in the household to an individual in generation 2, so in this case, $\mathcal{L}_{\beta_H}(\theta) = \widetilde{\mathcal{L}}_{\beta_H}(\theta)$, for all $\theta \in (-\infty,\infty)$, if and only if $\mathrm{Var}(\gtrvt)=0$. Thus, for $n_H \ge 4$, unless they are both one, $\widetilde{R}_r = R_r$ if and only if
$\mathrm{Var}(\gtrvt)=0$.
\end{rmk}



\subsection{Proof of Theorem~\ref{Theorem 5.2}}
\label{subsec:HWcompproof}

Where relevant, we now use the notation $\overset{n}{\geq}$ to denote that the inequality is strict if and only if $\max(n_H,n_W)>n$, so the statement of Theorem~\ref{Theorem 5.2} is as follows.

\begin{duplicate}
\label{Theorem 5.2a}
\begin{enumerate}
\item[(a)]
$R_*=1 \iff R_H =1 \iff R_I=1 \iff R_0 = 1 \Longrightarrow R_V = 1 $.

\item[(b)]
In a growing epidemic,
\[
R_* > R_H > R_I \overset{2}{\geq} R_V \overset{3}{\geq} R_0 > 1,
\]
and in a declining epidemic
\[
R_* < R_H < R_I \overset{2}{\leq} R_0 < 1.
\]

\item[(c)] Theorem~\ref{propos} holds also for the households-workplaces model.
\end{enumerate}

\end{duplicate}

\begin{proof}
We first prove \thp{a}. Let
\begin{equation}
\label{hwthreshpar}
f_{HW}(\mu_G,\mu_H,\mu_W)=\mu_G(1+\mu_H)(1+\mu_W)+\mu_H\mu_W,
\end{equation}
and note from~\eqref{hwrstar} that $R_*=1$ if and only if $f_{HW}(\mu_G,\mu_H,\mu_W)=1$.  Observe from~\eqref{gh} and ~\eqref{ghwi} that $g_H(1)=g_I^{(HW)}(1)=1-f_{HW}(\mu_G,\mu_H,\mu_W)$, so $R_*=1 \iff R_H = 1 \iff R_I=1$.

Turning to $R_0$, note that
\[
\mu_H\mu_W=\left(\sum_{\ell=1}^{n_H-1}\mu_\ell^H\right)\left( \sum_{{\ell'}=1}^{n_W-1}\mu_{\ell'}^W\right)
=\sum_{k=1}^{n_H+n_W-3} 
\left(\sum_{\ell=\max(1,k-n_W+2)}^{\min(k,n_H-1)}\mu_\ell^H\mu_{k+1-\ell}^W\right)
\]
and
\begin{eqnarray*}
(1+\mu_H)(1+\mu_W)&=&\left(\sum_{\ell=0}^{n_H-1}\mu_\ell^H\right) \left(\sum_{{\ell'}=0}^{n_W-1}\mu_{\ell'}^W\right)\\
&=&\sum_{k=0}^{n_H+n_W-2}
\left(\sum_{\ell=\max(0,k-n_W+1)}^{\min(k,n_H-1)}\mu_\ell^H\mu_{k-\ell}^W\right).
\end{eqnarray*}
Thus
\begin{equation}
\label{fHWc}
f_{HW}(\mu_G,\mu_H,\mu_W)=\sum_{k=0}^{n_H+n_W-2}c_k,
\end{equation}
where $c_k$ is defined in~\eqref{eq:coefficientsHW} ($c_0=\mu_G$) and, using~\eqref{gnHnW}, $g^{(HW)}_0 (1)=1-f_{HW}(\mu_G,\mu_H,\mu_W)$.  Hence, $R_0=1 \iff R_*=1$.
By definition, $R_V=1$ if $R_*=1$.

To prove \thp{b}, we first note that~\eqref{hwrstar} and~\eqref{hwthreshpar} imply that
\begin{equation}
\label{hwthreshcond}
\sign(R_*-1)=\sign(f_{HW}(\mu_G,\mu_H,\mu_W)-1).
\end{equation}
Thus, since the functions $g_H, g_I^{(HW)}$ and $g_0^{(HW)}$ are all strictly increasing on $(0,\infty)$, it follows that the reproduction numbers $R_*,R_H, R_I$ and $R_0$ are all strictly greater than 1 in a growing epidemic and all strictly smaller than 1 in a declining
epidemic. We consider now each of the comparisons in turn.

\begin{enumerate}
\item[(i)]
\textit{$R_*$ and $R_H$.}

Suppose that $R_*>1$.
Clearly, $R_*>R_H$ if $R_*=\infty$, so suppose that $R_*<\infty$.
An elementary calculation shows that
\[
g_H(R_*)=\frac{\mu_H}{\mu_G(1+\mu_H)^2}\left(f_{HW}(\mu_G,\mu_H,\mu_W)-1\right),
\]
whence, using (\ref{hwthreshcond}), $g_H(R_*)>0$.  It follows that $R_H<R_*$, since $g_H$ is strictly increasing on $(0,\infty)$ and $g_H(R_H)=0$.  A similar argument shows that $R_H>R_*$ if $R_*<1$.
\item[(ii)]
\textit{$R_H$ and $R_I$}

Elementary algebra gives
\[
g_I^{(HW)}(\lambda)-g_H(\lambda)=(\lambda-1)\left[\frac{\mu_H(\mu_G+\mu_W)}{\lambda^2}
+\frac{\mu_G\mu_H\mu_W}{\lambda^3}\right],
\]
so, for $\lambda>0$,
\begin{equation}
\label{signgI-gH}
\sign\left(g_I^{(HW)}(\lambda)-g_H(\lambda)\right)=\sign(\lambda-1).
\end{equation}
Recall that $R_H$ and $R_I$ are the unique roots in $(0,\infty)$ of $g_H$ and
$g_I^{(HW)}$, respectively.  Suppose that $R_H>1$.  Then, since $g_H(R_H)=0$, (\ref{signgI-gH}) implies that $g_I^{(HW)}(R_H)>0$,
whence $R_I<R_H$, since $g_I^{(HW)}(R_I)=0$ and $g_I^{(HW)}$ is increasing on $(0,\infty)$.  A similar argument shows that $R_I>R_H$ if $R_H<1$.

\item[(iii)]
\textit{$R_I$ and $R_V$}

Suppose that $R_I\!>\!1$ and a fraction $p$ of individuals is vaccinated with a perfect vaccine.  Then, as in the households model, $\mu_G, \mu_H$ and $\mu_W$ are reduced to $\mu_G(p), \mu_H(p)$ and $\mu_W(p)$, respectively, where $$\mu_G(p)=(1-p)\mu_G, \mu_H (p) \overset{2}{\leq} (1-p) \mu_H \qquad \mbox{and}
\qquad \mu_W (p) \overset{2}{\leq} (1-p) \mu_W,$$ and $R_I$ is reduced to $R_I(p)$, which is given by the unique solution of $g_{I,p}^{(HW)}(\lambda)=0$ in $(0,\infty)$, where
\begin{align*}
g_{I,p}^{(HW)}(\lambda)=&1-\frac{\mu_G(p)}{\lambda}-\frac{\mu_G(p)\mu_H(p)+\mu_G(p)\mu_W(p)+\mu_H(p)\mu_W(p)}{\lambda^2}\\
&-\frac{\mu_G(p)\mu_H(p)\mu_W(p)}{\lambda^3}.
\end{align*}
Suppose that $p>0$.  Then, the above inequalities imply that
\begin{align*}
g_{I,p}^{(HW)}((1-p)R_I)
&\overset{2}{\ge} 1-\frac{\mu_G}{R_I}-\frac{\mu_G\mu_H+\mu_G\mu_W+\mu_H\mu_W}{R_I^2}
-\frac{\mu_G\mu_H\mu_W}{R_I^3}\\
&=0,
\end{align*}
since $g_I^{(HW)}(R_I)=0$.  Hence, since $g_{I,p}^{(HW)}(R_I(p))=0$ and $g_{I,p}^{(HW)}$ is strictly increasing on $(0,\infty)$, $R_I(p) \overset{2}{\le} (1-p)R_I$.  In particular, $R_I(1-R_I^{-1})  \overset{2}{\le} 1$, whence $p_C \overset{2}{\le} 1-R_I^{-1}$ and, using (\ref{rv}), $R_V \overset{2}{\le} R_I$.

\item[(iv)]
\textit{$R_V$ and $R_0$}

Recall that $R_0$ is given by the unique root in $(0,\infty)$ of $g^{(HW)}_0$ defined at~\eqref{gnHnW}.  Suppose that $R_0>1$ and a fraction $p$ of the population is vaccinated with a perfect vaccine.  Then the post-vaccination basic reproduction number, $R_0(p)$ say, is given by the unique root in $(0,\infty)$ of  the function $g^{(HW)}_{0,p}$ defined by
\begin{equation*}
g^{(HW)}_{0,p} (\lambda)=1-\sum_{k=0}^{n_H+n_W-2} \frac{c_k(p)}{\lambda^{k}},
\end{equation*}
where $c_0(p)=\mu_G(p)$ and, for $k=1,2,\cdots,n_H+n_W-2$,  
\begin{align*}
c_k(p)=\mu_G(p)&\left(\sum_{\ell=\max(0,k-n_W+1)}^{\min(k,n_H-1)}\mu_\ell^H(p)\mu_{k-\ell}^W(p)\right)\\
&+\sum_{\ell=\max(1,k-n_W+2)}^{\min(k,n_H-1)} \mu_\ell^H(p)\mu_{k+1-\ell}^W(p),
\end{align*}
$\mu_G(p)=(1-p)\mu_G$, and, for $\ell=0,1,\cdots, \mu_\ell^H(p)$ (respectively $\mu_\ell^W(p)$) is the post-vaccination mean size of the $\ell$th generation in a typical single-household (respectively single-workplace) epidemic with one (unvaccinated) initial infective; the second sum in the expression for $c_k(p)$ is zero when $k=n_H+n_W-2$.

As in the proof of Theorem~\ref{Hcomp}, $$\mu_\ell^H(p){\ge} (1-p)^\ell
\mu_\ell^H\mbox{ }(\ell=0,1,\cdots,n_H-1)\mbox{ and }\mu_{\ell'}^W(p) {\ge} (1-p)^{\ell'} \mu_{\ell'}^W \mbox{ }({\ell'}=0,1,\cdots,n_W-1),$$ whence
$c_k(p) {\ge} (1-p)^{k+1} c_k$ $(k=1,2,\cdots,n_H+n_W-2)$.  Moreover, these inequalities are all equalities if $\max(n_H,n_W) \le 3$, otherwise at least one of them is strict.  Arguing exactly as in the proof of Theorem~\ref{Hcomp} shows that $R_V \overset{3}{\ge} R_0$.

\item[(v)]
\textit{$R_I$ and $R_0$}

We need to consider only a declining epidemic, since for a growing epidemic comparison of $R_I$ and $R_0$ follows from (iii) and (iv) above.  It is convenient to express~\eqref{hwthreshpar} and~\eqref{fHWc} as
\begin{equation}
\label{fHWc1}
\mu_G+\mu_G\mu_H+\mu_G\mu_W+\mu_G\mu_H\mu_W+\mu_H\mu_W=\sum_{k=0}^{n_H+n_W-2}c_k.
\end{equation}
Note that $c_0=\mu_G$ and that, in (\ref{fHWc1}), the contributions to $\mu_G\mu_H\mu_W$ come
from elements of the sums in $c_2,c_3,\cdots,c_{n_H+n_W-2}$ (and not $c_1$).  Thus, we may write
\[
\mu_G\mu_H+\mu_G\mu_W+\mu_H\mu_W=\sum_{k=1}^{n_H+n_W-2}c_k^{(1)} \quad \text{and}\quad
\mu_G\mu_H\mu_W=\sum_{k=2}^{n_H+n_W-2}c_k^{(2)},
\]
where $c_1^{(1)}=c_1$ and, for k=$2,3,\cdots,n_H+n_W-2$, $c_k=c_k^{(1)}+c_k^{(2)}$ with both
$c_k^{(1)}$ and $c_k^{(2)}$ being positive.  Now, from (\ref{gnHnW}),
\[
g_0^{(HW)}(\lambda)=1-\frac{\mu_G}{\lambda}-\sum_{k=1}^{n_H+n_W-2}\frac{c_k^{(1)}}{\lambda^{k+1}}
-\sum_{k=2}^{n_H+n_W-2}\frac{c_k^{(2)}}{\lambda^{k+1}}.
\]

Suppose that $R_I<1$.  Then,
\begin{align}
\label{g0hw<0}
g_0^{(HW)}(R_I)&=1-\frac{\mu_G}{R_I}-\sum_{k=1}^{n_H+n_W-2}\frac{c_k^{(1)}}{R_I^{k+1}}
-\sum_{k=2}^{n_H+n_W-2}\frac{c_k^{(2)}}{R_I^{k+1}}\nonumber\\
&\le 1-\frac{\mu_G}{R_I}-\frac{1}{R_I^2}\sum_{k=1}^{n_H+n_W-2}c_k^{(1)}
-\frac{1}{R_I^3}\sum_{k=2}^{n_H+n_W-2}c_k^{(2)}\\
&=1-\frac{\mu_G}{R_I}-\frac{\mu_G\mu_H+\mu_G\mu_W+\mu_H\mu_W}{R_I^2}
-\frac{\mu_G\mu_H\mu_W}{R_I^3}\nonumber\\
&=g_I^{(HW)}(R_I)\nonumber\\
&=0\nonumber,
\end{align}
by the definition of $R_I$.  Hence, $R_0 \ge R_I$, since $g_0^{(HW)}(R_0)=0$ and $g_0^{(HW)}$ is increasing on $(0,\infty)$. If $n_H=n_W=2$ then $c_1=\mu_G\mu_H+\mu_G\mu_W+\mu_H\mu_W$ and $c_2= \mu_G\mu_H\mu_W$, whence $g_0^{(HW)}=g_I^{(HW)}$ and $R_0=R_I$.  If $n_H>2$ or $n_W>2$ then it is readily seen using~\eqref{eq:coefficientsHW} that $c_3^{(2)}>0$, as $n_H+n_W-2 \ge 3$, 
which implies that the inequality (\ref{g0hw<0}) is strict, whence $R_0 \overset{2}{\ge} R_I$.
\end{enumerate}

Finally, we prove \thp{c}.  For $\theta \in (-\infty,\infty)$, let
\begin{equation*}
F_{HW}(\theta)=\mu_G \mathcal{M}_\gtrv(\theta) (\mathcal{L}_{\xi_H}(\theta)+1)(\mathcal{L}_{\xi_W}(\theta)+1)+\mathcal{L}_{\xi_H}(\theta)\mathcal{L}_{\xi_W}(\theta),
\end{equation*}
and define $\widetilde{F}_{HW}(\theta)$ and $F_{HW}^{(0)}(\theta)$ similarly, using $\widetilde{\mathcal{L}}_{\xi_H}(\theta)$ and $\widetilde{\mathcal{L}}_{\xi_W}(\theta)$ (recall~\eqref{wtildeHapprox} and~\eqref{wtildeWapprox}) for
$\widetilde{F}_{HW}(\theta)$, and $\mathcal{L}_{\xi_H}^{(0)}(\theta)$ and $\mathcal{L}_{\xi_W}^{(0)}(\theta)$ (recall~\eqref{LxiH0} and~\eqref{LxiW0}) for $F_{HW}^{(0)}(\theta)$.
Then, recalling~\eqref{HWrdefeq}, the real-time growth rate $r$, and its approximations $\tilde{r}$ and $r^{(0)}$ under the two approximate models, are given by
the unique real solutions of
\begin{equation}
\label{FHWroot}
F_{HW}(r)=1, \quad \widetilde{F}_{HW}(\tilde{r})=1 \quad \mbox{and}\quad F_{HW}^{(0)}(r^{(0)})=1.
\end{equation}
Note that $\mathcal{M}_\gtrv(0)\!=\!1$, $\mathcal{L}_{\xi_H}(0)\!=\!\widetilde{\mathcal{L}}_{\xi_H}(0)\!=\! \mathcal{L}_{\xi_H}^{(0)}(0)\!=\!\mu_H$ and $\mathcal{L}_{\xi_W}(0)\!=\!\widetilde{\mathcal{L}}_{\xi_W}(0)\!=\! \mathcal{L}_{\xi_W}^{(0)}(0)\!=\!\mu_W$, so, recalling~\eqref{hwthreshpar},
\begin{equation*}
F_{HW}(0)=\widetilde{F}_{HW}(0)=F_{HW}^{(0)}(0)=f_{HW}(\mu_G,\mu_H,\mu_W).
\end{equation*}
Thus, using part \thp{a} and its proof,
\begin{align*}
R_0=1 \iff R_*=1 &\iff f_{HW}(\mu_G,\mu_H,\mu_W)=1\\
&\iff F_{HW}(0)=\widetilde{F}_{HW}(0)=F_{HW}^{(0)}(0)=1.
\end{align*}
Hence, using~\eqref{FHWroot},
\begin{equation*}
R_*=1 \iff r=\tilde{r}=r^{(0)}=0,
\end{equation*}
and part \thp{a} of Theorem~\ref{propos} holds for the households-workplaces model, since $\mathcal{M}_\gtrv(0)=1$.

Turning to part \thp{b} of Theorem~\ref{propos}, applied to the households-workplaces model, suppose first that all households have size $n$.  Then an analogous argument to the derivation of~\eqref{Lbetatheta} yields, using the same notation,
\begin{equation*}
\mathcal{L}_{\xi_H}(\theta)=\sum_{i=1}^{n-1} \bbE\left[ \mathrm{e}^{-\theta T_i} \right]=\sum_{k=1}^{n-1} \mu_k^H   \bbE[{\rm e}^{-\theta T_1}|\chi_k(1)].
\end{equation*}
Arguing as for the households model then gives that $\mathcal{L}_{\xi_H}(\theta) \ge \mathcal{L}_{\xi_H}^{(0)}(\theta)$ if $\theta>0$, while if $\theta<0$, then $\mathcal{L}_{\xi_H}(\theta) \le \mathcal{L}_{\xi_H}^{(0)}(\theta)$ , and further that these inequalities hold also in the unequal household size setting, which, together with analogous inequalities for $\mathcal{L}_{\xi_W}(\theta)$ and $\mathcal{L}_{\xi_W}^{(0)}(\theta)$, imply that $F_{HW}(\theta) \ge F_{HW}^{(0)}(\theta)$ if $\theta>0$ and $F_{HW}(\theta) \le F_{HW}^{(0)}(\theta)$ if $\theta<0$.  Using~\eqref{FHWroot}, it follows that $0<r^{(0)}\le r$ in a growing epidemic and $r \le r^{(0)}<0$ in a declining epidemic. Part \thp{b} of Theorem~\ref{propos} now follows for the households-workplaces model, since $\mathcal{M}_\gtrv(\theta)$ strictly decreases with $\theta$.  The proof of part \thp{c} of Theorem~\ref{propos} for the households-workplaces model is omitted since it exploits arguments used in the corresponding proof for the households model in exactly the same way as is done above for part \thp{b}.
\end{proof}

\section{Conclusions}
\label{sec:conclusions}

In this paper, we focus on an SIR model for a directly transmissible infection spreading in a fully susceptible population, socially structured into households, or households and workplaces. However, most of our results extend readily to SEIR models. We collect together most of the reproduction numbers that have been defined in the literature (see Tables \ref{tab:RsH} and \ref{tab:RsHW}) and we show how they relate to each other. 
Particular emphasis is placed on the basic reproduction number $R_0$, for which we provide a simpler and more elegant method for its calculation than that introduced in the companion paper of this \cite{PelBalTra2012}.  
Extending the work of Goldstein et al.~\cite{GoldsteinEtal2009}, we add other reproduction numbers (namely $\hat{R}_2$, $R_I$ and $R_0$) to the ones they already discuss, and we provide new definitions for the reproduction numbers $R_{H\!I}$ and $R_2$ in a way that is more satisfactory when households have variable size: see \eqref{gHI} and \eqref{g2unequal}, respectively.

We extend the inequalities discussed in Goldstein et al.~\cite{GoldsteinEtal2009} (Table \ref{tab:Inequalities}) 
and, by doing so, we provide significantly sharper bounds for the vaccine-associated reproduction number $R_V$ than previously available, a result that holds consistently also for the network-households and households-workplaces models. 


More precisely, Goldstein et al.~\cite{GoldsteinEtal2009} proved that $R_\ast$, $R_r$ , $R_V$ , $R_{V\!L}$ and (if all households have the same size) $\bar{R}_{H\!I}$ share the same threshold at 1 and  $R_\ast \ge R_{V\!L} \ge R_V \ge \bar{R}_{H\!I}$ in a growing epidemic. They noted also that in most cases $R_r$ fits into the inequalities as
\[ R_\ast \geq R_{V\!L} \geq R_r \geq R_V \geq \bar{R}_{H\!I}.\]
Although $R_r$ may sometimes represent a practically useful upper bound for $R_V$, which is usually the quantity of interest for public health purposes, it can be excessively large at times (e.g.\ see Figure \ref{HSIRRr}) and it requires knowledge of the generation-time distribution $\gtpdf$. 
In general, $R_{V\!L}$ cannot be computed easily from the model parameters, so this leaves $R_\ast$ and $\bar{R}_{H\!I}$ as the only generally valid, time-independent and easy-to-calculate (from the basic model parameters) bounds for $R_V$, but $R_\ast$ is often excessively large and $\bar{R}_{H\!I}$ is not a threshold parameter when households are not all of the same size. 

Although a proof is still to be found, we conjecture that $R_0 \ge R_2$ in a growing epidemic, with the opposite inequality holding in a declining epidemic. Assuming this to be true, we have that, in a growing epidemic
\[ R_\ast \geq R_I \geq R_V \geq R_0 \geq R_2 \geq R_{H\!I} > 1 \]
and, in a declining epidemic,
\[ R_\ast \leq R_I \leq R_0 \leq R_2 \leq R_{H\!I} < 1. \]

Note that, even if the conjecture about $R_2$ does not hold, $R_I$ and $R_0$ provide sharper bounds for $R_V$ than $R_\ast$ and $R_{H\!I}$. Moreover, the numerical illustrations in Section \ref{sec:numerical} demonstrate that the improvement can be appreciable. This provides useful information for bracketing the critical vaccination coverage within an interval which does not depend on the fine details of the person-to-person contact process and is therefore robust to poor estimates of complex model components, such as the generation-time distribution. 

Turning to the spread in real-time, $R_r$ cannot always be related with $R_0$, although for virtually all models considered in the literature (including the standard SIR model and models with a deterministic time-varying infectivity profile) we have $R_r \ge R_0$ in a growing epidemic. Further, we have shown that $R_r$ and $R_{V\!L}$ cannot be ordered in general. 
A further reproduction number $\widetilde{R}_r$ has also been introduced, which in the case of the standard SIR model (and extensions to non-constant infection rates),  approximates $R_0$ better than $R_r$.

Other models with a different social structure have also been studied. These models all share the same qualitative construction of $R_0$ as presented in our previous paper \cite{PelBalTra2012}. As far as the network-households model is concerned, the relationships between $R_\ast, R_I, R_V$ and $R_0$ are the same as in the households model, although inequalities involving $R_0, R_2$ and $R_{H\!I}$ are more complex. Also, for the model with households and workplaces, the relationships between $R_\ast, R_I, R_V$ and $R_0$ are the same, with the additional presence of the household and workplace reproduction numbers as in Theorem \ref{Theorem 5.2}. 

Although our results stress how $R_V$ is bracketed between $R_I$ and $R_0$, we still believe that each of the reproduction numbers discussed have their own merit: $R_\ast$ carries a simple interpretation, is easiest to calculate, and in the households model gives the critical vaccination coverage when whole households are vaccinated uniformly at random; $R_V$ and $R_{V\!L}$ are important for practical reasons; $R_r$ is useful as $r$ can generally be estimated in new epidemics; $R_0$ represents a fundamental concept in epidemic models; $R_2$ is usually very close to $R_0$, but requires less knowledge about the epidemic model to be computed (in addition to $\mu_G$, only $\mu_1$ and the mean size of within-household epidemics) and might be easier to estimate from households studies, especially when within-household generations quickly overlap; and $R_{H\!I}$ requires even less knowledge than $R_2$, but is a coarser bound for $R_V$. For the households-workplaces model, in addition to the simple construction of $R_\ast$ and the bounds that $R_I$ and $R_0$ provide for $R_V$, $R_H$ still carries a simple interpretation, unlike $R_0$ is always finite and is informative about the control effort required when vaccination target entire households \cite{PelFerFra2009}.

Finally, much effort has been placed in studying the properties of $R_r$, in particular in relationship with $R_0$. 
As already mentioned above, it is not possible to order them in general, although in most of the models commonly considered in the literature $R_r \ge R_0$ in a growing epidemic (see Theorem \ref{propos}). However, 
on a more speculative but practically relevant note, consider the case of a non-random infectivity profile $\gtpdf(t)$ and assume that $\gtpdf$ is unimodal, with small variance and mean significantly larger than 0. Then, if instead of the true generations we deal with the computationally much more tractable rank generations (
approximating process $T_k$ by $\hat{T}_k$ in Section \ref{subsec:HRrproof}, c.f.\ equation \eqref{Tgbound})
the errors involved are small because generations do not easily overlap, in particular for realistically small household sizes. Furthermore, if we approximate the relative time at which real infections are made with that at which infections contacts are made (approximating $\hat{T}_k$ by $\sum_{m=0}^{k-1} \gtrvz{m}$, c.f.\ equation \eqref{casecThat}), the errors are also minor, because repeated infectious contacts between the same pair of individuals are likely to be all gathered around the mode of $\gtpdf$. Therefore, the quantitative values of $R_r$ and $R_0$ are very similar to each other, thus suggesting that $R_0$, the individual generation time distribution $\gtpdf$ and the real-time growth rate $r$ are approximately related as in the case of simple homogeneously mixing models. Given that many infections lead to infectivity profiles of the type described above (e.g.~influenza and SARS), this intuitive argument increases confidence in the estimates of $R_0$ obtained in the literature using models  that ignore the household structure.

\section*{Acknowledgement}
F.B. was supported in part by the UK Engineering and Physical Sciences Research Council (EPSRC: Grant No. EP/ E038670/1), L.P. by the Medical Research Council Methodology Program and the EPSRC and P.T. from the Swedish Vetenskapsrådet (Grant nr: 20105873). We also gratefully acknowledge support from the Isaac Newton Institute for Mathematical Sciences, Cambridge, where we held Visiting Fellowships under the Infectious Disease Dynamics programme and its follow-up meeting, during which part of the work was conducted.
We thank the reviewers for their constructive comments which have improved the presentation of the paper.

\appendix

\section{Comparison of $R_0$ and $R_2$}
\label{app:R0R2comp}

In this appendix we discuss the conjecture concerning the comparison of $R_0$ and $R_2$.  Recall that $R_0$ is given by the unique root in $(0,\infty)$ of the function $g_0$ defined at~\eqref{g0H} and $R_2$ is given by the unique root in $(b,\infty)$ of the function $g_2$ defined at~\eqref{g2unequal}.  Recall also that
$g_0(\lambda)=\sum_{n=1}^{n_H}\pi_n g_0^{(n)} (\lambda )$ and $g_2(\lambda)=\sum_{n=1}^{n_H}\pi_n g_2^{(n)} (\lambda )$, where $g_0^{(n)}$ and
$g_2^{(n)}$ are defined at~\eqref{gnH} and~\eqref{gn2}, respectively.

Observe that
$g_0^{(n)}=g_2^{(n)}$ for $n=1,2$, so $R_0=R_2$ if $n_H \le 2$.  We aim to show that if $n \ge 3$ then $\sign(g_2^{(n)}(\lambda)-g_0^{(n)}(\lambda))=\sign(\lambda-1)$ for all $\lambda>b^{(n)}$.  This implies that
if $n_H \ge 3$ then $\sign(g_2(\lambda)-g_0(\lambda))=\sign(\lambda-1)$ for all $\lambda>b$.  It would then follow, as in the proof of the comparison between $R_0$ and $R_{H\!I}$ in Theorem~\ref{Hcomp}, that if $R_0>1$ then $R_0>R_2$ and if $R_0<1$ then $R_0<R_2$.

We now fix $n\ge3$ and suppose that all households have size $n$, so $g_0^{(n)}=g_0$ and $g_2^{(n)}=g_2$.  Then, for any
$\lambda>b=1-\mu_1/\mu_H$,
\begin{align}
\label{g0-g2}
g_2(\lambda)-g_0(\lambda) &=\mu_G\left[\sum_{k=1}^{n-1}\frac{\mu_k}{\lambda^{k+1}}-\frac{\mu_1}{\lambda(\lambda-b)}\right]\nonumber\\
&=\frac{\mu_G}{\lambda-b}\left[\sum_{k=1}^{n-1} \frac{(\lambda-b)\mu_k}{\lambda^{k+1}}
-\frac{\mu_1}{\lambda}\right]\nonumber\\
&=\frac{\mu_G}{\lambda-b}\left[\sum_{k=2}^{n-1}\frac{\mu_k}{\lambda^k}-b\sum_{k=1}^{n-1}\frac{\mu_k}{\lambda^{k+1}}\right]\nonumber\\
&=\frac{\mu_G}{\lambda^2(\lambda-b)}f_{n}(\lambda^{-1}),
\end{align}
where $f_{n}$ is the polynomial of degree $n-2$ given by $f_n(x)=\sum_{j=0}^{n-2}c_j x^j$, with $c_j=\mu_{j+2}-b\mu_{j+1}$ $(j=0,1,\cdots,n-3)$ and $c_{n-2}=-b\mu_{n-1}$.
Note that $f_n(1)=0$.  Thus, as at~\eqref{fnfntilde},
\begin{equation}
f_n(x)=(x-1)\tilde{f}_n(x),
\end{equation}
where
\begin{equation}
\label{fntildepower1}
\tilde{f}_n(x)=\sum_{j=0}^{n-3}\tilde{c}_j x^j,
\end{equation}
with $\tilde{c}_j=\sum_{l=j+1}^{n-2}c_l$ $(j=0,1,\cdots,n-3)$.  Thus, 
$$\tilde{c}_j=\sum_{l=j+3}^{n-1} \mu_l - b \sum_{l=j+2}^{n-1} \mu_l \qquad (j=0,1,\cdots,n-4)$$ and $\tilde{c}_{n-3}=-b\mu_{n-1}$.  It follows from~\eqref{g0-g2} and~\eqref{fntildepower1} that $\sign(g_2(\lambda)-g_0(\lambda))=\sign(\lambda-1)$ if
$\tilde{c}_j \le 0$ ($j=0,1,\cdots,n-3)$ and at least one of these inequalities is strict.\\  Now  $\tilde{c}_{n-3}=-b\mu_{n-1}<0$\footnote{This assumes that $\mu_{n-1} > 0$, which is the case for most models studied in the literature. If $\mu_{n-1} = 0$ then the ensuing discussion may be modified in the obvious fashion.}, so a sufficient condition for $\sign(g_2(\lambda)-g_0(\lambda))=\sign(\lambda-1)$, and hence for the conjectured comparisons between $R_0$ and $R_2$, is that
\begin{equation*}
\frac{\sum_{l=j+3}^{n-1} \mu_l}{\sum_{l=j+2}^{n-1}\mu_l} \le b \quad (j=0,1,\cdots,n-4)
\end{equation*}
or equivalently that
\begin{equation}
\label{R0vR2conj}
\frac{1}{\mu_j} \sum_{l=j+1}^{n-1} \mu_l \le \frac{1}{\mu_1} \sum_{l=2}^{n-1} \mu_l \quad (j=2,3,\cdots,n-2).
\end{equation}

When $n=3$, the condition~\eqref{R0vR2conj} is vacuous, so the conjectured comparison between $R_0$ and $R_2$ holds when $n_H=3$.  We do not have a proof that~\eqref{R0vR2conj} holds in general.

Suppose that household generation sizes are the same as for a Reed-Frost model, as is the case when individuals have a non-random infectivity profile.  Let $p$ be the probability that a given infective infects a given susceptible household member.  Suppose that $n=4$.  The mean generation sizes are obtained easily using probabilities of different chains of infection (see e.g.~Bailey~\cite{Bailey1975}, Table 14.3) and are given by
\[
\mu_1=3p,\quad \mu_2=3p^2(1-p)(2-p^2)\quad \mbox{ and }\quad \mu_3=6p^3(1-p)^3.
\]
Thus,
\[
\frac{\mu_3}{\mu_2}=\frac{2p(1-p)^2}{2-p^2}\quad \mbox{ and }\quad
\frac{\mu_2+\mu_3}{\mu_1}=p(1-p)(2-p^2+2p(1-p)^2),
\]
whence, for $p \in [0,1]$,
\begin{align*}
\frac{\mu_3}{\mu_2}\le\frac{\mu_2+\mu_3}{\mu_1} &\iff \frac{2(1-p)}{2-p^2} \le 2-p^2+2p(1-p)^2\\
&\iff 2+6p(1-p)^2-4p^3+5p^4-2p^5 \ge 0.
\end{align*}
Let $\varphi(p)=4p^3-5p^4+2p^5$.  Then $\varphi'(p)=2p^2(1+5(1-p)^2)\ge0$, so $0\le \varphi(p)\le1$ for $p\in[0,1]$, whence $\frac{\mu_3}{\mu_2}\le\frac{\mu_2+\mu_3}{\mu_1}$ and~\eqref{R0vR2conj} holds for $n=4$, proving the conjectured comparison between
$R_0$ and $R_2$ in this case.  The expressions for the mean generation sizes become increasing unwieldy as $n$ increases.  However, numerical investigation using the recursive method for computing the mean generation sizes described in Appendix A of Pellis et al.~\cite{PelBalTra2012}
did not find any violation of~\eqref{R0vR2conj} for $n=5,6,\cdots,20$ and $p=0.001,0.002,\cdots,0.999$, suggesting that the conjectured comparison between $R_0$ and $R_2$ holds generally for the Reed-Frost model.  A similar investigation for the Markov model in which a typical infective contacts any given household member at rate $\lambda_H$ during an infectious period that has an $\mathrm{Exp}(1)$ distribution did not find any violation of~\eqref{R0vR2conj} for $n=4,5,\cdots,20$ and $\lambda_H=0.01,0.02,\cdots,10.00$, suggesting that the comparison between $R_0$ and $R_2$ holds generally for this model too.

\section{Comparison of $R_I$ and $R_{V\!L}$}
\label{app:RIRVLcomp}

In this appendix we give an example which demonstrates that, for the households model, the reproduction numbers $R_I$ and $R_{V\!L}$ cannot in general be ordered.  We consider an SIR epidemic among a population of households, all of which have size $3$.  The infectious period is assumed to
be constant and equal to one.  The individual to individual local infection rate is $\lambda_H$.
Suppose that the entire population is vaccinated with a leaky vaccine having efficacy $\mathcal{E}$ and let
$R_I(\mathcal{E})$ denote the post-vaccination version of $R_I$.  Then $R_I(\mathcal{E})$ is the largest eigenvalue of
\[
M_I(\mathcal{E}) = \left[ \begin{array}{cc}\mu_G(\mathcal{E})&\mu_H(\mathcal{E})\\\mu_G(\mathcal{E})&0\end{array}\right] ,
\]
where $\mu_G(\mathcal{E})=(1-\mathcal{E})\mu_G$ and $\mu_H(\mathcal{E})$ is the mean size of a single-household epidemic, with one initial infective and two initial susceptibles, both of whom are vaccinated. We assume throughout this appendix that $R_I>1$. By considering the
characteristic polynomial of $M_I(\mathcal{E})$ and arguing as for the comparison of $R_I$ and $R_V$ in the proof of Theorem~\ref{Hcomp}, it is readily seen that $\sign(R_{V\!L}-R_I)=\sign(\mu_H(\mathcal{E})-(1-\mathcal{E})\mu_H)$, with $\mathcal{E}=1-R_I^{-1}$.

The definition of the leaky vaccine action implies that $\mu_H(\mathcal{E})$ is given by mean size of the
single-household epidemic with $\lambda_H$ replaced by $\lambda_H(1-\mathcal{E})$.  Direct calculation shows that for the present population
\[
\mu_H(\mathcal{E})=2-4{\rm e}^{-2\lambda_H(1-\mathcal{E})}+2{\rm e}^{-3\lambda_H(1-\mathcal{E})}.
\]
Thus, if we let $x=1-\mathcal{E}$, then $$\sign(\mu_H(\mathcal{E})-(1-\mathcal{E})\mu_H)=\sign(u_1(x)-u_2(x)),$$
where $$u_1(x)=2-4{\rm e}^{-2\lambda_Hx}+2{\rm e}^{-3\lambda_Hx} \mbox{ and } u_2(x)=x(2-4{\rm e}^{-2\lambda_H}+2{\rm e}^{-3\lambda_H}).$$
Now $$u_1'(x)=8\lambda_H{\rm e}^{-2\lambda_Hx}-6\lambda_H{\rm e}^{-3\lambda_Hx}\mbox{ and }u_2'(x)=2-4{\rm e}^{-2\lambda_H}+2{\rm e}^{-3\lambda_H},$$ so, since $u_1(0)=u_2(0)=0$, $u_1(x)-u_2(x)>0$ (respectively $<0$) for all sufficiently small $x>0$ if $v(\lambda_H)>0$ (respectively $<0$), where $$v(\lambda_H)=2\lambda_H-2+4{\rm e}^{-2\lambda_H}-2{\rm e}^{-3\lambda_H}.$$    Further, $$v'(\lambda_H)=2-8{\rm e}^{-2\lambda_H}+6{\rm e}^{-3\lambda_H} \mbox{ and }v''(\lambda_H)=16{\rm e}^{-2\lambda_H}-18{\rm e}^{-3\lambda_H}.$$  Thus $\sign(v''(\lambda_H))=\sign(\lambda_H-\log(9/8))$  and elementary calculus shows that $v(\lambda_H)=0$ has a unique solution, $\lambda_H^*$ say, in $(0,\infty)$ and that $$\sign(v(\lambda_H))=\sign(\lambda_H-\lambda_H^*) \qquad \mbox{for $\lambda_H \in (0,\infty)$}.$$  (Numerical calculation shows that $\lambda_H^*\!\approx\!0.4219.$)  It follows that if $\lambda_H\!\in\!(0,\lambda_H^*)$ then $u_1(x)<u_2(x)$ for all sufficiently small $x>0$, whilst if $\lambda_H \in (\lambda_H^*, \infty)$ then $u_1(x)>u_2(x)$ for all sufficiently small $x>0$.
Recall that $x=1-\mathcal{E}$, where $\mathcal{E} = 1-R_I^{-1}$, so $x=R_I^{-1}$. Also, note that $R_I$ increases with $\mu_G$ for fixed $\lambda_H$. Hence, if $\lambda_H\in (0,\lambda_H^*)$, then there exists $\mu_G^*>0$ such that $R_I>R_{V\!L}$ for all $\mu_G>\mu_G^*$, whilst if $\lambda_H\in (\lambda_H^*,\infty)$, then there exists $\mu_G^\dagger>0$ such that $R_I<R_{V\!L}$ for all $\mu_G>\mu_G^\dagger$.
Thus the reproduction numbers $R_I$ and $R_{V\!L}$ cannot in general be ordered.

\section{Comparisons between $R_0$, $R_{H\!I}$, $\hat{R}_{H\!I}$ and $\bar{R}_{H\!I}$}
\label{app:R0RHIcomp}

In this appendix, we discuss the orderings of the individual reproduction numbers $R_0, R_{H\!I},\hat{R}_{H\!I}$ and $\bar{R}_{H\!I}$.

Recall from Section~\ref{subsec:RHI} that $a^{(n)} = \mu_H^{(n)} / (1 +\mu_H^{(n)})$ and that $R_{H\!I}$ is the unique positive solution of $g_{H\!I}$ (defined at \eqref{gHI}) in $(a,\infty)$, with $a=\max\left(a^{(n)},n=1,2,\cdots n_H\right)$.
Further, with $\hat{a} = \mu_H / (1+\mu_H)$, it follows from~\eqref{rhiunequal} that
$\hat{R}_{H\!I}$ is the unique root in $(\hat{a},\infty)$ of
\begin{equation*}
\hat{g}_{H\!I} (\lambda) = 1-\frac{\mu_G}{\lambda-\hat{a}} \qquad (\lambda > \hat{a}),
\end{equation*}
and, with $\bar{a}=\sum_{n=1}^{n_H}{\pi_n a^{(n)}}$, it follows from~\eqref{rhihatnew} that $\bar{R}_{H\!I}$
is the unique root in $(\bar{a},\infty)$ of
\begin{equation*}
\label{gbarHI}
\bar{g}_{H\!I}(\lambda) = 1 - \frac{\mu_G}{\lambda-\bar{a}} \qquad (\lambda > \bar{a}).
\end{equation*}
First, we compare $R_{H\!I}$ and $\bar{R}_{H\!I}$. For
fixed $\lambda>0$, define the function $\bar{f}_{\lambda}$ by
\begin{equation*}
\bar{f}_{\lambda}(x)=1-\frac{\mu_G}{\lambda - x} \qquad (x<\lambda).
\end{equation*}
Now $\bar{f}_{\lambda}''(x)=-2\mu_G(\lambda-x)^{-3}<0$ for $x < \lambda$, so $\bar{f}_{\lambda}$ is concave and by Jensen's inequality,
\begin{equation*}
g_{H\!I}(\lambda)=\sum_{n=1}^{n_H} \pi_n \bar{f}_{\lambda}(a^{(n)})
\le \bar{f}_{\lambda}(\bar{a})=\bar{g}_{H\!I}(\lambda)
\qquad (\lambda>a),
\end{equation*}
whence $\bar{R}_{H\!I} \le R_{H\!I}$ always, with equality only if $\mu_H^{(n)}$ is constant for all $n$ with $\pi_n > 0$. (Note that, as $\mu_H^{(1)}=0$, if $\pi_1 > 0$ this condition is satisfied only in the trivial cases where $n_H=1$ or there is no transmission within the household.) In particular, in a growing epidemic, Theorem~\ref{Hcomp}\thp{b} leads to
\begin{equation*}
R_0  > R_{H\!I} \ge \bar{R}_{H\!I},
\end{equation*}
but note that $\bar{R}_{H\!I}$, which is not a threshold parameter, might be smaller than 1 while the other two are larger than 1. Exploiting this fact, we now prove that the inequality between $R_0$ and $\bar{R}_{H\!I}$ is not necessarily reversed in a declining epidemic. Suppose that $n_H = 2$ and further that  $\pi_1 > 0, \pi_2 > 0$ and $\mu_H^{(2)} > \mu_H^{(1)}=0$. Then, there exists $\mu_G>0$ such that $\bar{R}_{H\!I} < R_{H\!I} = R_0 = 1$. Now, $\bar{R}_{H\!I}, R_{H\!I}$ and $R_0$ all depend continuously on $\mu_G$, so
reducing $\mu_G$ slightly gives a declining epidemic for which $\bar{R}_{H\!I}<R_0$. However, for a common household size, $\bar{R}_{H\!I}=R_{H\!I}$ and Theorem~\ref{Hcomp}\thp{b} implies that $\bar{R}_{H\!I}>R_0$ in a declining epidemic.
Thus, in general,  $\bar{R}_{H\!I}$ and $R_0$ cannot be ordered in a declining epidemic.

{
Now we compare $R_{H\!I}$ and $\hat{R}_{H\!I}$. First, note that
\begin{align*}
g_{H\!I}(\lambda) &= 1-\mu_G \sum_{n=1}^{n_H} \frac{\pi_n}{\lambda-a^{(n)}} \\
&= 1-\mu_G \sum_{n=1}^{n_H} \pi_n \frac{1}{\lambda-\frac{\mu_H^{(n)}}{1+\mu_H^{(n)}}} \\
&= 1-\mu_G \sum_{n=1}^{n_H} \pi_n \frac{1+\mu_H^{(n)}}{\lambda + (\lambda-1)\mu_H^{(n)}}.
\end{align*}
For fixed $\lambda>0$, define the function $\hat{f}_\lambda$ by
\begin{equation*}
\hat{f}_{\lambda}(x)=\frac{1+x}{\lambda +(\lambda-1) x} \qquad (x\neq \lambda/(1-\lambda)),
\end{equation*}
so $\hat{f}_\lambda''(x) = -2(\lambda-1) / (\lambda+(\lambda-1)x)^3$. Thus, if $\lambda>1$ then $\hat{f}_\lambda$ is strictly concave on $[0,\infty)$ so, using Jensen's inequality,
\begin{align*}
g_{H\!I}(\lambda) &= 1-\mu_G \sum_{n=1}^{n_H} \pi_n \hat{f}_\lambda(\mu_H^{(n)}) \\
&\ge 1-\mu_G \hat{f}_\lambda(\mu_H) = \hat{g}_{H\!I}(\lambda),
\end{align*}
and it follows that $R_{H\!I}\le\hat{R}_{H\!I}$ in a growing epidemic.

Suppose that $\lambda<1$. Then $\hat{f}_\lambda$ is strictly convex on $[\,0,\lambda/(1-\lambda)\,)$. Recall from \eqref{gHI} that $g_{H\!I}(\lambda)$ is defined for $\lambda> a =\max\pp{a^{(n)}:n=1,2,\cdots,n_H}$, where $a^{(n)} = \mu_H^{(n)} / \pp{1 + \mu_H^{(n)}}$. Thus, for each $n$, $\mu_H^{(n)} <\lambda/(1-\lambda)$ and applying Jensen's inequality as above yields that $g_{H\!I}(\lambda) \le \hat{g}_{H\!I}(\lambda)$ $(\lambda>a)$. It follows that $R_{H\!I}\ge\hat{R}_{H\!I}$ in a declining epidemic. 
These inequalities are strict except again in the case when $\mu_H^{(n)}$ is constant for all $n$ with $\pi_n > 0$. 
%

Finally, we give an example which demonstrates that the
reproduction numbers $R_0$ and $\hat{R}_{H\!I}$ cannot in general be ordered if the population contains households of different sizes.
We use notation analogous to that in the comparison between $R_0$ and $R_{H\!I}$ in the proof of Theorem~\ref{Hcomp}. Suppose that the population contains only households of sizes $1$ and $3$, so $\pi_3=1-\pi_1$.
Then $\mu_0=1, \mu_1=\pi_3 \mu_1^{(3)}$ and $\mu_2=\pi_3 \mu_2^{(3)}$, so, using \eqref{rhihatnew}, $\hat{R}_{H\!I}=\mu_G+\hat{a}$, with $\hat{a}=\pi_3(\mu_1^{(3)}+\mu_2^{(3)})/(1+\pi_3(\mu_1^{(3)}+\mu_2^{(3)}))$.
Let $\hat{h}_0(\lambda)=\hat{g}_{H\!I}(\lambda)-g_0(\lambda)$. Then, arguing as in (\ref{hlambda}) to (\ref{ctilde}), shows that
\begin{align}
\label{hbar}
\hat{h}_0(\lambda)&=-\frac{\mu_G}{\lambda(\lambda-\hat{a})}(\lambda^{-1}-1)\tilde{f}(\lambda^{-1})\nonumber\\
&=\frac{\mu_G(\lambda-1)}{\lambda^2(\lambda-\hat{a})}\tilde{f}(\lambda^{-1}),
\end{align}
where
\begin{equation}
\label{f3tilde}
\tilde{f}(x)=\hat{a}(\mu_1+\mu_2)-\mu_2+\hat{a}\mu_2 x.
\end{equation}
(In the notation of (\ref{ctilde}), it is easily shown that
$\tilde{c}_0=\hat{a}(\mu_1+\mu_2)-\mu_2$ and $\tilde{c}_1=\hat{a}\mu_2$.)
Substituting the above expressions for $\mu_1, \mu_2$ and $\hat{a}$ into (\ref{f3tilde}) yields that, for $x>0$,
\begin{equation}
\label{ftildegt0}
\tilde{f}(x)<0 \iff \mu_2^{(3)}>\frac{\pi_3(\mu_1^{(3)}+\mu_2^{(3)})}{1+\pi_3(\mu_1^{(3))}+\mu_2^{(3)})}
\left[\mu_1^{(3)}+\mu_2^{(3)}(1+x)\right].
\end{equation}
Thus, for any $x>0$, $\tilde{f}(x)<0$ for all sufficiently small $\pi_3 \in (0,1)$.

Recall that $R_0$ is the unique root in $(0,\infty)$ of $g_0$.  Suppose that $R_0>1$.  Then, since
$g_0(R_0)=0$, it follows from (\ref{hbar}) and (\ref{ftildegt0}) that, if $\pi_3 \in (0,1)$ is sufficiently small, then $\hat{g}_{H\!I}(R_0)<0$, whence $\hat{R}_{H\!I}>R_0$.
A similar argument shows that, if $R_0<1$ and $\pi_3 \in (0,1)$ is sufficiently small, then $\hat{R}_{H\!I}<R_0$.
Note that these inequalities are again the reverse of those proved for $R_{H\!I}$, and hence of those for the situation when all households have the same size, in which
$R_{H\!I}, \hat{R}_{H\!I}$ and $R_{H\!I}$ coincide.  Thus, $R_0$ and $\hat{R}_{H\!I}$ cannot in general be ordered.


\section{Random infectivity profile}
\label{app:RandomTVI}

The proof of Theorem~\ref{propos} in Section~\ref{subsec:HRrproof} already reveals that in order to show  that 
$R_r \ge R_0$ (respectively $R_r \le R_0$) 
does not generally hold in growing (respectively declining) epidemics, we should look for a model in which $\gtrvzt{m}$ is not stochastically smaller than $\gtrvt$. In particular this is the case if an individual with a large total infectivity makes its contacts relatively early after infection, while a an individual with a small total infectivity makes its contacts, long after infection.
Here we provide  a simple example in a household of size $n=3$. As before the individuals are denoted by $i=0$ (the initial infective in the household), $i=1$ and $i=2$ (the initial susceptibles).

Consider a random infectivity profile which either has its complete mass $\kappa_a$ at time $t_a$ or it has its complete mass $\kappa_b$ at time $t_b$, both with probability $1/2$. Thus, for $x>0$,
\begin{equation}
\int_0^x \mathcal{I}(t) {\rm d}t = \begin{cases}
\kappa_a \indic(x \geq t_a) & \text{with probability } 1/2 \\
\kappa_b \indic(x \geq t_b) & \text{with probability } 1/2 \\
\end{cases}.
\end{equation}
Here all parameters are non-negative. We assume further that
\begin{equation}\label{negcorH}
(\kappa_a-\kappa_b)(t_a-t_b)<0.
\end{equation}

Note that $\kappa_a + \kappa_b=2$, since the infectivity profile necessarily satisfies $\int_0^{\infty} \mathbb{E}[\mathcal{I}(t)] dt=1$. Furthermore,
from the definition of $\gtrv$, we have that
\begin{equation}
\label{goodmgf}
\mathcal{M}_\gtrv(r) = \frac{\kappa_a  \text{e}^{-r t_a} + \kappa_b \text{e}^{-r t_b}}{2}= \frac{\kappa_a  \text{e}^{-r t_a} + \kappa_b \text{e}^{-r t_b}}{\kappa_a + \kappa_b}.
\end{equation}

Define $p_a = 1-{\rm e}^{-\lambda_H \kappa_a}$  and $p_b = 1-{\rm e}^{-\lambda_H \kappa_b}$,
and note that an infective individual that has total infectivity $\kappa_a$ (respectively $\kappa_b$) infects each susceptible independently with probability $p_a$  (respectively $p_b$) and all infections occur at time $t_a$  (respectively $t_b$). 
Hence,
\begin{equation}
\label{MW_2ind}
\mathcal{M}_\gtrvt(r) = \frac{p_a  \text{e}^{-r t_a} + p_b \text{e}^{-r t_b}}{ p_a + p_b  }.
\end{equation}
Straightforward algebra gives
\begin{equation}
\label{differenceMGF}
\mathcal{M}_\gtrv(r) - \mathcal{M}_\gtrvt(r) = \frac{(p_b \kappa_a -p_a \kappa_b)  (\text{e}^{-r t_a} - \text{e}^{-r t_b})}{2 (p_a + p_b)}.
\end{equation}
Furthermore, since for $\lambda_H,r>0$, the function $(1-{\rm e}^{-r\lambda_H x})/x$ is decreasing for $x>0$, and $p_b \kappa_a -p_a \kappa_b > 0$ if and only if $p_b/\kappa_b > p_a/\kappa_a$, we obtain that $p_b \kappa_a -p_a \kappa_b > 0$ if and only if $\kappa_a > \kappa_b$. Combining this with \eqref{negcorH} we obtain that
$\mathcal{M}_\gtrv(r) - \mathcal{M}_\gtrvt(r) >0$ for $r>0$. If $R_*>1$, this implies that $R_0>\widetilde{R}_r$ (cf.~the end of the proof of Theorem~\ref{propos}\thp{b}), while if $R_*<1$ (and therefore $r<0$), then $R_0<\widetilde{R}_r$.

If we prove further that under the given conditions $\widetilde{R}_r >R_r$ in growing epidemics and $\widetilde{R}_r <R_r$ in declining epidemics, then we have constructed the desired counter example.
Therefore, in what follows, we only compare the constructions from which $\widetilde{R}_r$  and $R_r$ are deduced. In particular, we have to compare
$\mathcal{L}_{\beta_H}(\theta)$ and $\widetilde{\mathcal{L}}_{\beta_H}(\theta)$, through which
$\tilde{r}$ and $r$ are defined by~\eqref{Lbetardef}. In turn, $\tilde{r}$ and $r$ are used to define $\widetilde{R}_r$ and $R_r$ by \eqref{RrsfromM}.

Recall that $R_r=1/\mathcal{M}_\gtrv(r)$, where $r$ is the unique real value of $\theta$
which solves the first equation in~\eqref{Lbetardef}, and 
$\widetilde{R}_r=1/\mathcal{M}_\gtrv(\tilde{r})$, where $\tilde{r}$ is the unique real value of $\theta$ which solves the second equation in~\eqref{Lbetardef}.  Thus, cf.~\eqref{Lbetatheta} and~\eqref{rtgrowth3}, we need to compare the pair of times $(T_1, T_2)$ with the pair of times $(\sum_{l=0}^{k_1-1}\gtrvzt{l,1},\sum_{l=0}^{k_2-1}\gtrvzt{l,2})$, where $k_1$ and $k_2$ are the generation numbers of individual $1$ and $2$, while the $\gtrvzt{l,1}$s and $\gtrvzt{l,2}$s are independent random variables all distributed as $\gtrvt$. Note that we assume that individual $0$ is infected at time 0. For ease of reference, define
$(\widetilde{T}_0,\widetilde{T}_1,\widetilde{T}_2) = (0,\sum_{l=0}^{k_1-1}\gtrvzt{l,1},\sum_{l=0}^{k_2-1}\gtrvzt{l,2})$.
We now  compare
\begin{equation}
\label{comp13a}
\sum\limits_{i=0}^{2} \bbE\b{\text{e}^{-r T_i}} = \bbE\b{\text{e}^{-r T_{0}} +\text{e}^{-r T_{1}} + \text{e}^{-r T_{2}}}
\end{equation}
and
\begin{equation}
\label{comp13b}
\sum\limits_{i=0}^{2} \bbE\b{\text{e}^{-r \widetilde{T}_i}} = \bbE\b{\text{e}^{-r \widetilde{T}_{0}} +\text{e}^{-r \widetilde{T}_{1}} + \text{e}^{-r \widetilde{T}_{2}}}
\end{equation}
Since $T_0 = \widetilde{T}_0$ we only have to compare
\begin{equation}
\label{comp13}
\zeta(r) = \bbE\b{\text{e}^{-r T_{1}} + \text{e}^{-r T_{2}}} \qquad \mbox{and} \qquad
\widetilde{\zeta}(r) = \bbE\b{\text{e}^{-r \widetilde{T}_{1}} + \text{e}^{-r \widetilde{T}_{2}}}.
\end{equation}

After some algebra, we obtain
\begin{equation}
\label{zeta1}
\zeta(r) =
\left( p_a  \text{e}^{-r t_a} + p_b \text{e}^{-r t_b}\right) \left(1+ \frac{ p_a (1-p_a) \text{e}^{-rt_a} +  p_b (1-p_b) \text{e}^{-rt_b}}{2}\right).
\end{equation}

In order to compute $\widetilde{\zeta}(r)$, we compute explicitly the average number of cases in each generation, viz.
\begin{align}
\label{mus_randomTVIexample}
\mu_1 & = p_a +p_b,\\
\mu_2 & = (p_a +  p_b)(p_a (1-p_a) +  p_b (1-p_b))/2.
\end{align}
From the definition of $\widetilde{\zeta}(r)$ and from \eqref{MW_2ind} we deduce
\begin{equation}
\label{zeta3}
\widetilde{\zeta}(r) = \mu_1 \mathcal{M}_\gtrvt(r) + \mu_2 \left( \mathcal{M}_\gtrvt(r) \right)^2.
\end{equation}
Straightforward but lengthy algebra shows that
\begin{equation}
\label{analytic_case}
 \zeta(r) - \widetilde{\zeta}(r) \; =
\frac{p_a p_b \left( p_a \mathrm{e}^{-r t_a} +  p_b \mathrm{e}^{-r t_b} \right)}{2( p_a +  p_b)} \left( p_b - p_a \right) \left( \mathrm{e}^{-r t_a} - \mathrm{e}^{-r t_b} \right),
\end{equation}
from which we conclude that $\zeta(r)<\widetilde{\zeta}(r)$, and therefore $R_r<\widetilde{R}_r$, when $$ \left( p_b - p_a \right) \left( \mathrm{e}^{-r t_a} - \mathrm{e}^{-r t_b} \right)<0,$$ which we assumed for growing epidemics in \eqref{negcorH} since $p_a$ is increasing in $\kappa_a$.
This completes the counter example for growing epidemics.

If $R_\ast<1$ (i.e.\ $r<0$), the same counter example in this case leads to $R_r\ge \widetilde{R}_r> R_0$, showing that also in a declining epidemic $R_0$ and $R_r$ cannot be ordered in general.

\section{Comparison of $R_r$ and $R_{V\!L}$}
\label{app:RrRVLcomp}

In this appendix we give examples which demonstrate that, for the households model, $R_r$ and $R_{V\!L}$ cannot in general be ordered.
We consider an SIR epidemic among a population of households, all of which have size $2$.  The infectious periods of infectives are independent, each distributed according to a random variable $T_I$ having mean $1$ and moment-generating function $\mathcal{M}_{T_I}(\theta)=\mathbb{E}[{\rm e}^{-\theta T_I}]$.  Whilst infectious, the initial infective in a household contacts locally his/her other household member at the points of a Poisson process having rate $\lambda_H$.

Note that, since $\int_0^\infty \mathbb{P}(T_I \ge t) {\rm d}t=\mathbb{E}[T_I]=1$, the mean infectivity profile of an infective is $\gtpdf(t)=\mathbb{P}(T_I \ge t)$ $(t \ge 0)$, whence $$\mathcal{M}_\gtrv(\theta)=(1-\mathcal{M}_{T_I}(\theta))/\theta.$$ Consider a single household epidemic, label the initial infective $0$ and the other household member $1$.  Let $T_I$ denote individual $0$'s infectious period and $X$ denote the time of the first local infectious contact of individual $1$ by individual $0$.  Thus, $X \sim \mathrm{Exp}(\lambda_H)$ and $0$ infects $1$ locally if and only if $X < T_I$.  Hence, the probability that $0$ infects $1$ locally is $$\mathbb{E}_{T_I}[\mathbb{P}(X<I)]=\mathbb{E}_{T_I}[1-{\rm e}^{-\lambda_H T_I}]=1-\mathcal{M}_\gtrv(\lambda_H),$$ which is also the household mean generation size $\mu_1$.  Note that, in the notation of Section~\ref{subsec:Rr}, $\gtrvt$ is distributed as $(X|X<T_I)$, whence $$\mathcal{M}_{\gtrvt}(\theta)=\frac{\lambda_H}{\lambda_H+\theta} \frac{1-\mathcal{M}_\gtrv(\lambda_H+\theta)}{1-\mathcal{M}_\gtrv(\lambda_H)}.$$  Further, since all households have size $2$, $\mathcal{L}_{\beta_H}(r) = \widetilde{\mathcal{L}}_{\beta_H}(r)$ (see the observation at the end of Section~\ref{subsec:Rr}) and, using~\eqref{rtgrowth4}, the real-time growth rate $r$ is the unique solution of $F(r)=1$, where
\begin{equation*}
F(r)=\mu_G \mathcal{M}_\gtrv(r)\left(1+\mu_1 \mathcal{M}_\gtrvt(r)\right),
\end{equation*}
which, using the above expressions for $\mu_1$ and $\mathcal{M}_\gtrvt(\theta)$, may be written as
\begin{equation}
\label{FRr}
F(r)=\mu_G \mathcal{M}_\gtrv(r)\left(1+\lambda_H \mathcal{M}_\gtrv(\lambda_H+r)\right).
\end{equation}

Suppose that $r>0$, so $R_r>1$, and that the entire population is vaccinated with a leaky vaccine having efficacy $\mathcal{E}=1-R_r^{-1}$. Then, recalling~\eqref{Rrfromr}, $1-\mathcal{E}=\mathcal{M}_\gtrv(r)$.  Hence, after vaccination,
$\mu_G$ becomes $\mathcal{M}_\gtrv(r)\mu_G$ and $\lambda_H$ becomes $\mathcal{M}_\gtrv(r)\lambda_H$, so, using~\eqref{FRr}, the post-vaccination real-time growth rate, $r_\mathcal{E}$ say, is given by the unique real solution of $F_\mathcal{E}(r_\mathcal{E})=1$,
where
\begin{equation*}
F_\mathcal{E}(r_\mathcal{E})=\mu_G \mathcal{M}_\gtrv(r) \mathcal{M}_\gtrv(r_\mathcal{E})\left[1+\mathcal{M}_\gtrv(r)\lambda_H \mathcal{M}_\gtrv\left(\mathcal{M}_\gtrv(r)\lambda_H+r_\mathcal{E}\right)\right].
\end{equation*}
Now,
\begin{equation*}
F_\mathcal{E}(0)=\mu_G \mathcal{M}_\gtrv(r)\left[1+\mathcal{M}_\gtrv(r)\lambda_H \mathcal{M}_\gtrv\left(\mathcal{M}_\gtrv(r)\lambda_H\right)\right],
\end{equation*}
so, since $F(r)=1$, it follows from~\eqref{FRr} that 
$$\sign(F_\mathcal{E}(0)-1)=\sign(G(\lambda_H,r)),$$ where
\begin{equation*}
G(\lambda_H,r)=\mathcal{M}_\gtrv(r)\mathcal{M}_\gtrv\left(\mathcal{M}_\gtrv(r)\lambda_H\right)-\mathcal{M}_\gtrv(\lambda_H+r).
\end{equation*}

Note that (i) if $G(\lambda_H,r)=0$ then $r_\mathcal{E}=0$, so the post-vaccination epidemic is critical and $R_{V\!L}=R_r$; (ii) if $G(\lambda_H,r)>0$ then $r_\mathcal{E}>0$, so the post-vaccination epidemic is supercritical and $R_{V\!L}>R_r$; and (iii) if $G(\lambda_H,r)<0$ then $r_\mathcal{E}<0$, so the post-vaccination epidemic is subcritical and $R_{V\!L}<R_r$.

Suppose that $T_I \sim \mathrm{Exp}(1)$.  Then $\mathcal{M}_\gtrv(\theta)=(1+\theta)^{-1}$ and $G(\lambda_H,r)=0$, whence $R_{V\!L}=R_r$, as noted in Section~\ref{MarkovSIRHmod}.

Suppose that $T_I$ has probability density function
$
f_{T_I}(t)=2 t {\rm e}^{-2 t}$ $(t \ge 0),
$
i.e.~$T_I$ follows a gamma distribution with parameters $\alpha = \gamma = 2$ (see \eqref{gammapdf}). Then $\mathcal{M}_{T_I}(r)=\left(\frac{2}{2+r}\right)^2$, whence $\mathcal{M}_\gtrv(r)=\frac{4+r}{(2+r)^2}$.  Lengthy algebra then yields that
\begin{align*}
G(\lambda_H,r)&=\frac{\lambda_H r\left[(4+r)^2 \lambda_H+(2+r)^2(8+r)\right]}
{\left[(2+\lambda_H+r)\left(2(2+r)^2+\lambda_H(4+r)\right)\right]^2} >0,
\end{align*}
since $r>0$ and $\lambda_H>0$.  Thus, $R_{V\!L}>R_r$.

Suppose instead that $T_I$ has probability density function
\begin{equation*}
f_{T_I}(t)=\frac{1}{3}{\rm e}^{-\frac{2}{3}t}+{\rm e}^{-2t} \quad (t \ge 0),
\end{equation*}
so $T_I$ is an equally weighted mixture of $\mathrm{Exp}(\frac{2}{3})$ and $\mathrm{Exp}(2)$.  Then
\begin{equation*}
\mathcal{M}_{T_I}(r)=\frac{1}{2+3r}+\frac{1}{2+r}
\qquad \mbox{and} \qquad
\mathcal{M}_\gtrv(r)=\frac{1}{2}\left(\frac{3}{2+3r}+\frac{1}{2+r}\right).
\end{equation*}
Lengthy algebra now yields that
\begin{equation*}
G(\lambda_H,r)=-\frac{\lambda_H r (3r+8)\left[3(4+3r)\lambda_H+(2+3r)(6+3r)\right]}
{H(\lambda_H,r)
\left(2+\lambda_H+r\right)\left(2+3(\lambda_H+r)\right)},
\end{equation*}
where
\[
H(\lambda_H,r)=
\left[2(2+r)(2+3r)+\lambda_H(4+3r)\right]\left[2(2+r)(2+3r)+3\lambda_H(4+3r)\right].
\]
Hence $G(\lambda_H,r)<0$, since
$r>0$ and $\lambda_H>0$, so now $R_{V\!L}<R_r$.  Thus $R_r$ and $R_{V\!L}$ cannot in general be ordered.

It is difficult to make general statements since there is no simple expression for $G(\lambda_H,r)$. However, note that in the above examples, $R_{V\!L}>R_r$ when the infectious period distribution is less variable than an exponential distribution and $R_{V\!L}<R_r$ when it is more variable.

\section{Infinitely long latent periods}
\label{app:inflonglat}

We consider first a Markov SEIR households epidemic model, in which the latent and infectious periods follow exponential distributions with rates $\delta$ and $1$, respectively, and whilst infectious a typical infective makes global contacts at overall rate $\mu_G$ and contacts any given susceptible in his/her household at rate $\lambda_H$.  We study the limit of $R_r$ as $\delta^{-1} \to \infty$, so the latent periods become infinitely long with all other parameters held fixed, and compare that limit with $R_0$ and $R_V$, which are both independent of $\delta$. We restrict attention to a growing epidemic, i.e.~when $R_0>1$.

For fixed $\delta$, we may linearly rescale time by setting $t'=\delta t$ so that in the rescaled process the latent period is exponentially distributed with mean one.  In the limit as $\delta^{-1} \to \infty$, in the rescaled process the infectious period of an infective is reduced to a single point in time.  Note that the exponential-growth associated reproduction number $R_r$ and the mean generation sizes $\mu_1,\mu_2,\cdots, \mu_{n_H-1}$ are each invariant to this rescaling ($\mu_1,\mu_2,\cdots, \mu_{n_H-1}$ are also invariant to $\delta$).  A similar rescaling is used in the proof of Lemma B.3.1~in Goldstein et al.~\cite{GoldsteinEtal2009} but the argument presented there assumes a constant latent period and hence does not apply to the Markov SEIR model.

Consider the limit of the rescaled process as $\delta^{-1} \to \infty$.  In this process any infective makes all of his/her infectious contacts at the same time, i.e.~at the end of his/her latent period, and the latent periods of distinct infectives are independent $\mathrm{Exp}(1)$ random variables. It follows that the infectious contact interval $\gtrv \sim \mathrm{Exp}(1)$, whence $\mathcal{M}_\gtrv(r)=(1+r)^{-1}$.
Suppose that all households have size $n$ and that $n \le 3$.  Then in the epidemic graph $\mathcal{G}^{(n)}$ (see Section~\ref{Hmodel}), for $k=1,2$, if an individual, $i$ say, belongs to rank generation $k$, there is precisely one chain of directed edges from the initial infective to individual $i$ that has length $k$.  It follows that
the real-time growth rate $r$ satisfies $\mathcal{L}_{\beta_H}^{(0)}(r)=1$, where $\mathcal{L}_{\beta_H}^{(0)}(r)$ is defined at~\eqref{rtgrowth2}.  Note also that for this limiting process $\gtrvt\overset{D}{=}\gtrv$, where $\overset{D}{=}$ denotes equal in distribution.
It then follows that, for the limiting process, $R_r=\widetilde{R}_r=R_0$.  This conclusion holds also in the case of unequal household sizes, provided the maximum household size $n_H$ is at most $3$.

Suppose now that all households have size $n=4$.  Then in the epidemic graph $\mathcal{G}^{(n)}$, when $k=1,3$, it is still true that for any individual, $i$ say, belonging to rank generation $k$, there is precisely one chain of directed edges from the initial infective to individual $i$ that has length $k$, but when $k=2$ that is no longer the case.  Recall that the individuals in $\mathcal{G}^{(4)}$ are labelled $0,1,2,3$, where $0$ is the initial infective.  Suppose that individual $0$ contacts individuals $1$ and $2$, but not individual $3$, and that both individuals $1$ and $2$ contact individual $3$.  Thus there are two distinct paths of length $2$ from individual $0$ to individual $3$.
For $i=0,1,2,3$, let $T_E^{(i)}$ be the latent period of individual $i$ (in the limiting rescaled process), so $T_E^{(i)} \sim \mathrm{Exp}(1)$.  Then individuals $1$ and $2$ are both infected at time $T_E^{(0)}$ and, for $i=1,2$, individual $i$ attempts to infect individual $3$ at time $T_E^{(0)}+T_E^{(i)}$, so individual $3$ is infected at time $T_E^{(0)}+T'_E$, where $T'_E=\min(T_E^{(1)},T_E^{(2)})$.
Note that $T_E^{(0)}$ and $T'_E$ are independent, and $T'_E \sim \mathrm{Exp}(2)$, so $\mathcal{M}_{T'_E}(r)=2/(2+r)$.

Let the mean generation sizes $\mu_0,\mu_1,\mu_2$ and $\mu_3$ be defined as previously but now write $\mu_2=\mu_{21}+\mu_{22}$, where, $\mu_{2j}$ is the mean number of generation $2$ infectives that have precisely $j$ paths of length $2$ from the initial infective to them. Then,
\begin{eqnarray}
\label{lbetahinflat}
\mathcal{L}_{\beta_H}(r)=\mu_G \mathcal{M}_\gtrv(r) \left[ \phantom{_{_|}} 1+\mu_1 \mathcal{M}_\gtrv(r)+\mu_{21} (\mathcal{M}_\gtrv(r))^2 \right. & \nonumber \\
+ \;\left. \mu_{22} \mathcal{M}_\gtrv(r)\mathcal{M}_{T'_E}(r)+\mu_3(\mathcal{M}_\gtrv(r))^3 \phantom{_{_|}} \right].
\end{eqnarray}
Now $\mathcal{M}_{T'_E}(\theta)\ge \mathcal{M}_\gtrv(\theta)$, for $\theta>0$, whence $r>r^{(0)}$, where $r^{(0)}$ solves $\mathcal{L}_{\beta_H}^{(0)}(r)=1$, and it follows that $R_r>R_0$.  Again, $\widetilde{R}_r=R_0$, since $\gtrvt\overset{D}{=}\gtrv$.  Similar arguments show that $R_r>\widetilde{R}_r=R_0$ for any population with $n_H \ge 4$.

We now compare $R_r$ with $R_V$.  Recall from Theorem~\ref{Hcomp} that $R_0=R_V$ when $n_H \le 3$, whence $R_r=R_V$.  Thus suppose that all households have size $4$ and that a fraction $1-R_r^{-1}$ of the population is vaccinated with a perfect vaccine.  Then, using~\eqref{Rrfromr}, the probability that a given individual is not vaccinated is $R_r^{-1}=\mathcal{M}_\gtrv(r)$.  Hence, after vaccination, $\mu_G$ is reduced to $\mu_G^V=\mathcal{M}_\gtrv(r)\mu_G$ and, for $k=1,3$, $\mu_k$ is reduced to $\mu_k^V=(\mathcal{M}_\gtrv(r))^k\mu_k$, as
prior to vaccination any individual in generation $1$ or $3$ has precisely one chain of the appropriate length linking them to the initial infective.  Consider the situation described above, in which, prior to vaccination, individual $0$ contacts individuals $1$ and $2$, but not individual $3$, and that both individuals $1$ and $2$ contact individual $3$.  Individual $3$ still has two chains linking them to the initial infective after vaccination if and only if individuals $1,2$ and $3$ are not vaccinated, which happens with probability $(\mathcal{M}_\gtrv(r))^3$.  (Note that individual $0$ is assumed to be unvaccinated, as $\mu_G$ is reduced to $\mu_G^V$.)  Thus, $\mu_{22}^V=(\mathcal{M}_\gtrv(r))^3\mu_{22}$.  In the above situation, individual $3$ has precisely one chain linking them to individual $0$ after vaccination if and only if individual $3$ and exactly one of individuals $1$ and $2$ are not vaccinated, which occurs with probability $2(\mathcal{M}_\gtrv(r))^2(1-\mathcal{M}_\gtrv(r))$.  It follows that $\mu_{21}^V= (\mathcal{M}_\gtrv(r))^2\mu_{21}+2(\mathcal{M}_\gtrv(r))^2(1-\mathcal{M}_\gtrv(r))\mu_{22}$.  Hence, after vaccination, the growth rate, $r_V$ say, satisfies $\mathcal{L}_{\beta_H}^V(r_V)=1$, where
\begin{eqnarray*}
\mathcal{L}_{\beta_H}^V(r_V) & = & \mu_G \mathcal{M}_\gtrv(r_V) \, \left[\phantom{_{_|}} 1+\mu_1^V \mathcal{M}_\gtrv(r_V)+\mu_{21}^V(\mathcal{M}_\gtrv(r_V))^2 \right. \\
& &\left. +\;\mu_{22}^V \mathcal{M}_\gtrv(r_V)\mathcal{M}_{T'_E}(r_V)+\mu_3(\mathcal{M}_\gtrv(r_V))^3\;\right].
\end{eqnarray*}

Now $\mathcal{M}_\gtrv(0)=\mathcal{M}_{T'_E}(0)=1$ and $\mathcal{L}_{\beta_H}(r)=1$, so
\begin{align}
\mathcal{L}_{\beta_H}^V(0)&=\mu_G^V(1+\mu_1^V+\mu_{21}^V+\mu_{22}^V+\mu_3^V)\nonumber\\
&=\mu_G \mathcal{M}_\gtrv(r)\,\left[ \phantom{_{_|}} 1+\mu_1 \mathcal{M}_\gtrv(r)+\mu_{21}(\mathcal{M}_\gtrv(r))^2\; \right. \nonumber\\
&\quad+\left. \mu_{22}\left(2(\mathcal{M}_\gtrv(r))^2-(\mathcal{M}_\gtrv(r))^3\right)+\mu_3(\mathcal{M}_\gtrv(r))^3 \phantom{_{_|}}\right]\nonumber\\
&=\mathcal{L}_{\beta_H^{(1)}}(r)+\mu_G (\mathcal{M}_\gtrv(r))^2\mu_{22}\left[2\mathcal{M}_\gtrv(r)-(\mathcal{M}_\gtrv(r))^2-\mathcal{M}_{T'_E}(r)\right]\nonumber\\
&=1+\mu_G (\mathcal{M}_\gtrv(r))^2\mu_{22}\frac{r}{(1+r)^2(2+r)}\nonumber\\
&>1,
\label{LVbetaH1(0)}
\end{align}
since $r>0$. Hence $r_V>0$, since $\mathcal{L}_{\beta_H}^V$ is a decreasing function and $\mathcal{L}_{\beta_H}^V(r_V)=1$.  Thus vaccinating a fraction $1-R_r^{-1}$ of the population is insufficient to prevent a major outbreak, so $R_r<R_V$.  Numerical evidence suggests that the same conclusion holds whenever $n_H \ge 4$.  However, analytical progress is more difficult when $n_H>4$ since then in the limiting rescaled process it is no longer the case that an individual's rank and true generations necessarily coincide.

Consider now the households-workplaces version of the above Markov SEIR model.  Thus, whilst infectious, a typical infective makes global contacts at overall rate $\mu_G$, contacts any given susceptible in his/her household at rate $\lambda_H$ and any given susceptible in his/her workplace at rate $\lambda_W$.  We use the same rescaling as in the households model and study the limit of the rescaled process as $\delta^{-1} \to \infty$. As previously, we restrict attention to a growing epidemic. Then, as at~\eqref{HWrdefeq}, the real-time growth rate $r$ is given by the unique real solution of $F(r)=1$, where
\begin{equation*}
F(r)=\mu_G \mathcal{M}_\gtrv(r) (\mathcal{L}_{\beta_H}(r)+1)(\mathcal{L}_{\beta_W}(r)+1)+\mathcal{L}_{\beta_H}(r)\mathcal{L}_{\beta_W}(r),
\end{equation*}
and $\mathcal{M}_\gtrv(r)=(1+r)^{-1}$.  The same arguments as used for the households model yield that, for $r \ge 0$,
\begin{equation*}
\mathcal{L}_{\beta_H}(r)\overset{3}{\geq} \mathcal{L}_{\beta_H^{(2)}}(r) = \mathcal{L}_{\beta_H^{(3)}}(r) \quad \mbox{and} \quad \mathcal{L}_{\beta_W}(r)\overset{3}{\geq} \mathcal{L}_{\beta_W^{(2)}}(r) = \mathcal{L}_{\beta_W^{(3)}}(r),
\end{equation*}
whence
\begin{equation}
\label{RrR0HWcomplong}
R_r \overset{3}{\geq} \widetilde{R}_r = R_0.
\end{equation}

Turn now to the comparison of $R_r$ and $R_V$. From Theorem~\ref{Theorem 5.2} and~\eqref{RrR0HWcomplong}, $R_r=R_V$ when $n_H \le 3$ and $n_W \le 3$.  Suppose that
$n_H=n_W=4$ and that a fraction $1-R_r^{-1}$ ($=1-\mathcal{M}_\gtrv(r)$, where $r$ is the real-time growth rate without vaccination) of the population is vaccinated with a perfect vaccine.
Prior to vaccination,
\begin{equation*}
\mathcal{L}_{\beta_H}(r)=\mu_1^H \mathcal{M}_\gtrv(r)+\mu_{21}^H (\mathcal{M}_\gtrv(r))^2+\mu_{22}^H \mathcal{M}_\gtrv(r) \mathcal{M}_{T'_E}(r)+\mu_3^H (\mathcal{M}_\gtrv(r))^3,
\end{equation*}
where $\mu_1^H, \mu_{21}^H, \mu_{22}^H$ and $\mu_3^H$ are the mean generation sizes for a typical single-household epidemic ($\mu_2^H$ is decomposed into $\mu_{21}^H+\mu_{22}^H$ as above~\eqref{lbetahinflat}), and $\mathcal{L}_{\beta_W}(r)$ is given by the same formula with $\mu_1^H, \mu_{21}^H, \mu_{22}^H$ and $\mu_3^H$ replaced by $\mu_1^W, \mu_{21}^W, \mu_{22}^W$ and $\mu_3^W$, respectively.
After vaccination, the real-time growth rate, $r_V$ say, is given by the unique real solution of $F^V(r_V)=1$, where
\begin{equation*}
F^V(r_V)=\mu_G^V \mathcal{M}_\gtrv(r_V) (\mathcal{L}_{\beta_H}^V(r_V)+1)(\mathcal{L}_{\beta_W}^V(r_V)+1)+\mathcal{L}_{\beta_H}^V(r_V)\mathcal{L}_{\beta_W}^V(r_V),
\end{equation*}
where $\mu_G^V=\mathcal{M}_\gtrv(r) \mu_G$ and, for example,
\begin{align*}
\mathcal{L}_{\beta_H}^V(r_V)=&\mu_1^{HV} \mathcal{M}_\gtrv(r_V)+\mu_{21}^{HV} (\mathcal{M}_\gtrv(r_V))^2+\mu_{22}^{HV} \mathcal{M}_\gtrv(r_V) \mathcal{M}_{T'_E}(r_V)\\
&+\mu_3^{HV} (\mathcal{M}_\gtrv(r_V))^3,
\end{align*}
with
\begin{eqnarray*}
\mu_1^{HV} & = & \mathcal{M}_\gtrv(r)\mu_1^H,\\
\mu_{21}^{HV}& = & (\mathcal{M}_\gtrv(r))^2\left(\mu_{21}^H+2(1-\mathcal{M}_\gtrv(r))\mu_{21}^H\right), \\
\mu_{22}^{HV} & = & (\mathcal{M}_\gtrv(r))^3\mu_{22}^{H},\\
\mu_3^{HV} & = & (\mathcal{M}_\gtrv(r))^3\mu_3^H.
\end{eqnarray*}
Now,
\begin{align*}
\mathcal{L}_{\beta_H}^V(0)&=\mu_1^{HV}+\mu_{21}^{HV}+\mu_{22}^{HV}+\mu_3^{HV}\\
&=\mathcal{M}_\gtrv(r)\mu_1^H+(\mathcal{M}_\gtrv(r))^2\mu_{21}^H+(\mathcal{M}_\gtrv(r))^2(2-\mathcal{M}_\gtrv(r))\mu_{22}^H\\
&\quad+(\mathcal{M}_\gtrv(r))^3\mu_{3}^H\\
&=\mathcal{L}_{\beta_H}(r)+\mathcal{M}_\gtrv(r)\left[2\mathcal{M}_\gtrv(r)-(\mathcal{M}_\gtrv(r))^2-\mathcal{M}_{T'_E}(r)\right]\mu_{22}^H\\
&>\mathcal{L}_{\beta_H}(r)
\end{align*}
(cf.~\eqref{LVbetaH1(0)}) and, similarly, $\mathcal{L}_{\beta_W}^V(0)>\mathcal{L}_{\beta_W}^V(r)$.  It follows that $F^V(0)>F(r)=1$. Thus $r_V>0$, since $F^V$ is a decreasing function, whence $R_r<R_V$.  It is easily seen that the same conclusion holds if $n_H=4$ and $n_W \le 3$ or $n_W=4$ and $n_H \le 3$.  We conjecture that it also holds whenever $\max(n_H,n_W)\ge 4$.

\section{Estimating $r$ for households model with non-random infectivity profile}
\label{app:sellke}

In this appendix we describe the simulation-based method used in Section \ref{nonrandHM} for determining the real-time growth rate for a households model with a non-random infectivity profile $\mathcal{I}(t)=\gtpdf(t)$ $(t \ge 0)$, where $\int_0^\infty \gtpdf(t) {\rm d}t=1$.  For ease of exposition, we assume that all households have the same size $n$.  The extension to unequal household sizes is straightforward.  Recall that $t$ time units after he/she was infected, an infectious individual makes global contacts at overall rate $\mu_G \gtpdf(t)$ and, additionally, he/she contacts any given susceptible in his/her household at rate $\lambda_H \gtpdf(t)$, where we have suppressed the dependence of $\lambda_H$ on $n$.  For ease of exposition, we assume that all households have the same size $n$.  Our aim is to estimate the Laplace transform $\mathcal{L}_{\beta_H}(\theta)=\int_0^\infty \beta_H(t) {\rm e}^{-\theta t} {\rm d}t$ of the global infectivity profile of a household.

Consider $n_{sim}$ independent simulations of a single-household epidemic, with initially one infective and $n-1$ susceptibles, under the above disease dynamics.  For $s=1,2,\cdots,n_{sim}$, let $T_{s,0}=0$, $Z_s$ denote the size of the $s$th simulated epidemic, not counting the initial infective, and $T_{s,1},T_{s,2},\cdots,T_{s,Z_s}$ denote the corresponding infection times, assuming that the epidemic starts at time $t=0$.  Then the average global infectivity profile of the $n_{sim}$ epidemics is
\begin{equation*}
\hat{\beta}_H^{n_{sim}}(t)=\frac{1}{n_{sim}}\mu_G \sum_{s=1}^{n_{sim}}\sum_{i=0}^{Z_s} \gtpdf(t-T_{s,i}) \quad(t\ge 0),
\end{equation*}
where $\gtpdf(t)=0$ if $t<0$, whence an unbiased estimator of $\mathcal{L}_{\beta_H}(\theta)$ is
\begin{equation*}
\hat{\mathcal{L}}^{n_{sim}}_{\beta_H}(\theta)=\frac{1}{n_{sim}}\mu_G \sum_{s=1}^{n_{sim}}\sum_{i=0}^{Z_s} {\rm e}^{-\theta T_{s,i}} \mathcal{M}_\gtrv(\theta) \quad (\theta> \theta_0),
\end{equation*}
where $\theta_0=\inf\{\theta: \mathcal{M}_\gtrv(\theta) < \infty\}$.  The real-time growth rate is
estimated by solving $\hat{\mathcal{L}}^{n_{sim}}_{\beta_H}(r)=1$ numerically, yielding $\hat{r}_{n_{sim}}$ say.  Application of the strong law of large numbers yields that, for any $\theta > \theta_0$, $\hat{\mathcal{L}}^{n_{sim}}_{\beta_H}(\theta) \to \mathcal{L}_{\beta_H}(\theta)$ almost surely as $n_{sim} \to \infty$.  This may be strengthened to (c.f.~the Glivenko-Cantelli theorem)
$\max_{\theta>\theta_1}|\hat{\mathcal{L}}^{n_{sim}}_{\beta_H}(\theta)-\mathcal{L}_{\beta_H}(\theta)| \to 0$ almost surely as $n_{sim} \to \infty$, for any $\theta_1>\theta_0$. It follows that $\hat{r}_{n_{sim}} \to r$ almost surely as $n_{sim} \to \infty$.

Suppressing the suffix $s$, to simulate the size $Z$ and the corresponding infection times $T_1,T_2,\cdots,T_{n-1}$ of a single-household epidemic we use the following generalisation of the construction of Sellke~\cite{Sellke1983}.  Label the individuals in the household $0,1,\cdots,n-1$, where individual $0$ is the initial infective.  Let $Q_1,Q_2,\cdots,Q_{n-1}$ be independent and identically distributed $\mathrm{Exp}(\lambda_H)$ random variables.  The random variable $Q_i$ denotes individual $i$'s critical exposure to infection.  Let $Q_{(1)} \le Q_{(2)} \le \cdots \le Q_{(n-1)}$ be the random variables $Q_1,Q_2,\cdots,Q_{n-1}$ arranged in increasing order, i.e.~the order statistics of $Q_1,Q_2,\cdots,Q_{n-1}$. Note that, exploiting the lack-of-memory property of the exponential distribution, the random variables $Q_{(1)}, Q_{(2)}-Q_{(1)}, Q_{(3)}-Q_{(2)},\cdots, Q_{(n-1)}-Q_{(n-2)}$ are mutually independent, $Q_{(1)}\sim {\rm Exp}\left( (n-1)\lambda_H \right)$ and $Q_{(i)}-Q_{(i-1)}\sim {\rm Exp}\left( (n-i)\lambda_H \right) (i=2,3,\cdots,n-1)$.

The epidemic is constructed as follows.  The initial infective becomes infected at time $T_0=0$.  For $t \ge 0$, at time $t$, each individual accumulates exposure to infection from the initial infective at rate $\gtpdf(t)$.  For $i=1,2,\cdots,n-1$, individual $i$ becomes infected if and when his/her accumulated exposure to infection reaches $Q_i$.  Thus if $Q_{(1)}>1$ then no susceptible is infected in the epidemic (recall that $\int_0^\infty \gtpdf(t) {\rm d}t=1$).  Suppose that
$Q_{(1)}<1$ (note that $\mathbb{P}(Q_{(1)}=1)=0$ since $Q_{(1)}$ is a continuous random variable).  Then the first infection takes place at time $T_1$ given by $\int_0^{T_1} \gtpdf(t){\rm d}t= Q_{(1)}$.  For $t > T_1$, at time $t$, each remaining susceptible accumulates exposure to infection at rate $\gtpdf(t)$ from the initial infective and at rate $\gtpdf(t-T_1)$ from the individual who was infected at time $T_1$.  Thus, if $Q_{(2)}>2$, there is no further spread of infection, whilst if  $Q_{(2)}<2$ the next infection occurs at time $T_2$ satisfying  $\int_0^{T_2} \gtpdf(t)+ \gtpdf(t-T_1){\rm d}t= Q_{(2)}$.  The construction of the epidemic continues in the obvious fashion.  It is readily seen that $Z=\min(z:Q_{(z+1)} > z+1)$ and, for $i=1,2,\cdots,Z$, the $i$th infection time $T_i$ is given implicitly by
\begin{equation*}
\sum_{j=0}^{i-1} \int_{T_j}^{T_i} \gtpdf(t-T_j) {\rm d}t= Q_{(i)}.
\end{equation*}
Note that for the example in Section~\ref{nonrandHM} the infections times are easily simulated using MATLAB since the substitution $t'=\gamma t$ converts $\int_{T_j}^{T_i} \gtpdf(t-T_j) {\rm d}t$ into an incomplete gamma function.






%

%

\end{document}